\documentclass[10pt, a4paper]{article}
\usepackage{amsmath, amsthm, amssymb, mathtools}
\usepackage[svgnames]{xcolor}
\usepackage[left=1.5cm, right=1.5cm, top=1.5cm, bottom=2cm]{geometry}
\usepackage{hyperref}
\usepackage[notref, notcite, final]{showkeys}
\usepackage{bbm}
\usepackage{enumitem}
\usepackage{float}
\usepackage[left]{lineno}
\usepackage{lmodern}

\usepackage{caption}
\usepackage{graphicx}
\usepackage[backend=bibtex, sorting=nyt, style=authoryear, url=false, isbn=false, eprint=false, maxbibnames=4, giveninits=true]{biblatex}
\renewbibmacro{in:}{\ifentrytype{article}{}{\printtext{\bibstring{in}\intitlepunct}}}
\DeclareFieldFormat
[article,inbook,incollection,inproceedings,patent,thesis,unpublished]
{title}{#1\isdot}
\DeclareFieldFormat[article]{volume}{\mkbibbold{#1}}
\DeclareFieldFormat[article]{number}{\mkbibbold{#1}}

\usepackage{xr}

\allowdisplaybreaks

\newcommand\dd{\mathrm{d}}

\newcommand\R{\mathbb R}

\theoremstyle{plain}
\newtheorem{theorem}{Theorem}[section]
\newtheorem{lemma}[theorem]{Lemma}
\newtheorem{proposition}[theorem]{Proposition}
\newtheorem{corollary}[theorem]{Corollary}
\theoremstyle{definition}
\newtheorem{definition}[theorem]{Definition}
\newtheorem{remark}[theorem]{Remark}
\newtheorem{assumption}[theorem]{Assumption}
\newtheorem{claim}[theorem]{Claim}

\counterwithin{equation}{section}

\definecolor{QG}{named}{ForestGreen}
\definecolor{AD}{named}{red}
\newcommand{\writefoot}[1]{
	\renewcommand{\thefootnote}{}
	\footnotetext{\hspace{-16.5pt}\scriptsize#1}
	\renewcommand{\thefootnote}{\arabic{footnote}}
}

\author{Jean-Baptiste Burie, Arnaud Ducrot, Quentin Griette\footnote{Corresponding author. e-mail: \href{mailto:quentin.griette@univ-lehavre.fr}{\texttt{quentin.griette@univ-lehavre.fr}}}}

\begin{document}
\writefoot{\small \textbf{AMS subject classifications (2020).} 34D05, 92D25, 37L15, 37N25 \smallskip}
\writefoot{\small \textbf{Keywords.} ordinary differential equations, asymptotic behavior, population dynamics, infinite dynamical system.\smallskip}

\writefoot{\small \textbf{Acknowledgements:} 
J.-B.B. and A.D. are supported by the ANR project ArchiV ANR-18-CE32-0004. 
Q.G. was partially supported by ANR grant  ``Indyana'' number  ANR-21-CE40-0008. 
A.D. and Q.G. acknowledge the support of the Math AmSud program for project MATH-AMSUD-22-MATH-09. 
The authors thank Pierre Magal for stimulating discussions {and two anonymous reviewers whose comments contributed greatly to improve the quality of the article.}
}

\begin{center}
	\hypersetup{hidelinks}
	\vspace{1cm}
	\renewcommand{\thefootnote}{\fnsymbol{footnote}}
	\begin{minipage}{0.9\textwidth}
		\centering
		\LARGE{\bf Asymptotic behavior of an epidemic model with infinitely many variants}\bigskip
	\end{minipage}

	\Large
	Jean-Baptiste Burie$^{a}$, Arnaud Ducrot$^{b}$, and Quentin Griette$^{b,}$\footnote{Corresponding author. e-mail: \href{mailto:quentin.griette@univ-lehavre.fr}{\texttt{quentin.griette@univ-lehavre.fr}}}\medskip \\
	\medskip

	\today
	\bigskip

	\normalsize
	{\it $^a$ Institut de Math\'ematiques de Bordeaux, Universit\'e de Bordeaux, \\
	CNRS, IMB, UMR 5251, \\ 
	351, cours de la Lib\'eration, F-33400 Talence, France.}\\
	\medskip

	{\it $^b$ Normandie Univ, UNIHAVRE, LMAH, FR-CNRS-3335,\\ 
	ISCN, 76600 Le Havre, France. 
	}
	\hypersetup{hidelinks=false}
\end{center}
\bigskip

\begin{abstract}
    We investigate the long-time dynamics of a SIR epidemic model with infinitely many pathogen variants infecting a homogeneous host population. We show that the basic reproduction number $\mathcal{R}_0$ of the pathogen can be defined in that case and corresponds to a threshold between the persistence ($\mathcal{R}_0>1$) and the extinction ($\mathcal{R}_0\leq 1$) of the pathogen. When $\mathcal{R}_0>1$ and the maximal fitness is attained by at least one variant, we show that the systems reaches an equilibrium state that can be explicitly determined from the initial data. When $\mathcal{R}_0>1$ but none of the variants attain the maximal fitness, the situation is more intricate. 
    We show that, in general, the pathogen is uniformly persistent and any family of variants that have a fitness which is uniformly lower than the optimal fitness, eventually gets extinct. We derive a condition under which the total pathogen population converges to a limit which can be computed explicitly. 
    We also find counterexamples that show that, when our condition is not met, the total pathogen population may converge to an unexpected value, or the system can even reach an eternally transient behavior where the total pathogen population between several values. We illustrate our results with numerical simulations that emphasize the wide variety of possible dynamics. 
\end{abstract}
\medskip

\section{Introduction}
In this article we investigate the large time behavior of the SIR epidemic model
\begin{subequations}\label{eq:main}
	\begin{equation}\label{eq:main-a}
		\left\{\begin{aligned}\relax
			&\frac{\dd}{\dd t}S(t) 	= \Lambda -\theta S(t) - \sum_{k=0}^{+\infty}\beta_k S(t)I_k(t), & &t>0, \\ 
			&\frac{\dd}{\dd t}I_n(t)= \beta_nS(t)I_n(t) - \gamma_n I_n(t), & &  t>0,n\in\mathbb{N},
		\end{aligned}\right.
	\end{equation}
	with the initial data
	\begin{equation}\label{eq:main-b}
		S(0)=S_0\in(0, +\infty), \qquad (I_n(0))_{n\in\mathbb{N}}=(I_n^0)_{n\in\mathbb{N}}\in\ell^1_+,
	\end{equation}
	where $\ell^1_+$ denotes the space of non-negative summable sequences. 
\end{subequations}

This model  describes the evolution of a population of hosts that can be, at any time $t>0$, either free of infection and immunity ($S(t)$, the susceptible population), or infected by a pathogen of type $n\in\mathbb{N}$ ($I_n(t)$, the infected population of type $n$). The parameter $\Lambda>0$ models a constant influx of susceptible hosts, $\theta>0$ the death rate of the hosts in the absence of infection, $\beta_n$ the transmission parameter of the pathogen of type $n$, and $\gamma_n$ the recovery rate of a pathogen of type $n$. Both $\beta_n$ and $\gamma_n$ are bounded sequences.
SIR models are ubiquitous in the literature concerning mathematical epidemiology and have been extensively studied. Without the pretention of reconstructing the entire history of the model, let us cite the works of \textcite{Ker-McK-1927} that might well be its first occurrence in the literature, and was immediately applied to a plague outbreak in the island of Bombay.

In this article we consider that many variants of the pathogen (possibly inifinitely many) compete to infect susceptible hosts. All of the possible genotypes are listed in an infinite sequence indexed by $n\in\mathbb{N}$, and for each genotype we denote $(\beta_n, \gamma_n)$ the associated phenotype composed of the transmission coefficient $\beta_n$ and the recovery rate of the infection $\gamma_n$. We do not specify the particular mechanisms that link the underlying variable $n\in\mathbb{N}$ and the phenotype $(\beta_n, \gamma_n)$ but focus on the dynamics of the population under \eqref{eq:main} conditionally to the knowledge of these mechanisms.  We do not take mutations into account and consider that the pathogen is asexual; therefore, \eqref{eq:main} can be considered a pure competition model where the pathogens compete for a single resource (the susceptible hosts).

{When $I_n^0$ is zero except for a finite number of indices, 
our problem is reduced to a system of ordinary differential equations:
\begin{subequations}\label{eq:SI-finite}
	\begin{equation}
		\left\{\begin{aligned}\relax
			&\frac{\dd \phantom{t}}{\dd t}S(t) = \Lambda -\theta S(t) - S(t)\big(\beta_1 I_1(t)+\beta_2 I_2(t)+\ldots + \beta_n I_n(t)\big)\\ 
			&\frac{\dd \phantom{t}}{\dd t}I_1(t)=\beta_1 S(t)I_1(t) - \gamma_1I_1(t)  \\
			& \qquad\vdots \\ 
			&\frac{\dd \phantom{t}}{\dd t}I_N(t)=\beta_N S(t)I_N(t) - \gamma_NI_N(t), 
		\end{aligned}\right.
	\end{equation}
	with the initial data
	\begin{equation}
		S(0)=S_0\in(0, +\infty),\qquad   I_1(0)=I_1^0\in(0, +\infty), \qquad \ldots, \qquad I_N(0)=I_N^0\in(0, +\infty).
	\end{equation}
\end{subequations}
In this context, \textcite{Hsu-Hub-Wal-77, Hsu-78} showed for a similar system of ordinary differential equations that the solution eventually converges to an equilibrium which may not be unique but is always concentrated on the equations that maximize the fitness $\beta_n/\gamma_n$. Here we extend these results to an infinite-dimensional dynamical system (we consider infinitely many variants) and prove that, for some well-chosen coefficients, the system stays eternally in a transient state and never converges to a single equilibrium.  \textcite{Thi-11} considers a related model in which a continuous distribution of host classes is infected by a pathogen that can be transmitted across classes; he proves, among other results and under very general assumptions, the global stability of the endemic equilibrium. While his model is different in nature and in behavior from ours,  we consider it as an inspiration for future works. 
}


The problem of several species competing for a single resource has received a lot of attention in the literature.  In this context, the ``Competitive exclusion principle''  states that ``Complete competitors cannot coexist'', which means that given a number of species competing for the same resource in the same place, only one can survive in the long run. This idea was already present to some extent in the book of Darwin, and is sometimes referred to as Gause's law \parencite{Har-60}.
This problem of survival of competitors has attracted the attention of mathematicians since the '70s and many studies have proved this property in many different contexts -- let us mention the seminal works of 
\textcite{Hsu-Hub-Wal-77, Hsu-78}
followed by 
\textcite{Arm-McG-80, But-Wol-85, Wol-Lu-92, Hsu-Smi-Wal-96, Wol-Xia-97, Li-99}, 
to cite a few -- and also disproved in other contexts, for instance in fluctuating environments, see 
\textcite{Cus-80} and 
\textcite{Smi-81}. 
\textcite{Ack-All-03} study the competitive exclusion in an epidemic model with a finite number of strains, and describe how different species can coexist in some cases.

In our model, the fitness of a variant with genotype $n$ is given by the formula $\mathcal{R}^n_0 = \frac{\Lambda \beta_n}{\theta \gamma_n}=\frac{\Lambda}{\theta}\alpha_n$, where $\alpha_n:=\frac{\beta_n}{\gamma_n}$; the competitive exclusion principle implies that the only genotypes that {eventually remain} are the ones that maximize $\mathcal{R}_0^n$. We prove that this is correct - asymptotically, only the genotypes that maximize the fitness survive - but incomplete. Indeed it does not suffice to describe the asymptotic behavior of the population, especially when there are equality cases in the fitness of the variants (i.e. $\frac{\beta_{n_1}}{\gamma_{n_1}}=\frac{\beta_{n_2}}{\gamma_{n_2}}$ with $n_1\neq n_2$, and possibly $\gamma_1\neq \gamma_2$), of the maximal fitness is not reached by any genotype, or both. In the latter case we can even observe an alternation of the prevalent variant and eternal oscillations in the total number of infected, see the second example in section \ref{sec:bivalent-gamma}.
  Similar behaviors have been observed in the literature for related models (among others, \textcite{Hsu-78} already gives a similar description). 
For example in the epidemiological context of \textcite{DayGandon2007}, it has been observed that a  strain 1 with a higher value of $\gamma$ and a slightly lower $\mathcal R_0$ value than a strain 2 may nevertheless be dominant for some time. These borderline cases shed light on our understanding of transient dynamics, see also \textcite{BDDD2020} where {estimates for the transient dynamics} for a related evolutionary model are provided based on the {local flatness of the fitness function}.

{Quantitative traits} such as the virulence or the transmission rate of a pathogen, the life expectancy of an individual and more generally any observable feature such as height, weight, muscular mass, speed, size of legs, etc. are naturally represented using  continuous {or discrete} variables. Such a description of a population seems highly relevant and has been used mostly in modelling studies involving some kind of evolution \parencite{Mag-00, Mag-02, Bar-Per-07,Des-Mis-Jab-Rao-08,Bar-Mir-Per-09,  Bou-Cal-Meu-Mir-Per-Rao-12, Jab-Rao-11,Rao-11, Lor-Per-14, Gri-19, Ducrot-Magal-Liu-Griette}. In many of those models, mutation is considered as a process that is continuously occurring through time. That assumption may or may not be realistic depending on the context. It may also be realistic to model mutation  as discrete events in time; in that case, the behavior of the population between such mutation events is correctly described by pure competition equations like \eqref{eq:main}. {In the same spirit, our work also provides a precise description of what happens in the vanishing mutation limit, and the trajectories of a model with a small but non-zero mutation operator are expected to be close, at least transiently, to those of the limit.}

{The structure of the paper is as follows. In section \ref{sec:main} we present our main results. In section \ref{sec:examples} we provide some examples and numerical simulations that present different asymptotic behaviors. In section \ref{sec:discussion} we propose a discussion of our results. Finally in section \ref{sec:proof} we prove the results we claimed in section \ref{sec:main}.}

\paragraph{Data availability} Data sharing not applicable to this article as no datasets were analyzed during the current study.

\paragraph{Conflict of Interest.} The authors declare no conflict of interest.
\section{Main results}
\label{sec:main}

In this article we study the solutions of \eqref{eq:main-a} supplemented with the initial data \eqref{eq:main-b}. Before starting, let us precise that we use the notation $\ell^1$ to denote the Banach space of absolutely summable real sequences equipped with the norm
\begin{equation*}
	\Vert (a_n)_{n\in\mathbb{N}}\Vert_{\ell^1}:= \sum_{n=0}^{+\infty} |a_n|, 
\end{equation*}
and $\ell^1_+$ is the positive cone of $\ell^1$, that is to say the set of non-negative summable sequences. 
Similarly, $\ell^\infty$ denotes the Banach space of bounded sequences equipped with the norm
\begin{equation*}
	\Vert (a_n)_{n\in\mathbb{N}}\Vert_{\ell^\infty}:=\sup_{n\in\mathbb{N}} |a_n|.
\end{equation*}
We will make the following assumption on  the parameters arising in \eqref{eq:main}. 
\begin{assumption}\label{as:params}
	The constants $\Lambda>0$ and $\theta>0$ are given.
	The sequences $\big(\beta_n\big)_{n\in\mathbb{N}}\in\ell^\infty$ and $\big(\gamma_n\big)_{n\in\mathbb{N}}\in\ell^\infty$ are bounded and we assume that there exist constants $0<\beta_0<\beta^\infty$ and $0<\gamma_0<\gamma^\infty$ such that 
	\begin{equation*}  
		0< \beta_0\leq \beta_n\leq \beta^\infty\qquad  \text{ and }\qquad  0<\gamma_0\leq \gamma_n\leq \gamma^\infty,\qquad \text{ for all } n\in{\mathbb{N}}.
	\end{equation*}
	As a consequence of this assumption, the sequence $(\beta_n/\gamma_n)_{n \in \mathbb{N}}$ is bounded.  We let $\alpha^*$ be the maximal fitness defined by
	\begin{equation*}
		\alpha^*:=\sup_{k\in\mathbb{N}}\frac{\beta_k}{\gamma_k}. 
	\end{equation*}
\end{assumption}

{Our next assumptions ensures that the maximal fitness $\alpha^*$ is effectively attained (possibly at infinity) by a non-neglibile population of infected. In other words, our model \eqref{eq:main} is not equivalent to another model with a strictly lower maximal fitness.}
\begin{assumption}\label{as:init}  
	We let $S_0>0$, $(I_n^0)_{n\in\mathbb{N}}\in \ell^1_+$ be given and assume that there exists a sequence of indices $n_k\in\mathbb{N}$ with 
	\begin{equation*}
		I_{n_k}^0>0 \text{ and } \lim_{k\to+\infty} \frac{\beta_{n_k}}{\gamma_{n_k}}=\alpha^*.
	\end{equation*}
\end{assumption}
{Let us precise that the sequence $n_k$ in Assumption \ref{as:init} can be eventually stationary; in particular we do not assume that the set of positive components is infinite.}

We define the basic reproductive number  $\mathcal R_0$ by
\begin{equation}\label{def-alpha*}
	\mathcal R_0:=\frac{\Lambda}{\theta}\alpha^*.
\end{equation}

{Finally, we make a technical assumption to avoid unnecessary theoretical discussions.
\begin{assumption}[Finite $\omega$-limit sets]\label{as:omega}
	We assume that the $\omega$-limit sets of the sequences $\gamma_n$ and $\alpha_n:=\frac{\beta_n}{\gamma_n}$ , defined by
	\begin{equation*}
	    \omega\left((\gamma_n)_{n\in\mathbb{N}}\right):=\bigcap_{n\in\mathbb{N}} \overline{\{\gamma_k\,:\, k\geq n\}} \text{ and } \omega\left((\alpha_n)_{n\in\mathbb{N}}\right):=\bigcap_{n\in\mathbb{N}} \overline{\{\alpha_k\,:\, k\geq n\}} , 
	\end{equation*}
	are finite.
\end{assumption}
}

\subsection{The Cauchy problem: existence and uniqueness}
Our first result concerns the existence and uniqueness of the solution to the  Cauchy problem \eqref{eq:main}. We show that the system \eqref{eq:main} is well posed in the sense of Hadamard. 
\begin{proposition}[The Cauchy problem]\label{prop:Cauchy}
	Let Assumption \ref{as:params} hold, and  $S_0>0$ and $(I_n^0)_{n\in\mathbb{N}}\in \ell^1_+$ be given.  Then the system \eqref{eq:main} has a unique global mild solution $\big(S(t), (I_n(t))_{n\in\mathbb{N}}\big)\in C^0\big([0, +\infty), \mathbb{R}\big)\times C^0\big([0, +\infty), \ell^1_+\big)$, which is also a classical solution:
	\begin{equation*}
		S(t)\in C^1\big([0, +\infty), \mathbb{R}\big) \text{ and } (I_n(t))_{n\in\mathbb{N}}\in C^1\big([0, +\infty), \ell^1_+\big).
	\end{equation*}
	Moreover for all $t>0$ the map 
	\begin{equation*}
		\left(S_0, \big(I_n^0\big)_{n\in\mathbb{N}}\right)\in \mathbb{R}\times\ell^1_+ \longmapsto \left(S(t), \big(I_n(t)\big)_{n\in\mathbb{N}}\right) \in \mathbb{R}\times \ell^1_+
	\end{equation*}
	is continuous.
\end{proposition}
As an important consequence of Proposition \ref{prop:Cauchy}, we note that each component of $(I_n(t))$ can be computed from $S(t)$. More precisely, if we set
\begin{equation*}
	\overline{S}(t):= \frac{1}{t}\int_0^t S(s)\dd s, 
\end{equation*}
then we have the following formula for $I_n(t)$: 
\begin{equation}\label{eq:I_n}
    I_n(t) = e^{t\left(\beta_n \overline{S}(t) - \gamma_n\right)}, \qquad \text{ for all } n\in\mathbb{N}\text{ and } t\geq 0.
\end{equation}
We will use equation \eqref{eq:I_n} repeatedly in the proofs of our results.

\subsection{{Persistence and asymptotic behavior}}
\label{sec:main-general}

Next we investigate the asymptotic behavior of the solutions to \eqref{eq:main} when $t\to+\infty$. We first show that the population of pathogens gets extinct if {$\mathcal{R}_0\leq1$}.
\begin{proposition}[Extinction]\label{prop:Extinction}
    Let Assumption \ref{as:params} hold, and  $S_0>0$ and $(I_n^0)_{n\in\mathbb{N}}\in \ell^1_+$ be given. Suppose that  {either $\mathcal{R}_0<1$, or $\mathcal{R}_0=  1$ and Assumption \ref{as:omega} holds true}. Then we have
	\begin{equation*}
		\lim_{t\to+\infty} \sum_{n=0}^{+\infty} I_n(t) = 0.
	\end{equation*}
\end{proposition}
When $\mathcal{R}_0>1$, on the contrary, we can show that the pathogen survives in large time. 
\begin{theorem}[Persistence]\label{thm:discrete}
	Suppose that the Assumptions \ref{as:params}, \ref{as:init} and \ref{as:omega} hold true, and assume that $\mathcal{R}_0>1$. Then we have
	\begin{equation*}
		S(t) \xrightarrow[t\to+\infty]{} \frac{1}{\alpha^*} \text{ and } S'(t)\xrightarrow[t\to+\infty]{} 0.
	\end{equation*}
	Concerning the behavior of $I_n(t)$, we distinguish two cases.
	\begin{enumerate}[label={\rm \roman*)}]
		\item \label{item:persistence-convergence}
		    Suppose that there is {some} $i\in\mathbb{N}$ such that  $\frac{\beta_i}{\gamma_i} = \alpha^*$, possibly for multiple indices. Then  $ I_n(t) $ converges in $\ell^1$ to the following asymptotic stationary state 
			\begin{equation*}
				I^\infty_n  = \begin{cases}
					0 & \text{ if }  \frac{\beta_n}{\gamma_n}<\alpha^*, \\ 
					e^{\tau \gamma_n}I_n^0 & \text{ if } \frac{\beta_n}{\gamma_n}=\alpha^*, 
				\end{cases} \text{ for all } n\in\mathbb N, 
			\end{equation*}
			where the constant $\tau\in\mathbb R$ is the unique solution of the equation:
			\begin{equation*}
				\sum_{\{n\in\mathbb N\,:\,\frac{\beta_n}{\gamma_n} = \alpha^*\}} \gamma_n I_n^0\, e^{\tau \gamma_n} = \frac{\theta}{\alpha^*}(\mathcal R_0-1). 
			\end{equation*}
		\item \label{item:persistence-divergence}
			Suppose that for all $n\in\mathbb{N}$,  we have $\frac{\beta_n}{\gamma_n}<\alpha^*.$ Then one has  $I_n(t)\to 0$ for all $n\in\mathbb N$ as $t\to\infty$, while 
			\begin{equation}\label{eq:divergencecase-persistence}
				\liminf_{t\to+\infty} \sum_{n\in\mathbb N}I_n(t) >0. 
			\end{equation}
			Moreover if $(n_k)_{k\in\mathbb{N}}$ is a sequence of integers such that 
			\begin{equation*}
				\sup_{k\in\mathbb{N}} \dfrac{\beta_{n_k}}{\gamma_{n_k}}<\alpha^*, 
			\end{equation*}
			then we have 
			\begin{equation}\label{eq:divergencecase-extinction}
				\limsup_{t\to+\infty} \sum_{k\in\mathbb{N}} I_{n_k}(t) = 0.
			\end{equation}
	\end{enumerate}
\end{theorem}
{The proof of Theorem \ref{thm:discrete} will be given in section \ref{sec:discrete}.}

{Let us explain in a few words the content of Theorem \ref{thm:discrete}. Our basic assumption is that the basic reproduction number $\mathcal{R}_0$ is greater than one, because the infected population gets extinct otherwise. There are two typical situation. Case \ref{item:persistence-convergence} corresponds to the case when the maximal fitness is attained by at least one variant; in this case, the behavior of the infinite system is similar to the one of the finite system: we observe the unconditional convergence to an equilibrium state, that can be computed from the initial data. The case \ref{item:persistence-divergence}, when none of the variants attain the maximal fitness, is more intricate. In general, we can only draw two conclusions: the first is that the pathogen persists in large time as a whole (that's \eqref{eq:divergencecase-persistence}), and the second is that any family of variants whose fitness is uniformly dominated, eventually gets extinct (that's \eqref{eq:divergencecase-extinction}). In section \ref{sec:examples} we will give a counterexample showing that, in some cases, the total pathogen population does not converge to a limit.  }  

  {To go deeper in our analysis, we} can be somewhat more precise on the behavior of the total pathogen population at the expense of a slightly stronger assumption on the coefficients. {When the phenotypic values $\beta_n$ and $\gamma_n$ are in some sense uniformly represented in the initial state}, the pathogen strains that win the competition are the ones that maximize $\frac{\beta_n}{\gamma_n}$ first, and then $\gamma_n$, {as we will show in Proposition \ref{prop:convergence-mass}}. {We precise now what we mean by ``in some sense uniformly represented in the initial state''}. First we properly define the notion of ``maximal reachable recovery rate''. 
\begin{definition}[Maximal reachable recovery rate]\label{def:maximal-reachable-recovery-rate}
    Let Assumptions \ref{as:params} and \ref{as:omega} hold. Then the \textit{maximal reachable recovery rate} is 
	\begin{equation*}
	    \gamma^*:=\sup\left\{\gamma\in\omega\left((\gamma_n)_{n\in\mathbb{N}}\right)\,:\, (\alpha^*, \gamma)\in\omega\left( \big((\alpha_n, \gamma_n)\big)_{n\in\mathbb{N}}\right)\right\},
	\end{equation*}
	wherein $\omega\left( \big((\alpha_n, \gamma_n)\big)_{n\in\mathbb{N}}\right)$ is the $\omega$-limit set of the joint sequence $\big((\alpha_n, \gamma_n)\big)_{n\in\mathbb{N}}$.
	In other words, $\gamma^*$ is the maximal limit value of a subsequence of $\gamma_n$ such that $\frac{\beta_{n_k}}{\gamma_{n_k}}\to \alpha^*$.
\end{definition}
In the following assumption we impose that the initial mass of pathogens is never negligible around the maximal value $\gamma^*$. 
\begin{assumption}\label{as:disintegration}
	Suppose that the Assumptions \ref{as:params}, \ref{as:init} and \ref{as:omega} hold true, and let $\gamma^*$ be the maximal reachable recovery rate as in Definition \ref{def:maximal-reachable-recovery-rate}. We assume that for each $\varepsilon>0$ sufficiently small, there exist constants $\delta>0$ and $m>0$ such that for each value $\alpha\in [\alpha^*-\delta, \alpha^*]$ such that there exists $n\in\mathbb{N}$ with $\frac{\beta_n}{\gamma_n}=\alpha$, we have
	\begin{equation*}
		\dfrac{\displaystyle \smashoperator[r]{\sum_{\left\{n\,:\,\frac{\beta_n}{\gamma_n}=\alpha \text{ and }\gamma_n\geq \gamma^*-\varepsilon\right\}}} I_n^0 \phantom{balbla}}{\displaystyle\smashoperator[r]{\sum_{\left\{n\,:\,\frac{\beta_n}{\gamma_n}=\alpha \right\}}}I_n^0} \geq m. 
	\end{equation*}
	In other words, the probability of ``picking'' a pathogen with $\gamma_n$ close to $\gamma^*$ conditionally to the fact that $\alpha_n=\alpha$ with $\alpha$ close to $\alpha^*$ has a uniform positive lower bound. 
\end{assumption}
{Let us give a few examples of initial data that do or do not satisfy the assumption \ref{as:disintegration}. If $\gamma_n$ converges to its limit, then the limit is necessarily $\gamma^*$ and assumption \ref{as:disintegration} holds independently of the initial data $I^0_n>0$. If $\omega\big((\gamma_n)_{n\in\mathbb{N}}\big)=\{\gamma^1, \gamma^2\}$ has exactly two elements $\gamma^1<\gamma^2$, we have to look at the equality cases $\mathcal{A}_n:=\{k\,:\, \alpha_k=\alpha_n\}$. For simplicity, suppose that  $\gamma_{2n}\to \gamma^1$ and $\gamma_{2n+1}\to\gamma^2$ and $\mathcal{A}_n=\{2n, 2n+1\}$,  then Assumption \ref{as:disintegration} is satisfied if, and only if, there is $m>0$ such that 
\begin{equation*}
    \dfrac{I^0_{2n+1}}{I^0_{2n}+I^0_{2n+1}}\geq m. 
\end{equation*}
When Assumption \ref{as:disintegration} is satisfied, we can prove the following result.
}
\begin{proposition}\label{prop:convergence-mass}
	Suppose that Assumption \ref{as:disintegration} holds and that $\frac{\beta_n}{\gamma_n}<\alpha^*$ for all $n\in\mathbb{N}$. 
	Then the total pathogen population converges to a positive limit 
	\begin{equation}\label{eq:limit-mass-discrete}
		\lim_{t\to+\infty} \sum_{n=0}^{+\infty}I_n(t) = \frac{\theta}{\alpha^*\gamma_\infty}(\mathcal R_0-1).
	\end{equation}
	Moreover if $n_k\in\mathbb{N}$ is a sequence of integers  such that 
	\begin{equation} \label{eq:propconvergence-suboptimal-trait}
		\sup_{k\in\mathbb{N}} \dfrac{\beta_{n_k}}{\gamma_{n_k}}=\alpha^* \text{ and } \sup_{k\in\mathbb{N}} \gamma_{n_k}<\gamma^*, 
	\end{equation}
	then we have 
	\begin{equation}\label{eq:propconvergence-suboptimal-vanishes}
		\limsup_{t\to+\infty} \sum_{k\in\mathbb{N}} I_{n_k}(t) = 0.
	\end{equation}
\end{proposition}

The proof of the results of this Section will be given in section \ref{sec:proof}. First we present some particular choices for which the replacement dynamics of the variants can be analytically understood.

\section{Examples}
\label{sec:examples}
    In this section we provide examples of explicit choices of the coefficients for which the asymptotic dynamics can be understood analytically. We also provide numerical simulations of the corresponding set of ODEs, with a particular attention to the dynamics of variants replacement. 

We subdivide further the section in two subsections: in section \ref{sec:monovalent-gamma} we deal with examples for which $\gamma_n\equiv \gamma$ is a constant sequence, while in section \ref{sec:bivalent-gamma} we deal with examples for which $\gamma_n$ takes alternatively two values.

\subsection{Replacement dynamics 1: Monovalent $\gamma_n$}
\label{sec:monovalent-gamma}
{
In this section we investigate the asymptotic transition time between the prevalence of two given variants, depending on their characteristics. We place ourselves in the case when $\gamma_n$ is \textit{monovalent}, that is to say, $\gamma_n\equiv \gamma>0$ is independent of $n\in\mathbb{N}$. We also assume that $\beta_n<\beta^*$ for all $n\geq 0$, so that we already know that the mass converges thanks to Proposition \ref{prop:convergence-mass}:
\begin{equation}\label{eq:convergence-mass}
    \mathcal{I}(t):=\sum_{n=0}^{+\infty}I_n(t)\xrightarrow[t\to+\infty]{}\mathcal{I}_\infty.
\end{equation}
Now let us observe that the total mass rewrites as
\begin{equation*}
    \sum_{n=0}^{+\infty}I_n(t) = \sum_{n=0}^{+\infty}I_n^0e^{\beta_n \overline{S}(s)-\gamma_n t} = \sum_{n=0}^{+\infty} I_n^0e^{(\beta_n-\beta^*)\overline{S}(t) t +(\beta^*\overline{S}(t)- \gamma) t} = e^{(\beta^*\overline{S}(t)-\gamma)t}\sum_{n=0}^{+\infty} I_n^0e^{-(\beta^*-\beta_n)\overline{S}(t) t},
\end{equation*} 
and in particular:
\begin{equation}\label{eq:fundamental}
	\sum_{n=0}^{+\infty}I_n(t) = e^{t\left(\beta^*\overline{S}(t)-\gamma\right)} F\big(t\overline{S}(t)\big), 
\end{equation} 
where $F$ is the function defined by 
\begin{equation}\label{eq:defF}
	F(\xi):=\sum_{n=0}^{+\infty}I_n^0 e^{-(\beta^*-\beta_n)\xi}.
\end{equation}
}

Notice that, since $\beta_n<\beta^*$ for all $n\geq 0$, by \eqref{eq:defF} we have
\begin{equation*}
	\lim_{\xi\to+\infty} F(\xi)=0. 
\end{equation*}
so that, using the fact that $\sum I_n(t)=\mathcal{I}_\infty+o(1)>0$, we get
\begin{equation*}
    \lim_{t\to+\infty} e^{(\beta^*\overline{S}(t) - \gamma)t} = +\infty.
\end{equation*}
Moreover, we have
\begin{equation*}
    \ln\left(\sum_{n=0}^{+\infty}I_n(t)\right)= t\left(\beta^*\overline{S}(t)-\gamma\right) + \ln\left(F(t\overline{S}(t))\right)
\end{equation*}
so that thanks to Theorem \ref{thm:discrete} we have
\begin{equation*}
    \frac{1}{t}\ln\big(F(\overline{S}(t))\big) = -{\left(\beta^*\overline{S}(t)-\gamma\right)}+\frac{1}{t}\ln\left(\sum_{n=0}^{+\infty}I_n(t)\right)\xrightarrow[t\to+\infty]{}0.
\end{equation*}
These algebraic remarks will serve to  estimate the replacement speed of the variants. In particular, we will use the key relation
\begin{equation}\label{eq:formula-Sbar}
	\overline{S}(t) = \frac{\gamma}{\beta^*}   -\frac{1}{t\beta^*}\ln\big(F(t\overline{S}(t))\big) +\frac{1}{t\beta^*} \ln\left(\sum_{n=0}^{+\infty} I_n(t) \right).
\end{equation}

{
\subsubsection{Monovalent example 1: the algebraic-algebraic case.}

In this subsection we assume the following framework. 
\begin{assumption}\label{as:algebraic-algebraic}
We let $\gamma_n\equiv\gamma>0$ be a constant sequence and assume that the initial data is algebraic and the convergence of the fitness to its maximum is algebraic :
\begin{equation}\label{eq:algebraic-algebraic}
	I_n^0=\frac{1}{(n+1)^A}\text{ with $A>1$ while }\beta_n=\beta^*-\frac{1}{n+1}\text{ for all $n\geq 0$}.
\end{equation}
\end{assumption}
\begin{claim}\label{claim:algebraic-algebraic}
Under assumption \ref{as:algebraic-algebraic}, there exist constants $C>1$ and $\hat t>0$ large enough such that for all $t\geq \hat t$ and for all $n\geq 0$ one has 
\begin{equation*}
\frac{1}{Ct}f\left(\frac{(n+1)}{\gamma t+(A-1)\ln t}\right)\leq I_n(t)\leq \frac{C}{t}f\left(\frac{(n+1)}{\gamma t+(A-1)\ln t}\right),
\end{equation*} 
where the profile function $f$ is given by
\begin{equation}\label{function3}
f(X)=\exp\left(-A \ln X-\frac{1}{X\beta^*}\right).
\end{equation}
\end{claim}
\begin{figure}[H]
    \centering
    \includegraphics{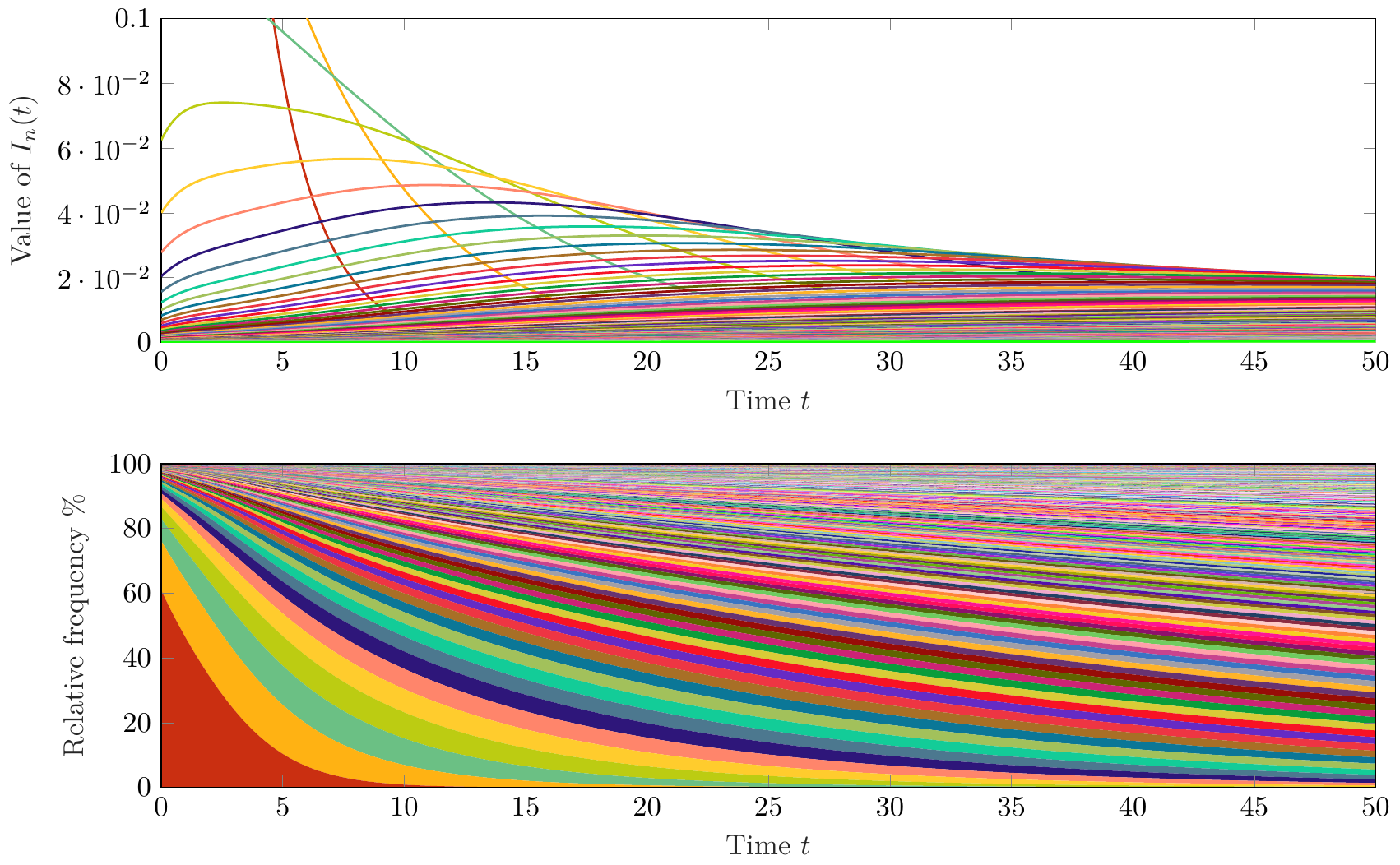}
    \caption{Plots of the solution $I_n(t)$ in the monovalent case 1. {Colors are chosen at random so that each variant has a different color.} The bottom figure suggests that the diversity of variants increases with time, since there are more and more lines with a noticeable width.  \textbf{Top figure:} Value of $I_n(t)$. \textbf{Bottom figure:} Relative frequencies of the variants as a function of time. This is the proportion of each variant in the global population. \textbf{Parameters:} $\Lambda=1$, $\theta=1$, $S_0=1$, $\gamma=\frac{1}{2}$, $A=2$, $\beta^*=1$. We used a total of $N=300$ variants for this simulation. For interpretation
of the colors in the figure(s), the reader is referred to the online version of this article.}\label{Fig1}
\end{figure}

We start with a technical lemma.
\begin{lemma}\label{LE-exp}
    The function $F(\xi)$ defined by \eqref{eq:defF} satisfies 
$$
    \ln F(\xi)=-(A-1)\ln \xi+\mathcal{O}(1)\text{ as }\xi\to\infty.
$$
\end{lemma}

\begin{proof}
Let $u_n(\xi):=\exp\left(-A\ln (n+1)-\frac{\xi}{n+1}\right)$ so that
$$
F(\xi)=\sum_{n=0}^\infty u_n(\xi).
$$
For $\xi>0$ we define $x(\xi)>0$ and $n(\xi)\in\mathbb N$ by
$$
x(\xi)+1=\frac{\xi}{A}\text{ and } n(\xi)=\lfloor x(\xi)\rfloor\text{ the integer part of $x(\xi)$}.
$$
Next we have for $p\in\mathbb N$ and $\xi>0$:
$$
\frac{u_{n(\xi)+p}(\xi)}{u_{n(\xi)}(\xi)}=\left(\frac{n(\xi)+1}{n(\xi)+1+p}\right)^A\exp\left(\xi\frac{p}{(n(\xi)+1)(n(\xi)+1+p)}\right)
$$
Hence we get
\begin{equation*}
\begin{split}
    \sum_{p=0}^{\infty}\frac{u_{n(\xi)+p}(\xi)}{u_{n(\xi)}(\xi)}&\leq \sum_{p=0}^\infty\left(\frac{n(\xi)+1}{n(\xi)+1+p}\right)^A\exp\left(\frac{\xi}{(n(\xi)+1)}\right)\\
&\leq (n(\xi)+1)^A\exp\left(\frac{\xi}{(n(\xi)+1)}\right) \sum_{k=n(\xi)+1}^\infty\frac{1}{k^A}.
\end{split}
\end{equation*}
Since 
    \begin{equation*}
	\sum_{k=n(\xi)+1}^{+\infty} \frac{1}{k^A} \leq \int_{n(\xi)+1}^{+\infty} \frac{1}{y^A}\dd y = \frac{1}{A-1}(n(\xi)+1)^{-(A-1)}, 
    \end{equation*}
we obtain that there exists some constant $C>0$ and $\xi_0$ large enough such that 
\begin{equation*}
\sum_{p=0}^\infty\frac{u_{n(\xi)+p}(\xi)}{u_{n(\xi)}(\xi)}
\leq C n(\xi)\exp\left(\frac{\xi}{(n(\xi)+1)}\right) \;\;\forall \xi>\xi_0.
\end{equation*}
This also rewrites as
\begin{equation*}
\sum_{p=0}^\infty u_{n(\xi)+p}(\xi)
\leq C n(\xi)u_{n(\xi)}(\xi)\exp\left(\frac{\xi}{(n(\xi)+1)}\right) \;\;\forall \xi>\xi_0.
\end{equation*}

On the other hand we have
$$
\sum_{k=0}^{n(\xi)} \frac{u_k(\xi)}{u_{n(\xi)}(\xi)}=\sum_{k=0}^{n(\xi)} \left(\frac{n(\xi)+1}{k+1}\right)^A \exp\left(\frac{\xi}{n(\xi)+1}\frac{k-n(\xi)}{k+1}\right).
$$
Next due to the definition of $n(\xi)$ we have
$$
     n(\xi)\leq x(\xi)=\frac{\xi}{A}-1<n(\xi)+1,
$$
so that
$$
    \frac{A n(\xi)}{n(\xi)+1}\leq\frac{\xi{-A}}{n(\xi)+1}\leq A .
$$
Hence there exists $\xi_1>0$ large enough so that
$$
A \geq\frac{\xi}{n(\xi)+1}\geq \frac{A}{2},\;\forall \xi>\xi_1.
$$
As a consequence for all $\xi>\xi_1$ we have
\begin{equation*}
\begin{split}
\sum_{k=0}^{n(\xi)} \frac{u_k(\xi)}{u_{n(\xi)}(\xi)}&\leq \sum_{k=0}^{n(\xi)} \left(\frac{n(\xi)+1}{k+1}\right)^A \exp\left(\frac{A}{2}\frac{k-n(\xi)}{k+1}\right)\\
&\leq \sum_{k=0}^{n(\xi)} \left(\frac{n(\xi)+1}{k+1}\right)^A \exp\left(\frac{A}{2}\left(1-\frac{n(\xi)+1}{k+1}\right)\right)
\end{split}
\end{equation*}
and similarly
$$
\sum_{k=0}^{n(\xi)} \frac{u_k(\xi)}{u_{n(\xi)}(\xi)}\geq \sum_{k=0}^{n(\xi)} \left(\frac{n(\xi)+1}{k+1}\right)^A \exp\left(A\left(1-\frac{n(\xi)+1}{k+1}\right)\right).
$$
Now using Riemann sums with the continuous function 
$$
x\mapsto \begin{cases} 0\text{ if $x\leq 0$}\\ \frac{1}{x^A}\exp\left(\frac{1}{2A}\left(1-\frac{1}{x}\right)\right)\text{ if $x>0$},
\end{cases}
$$
let us observe that we have
$$
n\sum_{k=0}^{n} \left(\frac{n+1}{k+1}\right)^A \exp\left(\frac{1}{2A}\left(1-\frac{n+1}{k+1}\right)\right)\to \int_0^1 \frac{1}{x^A}\exp\left(\frac{1}{2A}\left(1-\frac{1}{x}\right)\right)dx\text{ as }n\to\infty.
$$
As a consequence, there exists some constant, still denoted by $C>1$ large engou and $\xi_2>0$ large enough such that
\begin{equation*}
\begin{split}
C^{-1} n(\xi)u_{n(\xi)}(\xi)\leq \sum_{k=0}^{n(\xi)} u_k(\xi)&\leq Cn(\xi)u_{n(\xi)}(\xi),\;\forall \xi>\xi_2.
\end{split}
\end{equation*}
Coupling the two above estimates ensures that there exists $C>1$ and $\hat{\xi}\gg 1$ large enough so that
$$
F(\xi)\leq C n(\xi)u_{n(\xi)}(\xi)\leq C\xi^{1-A},\;\forall \xi>\hat \xi,
$$
while 
$$
F(\xi)\geq \sum_{k=0}^{n(\xi)} u_k(\xi)\geq C^{-1}n(\xi)u_{n(\xi)}(\xi).
$$
The completes the proof of the lemma.

\end{proof}

\begin{proof}[Proof of Claim \ref{claim:algebraic-algebraic}]
As before the proof relies on Lemma \ref{LE-exp} and \eqref{eq:formula-Sbar}. Indeed using these  we have
	\begin{align*}
		\overline{S}(t) &= \frac{\gamma}{\beta^*} -\frac{1}{t\beta^*}\ln\big(F(t\overline{S}(t))\big) +\frac{1}{t\beta^*} \ln\left(\sum_{n=0}^{+\infty} I_n(t) \right) \\ 
		&=\frac{\gamma}{\beta^*} +\frac{A-1}{\beta^*}\dfrac{\ln(t)}{t} +\mathcal{O}\left(\frac{1}{t}\right). 
	\end{align*}
	As a consequence we obtain
	\begin{align*}
		I_n(t) &= I_n^0\exp\left(\big(\beta^*-\frac{1}{n+1}\big) \left(\frac{\gamma}{\beta^*} + \frac{A-1}{\beta^*} \frac{\ln(t)}{t}+\mathcal{O}\left(\frac{1}{t}\right)\right)t-\gamma t \right) \\ 
		&= I_n^0 \exp\left((A-1)\ln(t)-\frac{1}{n+1}\left(\frac{\gamma}{\beta^*}+\frac{A-1}{\beta^*}\frac{\ln(t)}{t}\right)t+\mathcal{O}(1)\right)\\
		&= \exp\left(-A \ln(n+1)+(A-1)\ln(t)-\frac{1}{n+1}\left(\frac{\gamma}{\beta^*}+\frac{A-1}{\beta^*}\frac{\ln(t)}{t}\right)t+\mathcal{O}(1)\right).
	\end{align*}
This rewrites uniformly for $n\geq 0$ as follows
	\begin{align*}
		I_n(t) &= e^{\mathcal O(1)} \frac{1}{t}\exp\left(-A \ln\frac{(n+1)}{\gamma t+(A-1)\ln t}-\frac{1}{\beta^*}\frac{\gamma t+(A-1)\ln t}{n+1}\right). 
	\end{align*}
	Recalling the definition of the function $f$ in \eqref{function3}, the above equality becomes
	$$
	I_n(t)\asymp \frac{1}{t}f\left(\frac{(n+1)}{\gamma t+(A-1)\ln t}\right)\text{ with }f(X)=\exp\left(-A \ln X-\frac{1}{X\beta^*}\right),
	$$
	that completes the proof of the lemma.
	\end{proof}
	
}

\subsubsection{Monovalent example 2: the exponential-exponential case.}
{We work under the following assumption.
{
\begin{assumption}\label{as:analytic-example-1}
	We assume that $\gamma_n\equiv \gamma>0$ is a positive constant, $I_n^0=e^{-An}$ and  $\beta_n=\beta^*-Be^{-Cn}$ for some constants  $A>0$, $B>0$, and $C>0$. In other words, the initial data is exponential and the convergence of the fitness to its maximum is also exponential. 
\end{assumption}
In this case we can show that the function $F(\xi)$ behaves like $\xi^{-\frac{A}{C}}$ and that $I_n(t)$ eventually reaches a fixed shape shifting toward $+\infty$ like $\frac{1}{C}\ln(t)$.
More precisely $I_n(t)$ behaves asymptotically like $ \mathcal{E}\big(n-n_0(t)\big)$ as $t\to+\infty$, where $n_0(t):=\frac{1}{C}\ln(t)$ and
\begin{equation}\label{eq:defcalE}
	\mathcal{E}(\nu):=\exp\left(-\frac{\gamma}{\beta^*} Be^{-C\nu} - A\nu\right).
\end{equation}
\begin{claim}[Asymptotic behavior]\label{claim:exponential-exponential}
    Let Assumption \ref{as:analytic-example-1} hold and $\mathcal{E}$ be defined by \eqref{eq:defcalE}. Let $n_0(t):=\frac{1}{C}\ln(t) $ and $\mu(t):=\mu_0\dfrac{\ln \big(\ln t\big)}{\ln t}$ for $\mu_0>\max(C, 1)$ and $t>1$. Then there exists a constant $K>0$ such that for all $n\geq \mu(t) n_0(t)$ we have  
	\begin{equation*} 
	    \frac{1}{K} \mathcal{E}\big(n-n_0(t)\big) \leq I_n(t) \leq K \mathcal{E}\big(n-n_0(t)\big).
	\end{equation*}
	and $\displaystyle\sum_{k=0}^{\lfloor \mu(t)n_0(t)\rfloor}I_k(t) \xrightarrow[t\to+\infty]{}0$ {(where $\lfloor\mu(t)n_0(t)\rfloor$ is the biggest integer smaller than $\mu(t)n_0(t)$)}.
\end{claim}
\begin{figure}[H]
    \centering
    \includegraphics{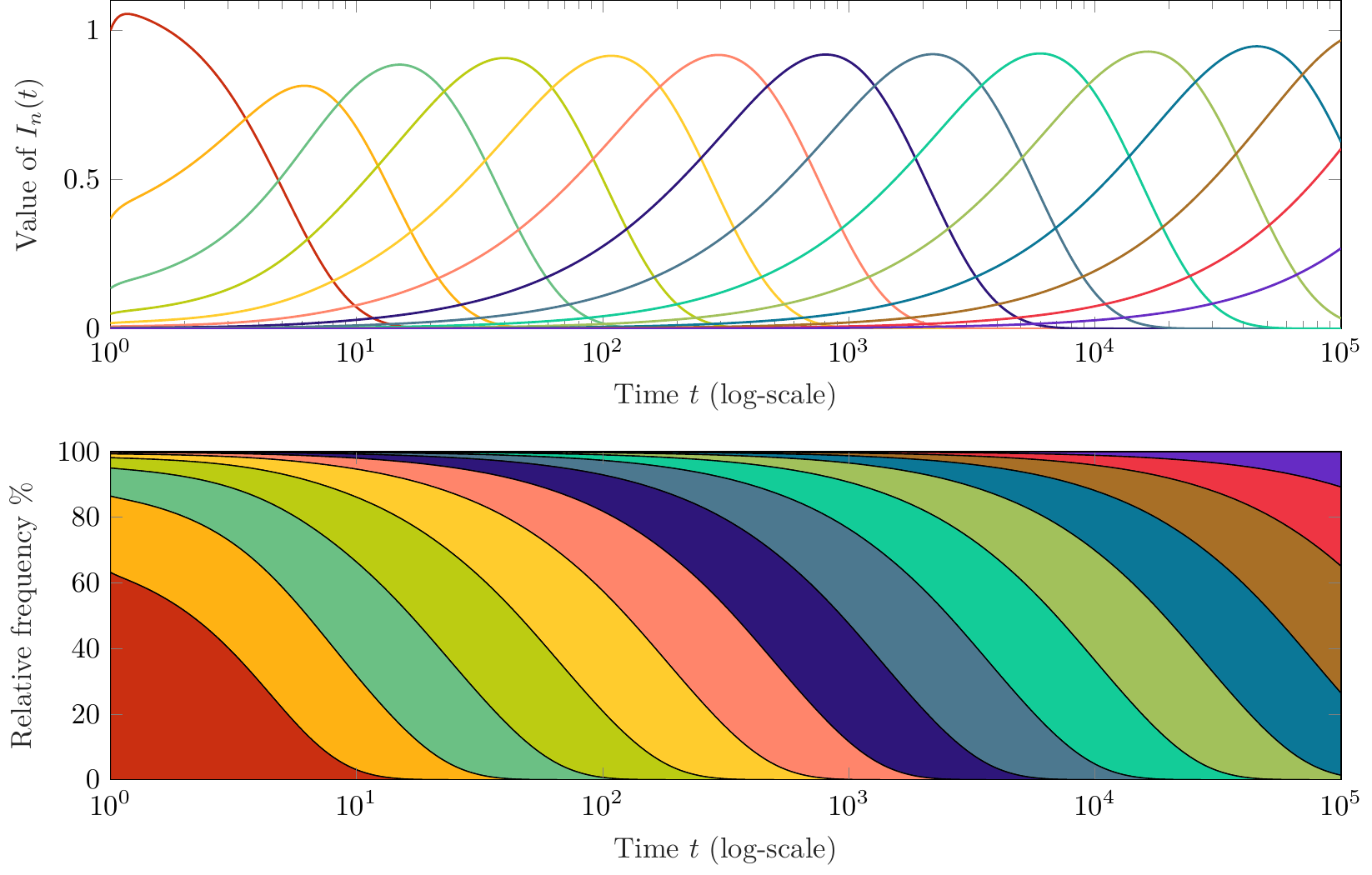}
	\caption{Plots of the solution $I_n(t)$ in the monovalent case 2. {Colors are chosen at random so that each variant has a unique color.} The top figure suggests that each variant has a similar behavior in time. The bottom figure suggests that the diversity of variants remains approximately constant in time (more precisely, the number of non-negligible variants in the population is uniformly greater than a positive constant). \textbf{Top figure:} Value of $I_n(t)$. \textbf{Bottom figure:} Relative frequencies of the variants as a function of time. \textbf{Parameters:} $\Lambda=10$, $\theta=10$, $S_0=1$, $\gamma={2}$, $\beta^*=4$, $A=1$, $B=1$, $C=1$. We used a total of $N=14$ variants for this simulation. For interpretation of the colors in the figure(s), the reader is referred to the online version of this article.}\label{Fig2}
\end{figure}

}
 We begin with the asymptotic expansion of $F$. 
\begin{lemma}\label{claim:asymp-F-example-1}
	Let Assumption \ref{as:analytic-example-1} hold and let $F(\xi)$ be defined by \eqref{eq:defF}. Then
	\begin{equation*}
	    F(\xi)= \xi^{-\frac{A}{C}}\times e^{\mathcal{O}(1)}\text{ as } \xi\to\infty.
	\end{equation*}
\end{lemma}}
\begin{proof}
	We claim that $F$ satisfies the functional equation:
	\begin{equation*}
		F(\xi) = e^{-B\xi}+e^{-A}F(e^{-C}\xi).
	\end{equation*}
	Indeed, 
	\begin{align*}
		F(\xi)&=\sum_{n=0}^{+\infty} e^{-B\xi e^{-C n}-A n}=\sum_{n=0}^{+\infty} e^{-A n}\sum_{k=0}^{+\infty} \dfrac{(-B\xi)^k}{k!}\big(e^{-C n}\big)^k=\sum_{k=0}^{+\infty} \dfrac{(-B\xi)^k}{k!}\sum_{n=0}^{+\infty} e^{-(A+kC)n} \\ 
		& = \sum_{k=0}^{+\infty}\dfrac{(-B\xi)^k}{k!}  \cdot\dfrac{1}{1-e^{-A-kC}}=\left(\sum_{k=0}^{+\infty} \dfrac{(-B\xi)^k}{k!}\right) + \sum_{k=0}^{+\infty} \dfrac{(-B\xi)^k}{k!}\left(\dfrac{1}{1-e^{-A-kC}}-1\right)\\
		&=e^{-B\xi}+\sum_{k=0}^{+\infty} \dfrac{(-B\xi)^k}{k!}\dfrac{e^{-A-kC}}{1-e^{-A-kC}} = e^{-B\xi}+e^{-A}\sum_{k=0}^{+\infty} \dfrac{(-Be^{-C}\xi)^k}{k!}=e^{-B\xi}+e^{-A}F\left(e^{-C\xi}\right).
	\end{align*}
	We deduce that
	\begin{align*}
		F(e^{C}\xi)&=e^{-Be^{C}\xi}+e^{-A}F(\xi),\\
		F(e^{2C}\xi)&=e^{-Be^{2C}\xi}+e^{-A}F(e^{C}\xi) = e^{-Be^{2C}\xi}+e^{-A-Be^{C}\xi}+e^{-2A}F(\xi),\\
		F(e^{2C}\xi)&=e^{-Be^{3C}\xi}+e^{-A}F(e^{2C}\xi)
		= e^{-Be^{3C}\xi}+e^{-A-Be^{2C}\xi}+e^{-2A-Be^{C}\xi}+e^{-3A}F(\xi),\\
		&\quad\vdots \\ 
		F(e^{nC}\xi) &= e^{-nA}F(\xi)+ \sum_{k=1}^{n}e^{-(n-k)A -Be^{kC}\xi} = e^{-nA}\left(F(\xi)+\sum_{k=1}^{n}e^{kA-B e^{kC}\xi}\right).
	\end{align*}
	For $\xi\in\left[1, e^{C}\right)$, we let $X=\xi e^{nC}$ so that 
	\begin{equation*}
		n=\left\lfloor \frac{1}{C}\ln(X)\right\rfloor \text{ and } \xi = Xe^{-C\lfloor\ln(X)/C\rfloor}.
	\end{equation*}
	We get:
	\begin{equation*}
		F(X)=\frac{1}{\exp\left(A\left\lfloor\frac{1}{C}\ln(X)\right\rfloor\right)} \left(F\left(\xi\right)+\sum_{k=1}^{n}e^{kA-B e^{kC}\xi}\right),
	\end{equation*}
	and finally, as $X\to+\infty$,
	\begin{equation*}
		F(X)={X^{-\frac{A}{C}}}\times e^{\mathcal{O}(1)}. \qedhere
	\end{equation*}
\end{proof}
{ 
\begin{proof}[Proof of Claim \ref{claim:exponential-exponential}]
	Thanks to Claim \ref{claim:asymp-F-example-1} we have, recalling \eqref{eq:formula-Sbar}
	\begin{align*}
		\overline{S}(t) &= \frac{\gamma}{\beta^*} -\frac{1}{t\beta^*}\ln\big(F(t\overline{S}(t))\big) +\frac{1}{t\beta^*} \ln\left(\sum_{n=0}^{+\infty} I_n(t) \right) \\ 
		&=\frac{\gamma}{\beta^*} +\frac{A}{C\beta^*}\dfrac{\ln(t)}{t} +\mathcal{O}\left(\frac{1}{t}\right). 
	\end{align*}
	Therefore, 
	\begin{align*}
		I_n(t) &= I_n^0\exp\left(\big(\beta^*-B e^{-C n}\big) \left(\frac{\gamma}{\beta^*} + \frac{A}{C\beta^*} \frac{\ln(t)}{t}+\mathcal{O}\left(\frac{1}{t}\right)\right)t-\gamma t \right) \\ 
		&= I_n^0 \exp\left(\frac{A}{C}\ln(t)-B e^{-C n}\left(\frac{\gamma}{\beta^*}+\frac{A}{C\beta^*}\frac{\ln(t)}{t}\right)t+\mathcal{O}(1)\right)\\
		&= \exp\left(-Be^{-C n + \ln(t)} \left(\frac{\gamma}{\beta^*}+\frac{A}{C\beta^*}\frac{\ln(t)}{t}\right) +\frac{A}{C}\ln(t)-A n + \mathcal{O}(1)\right) \\ 
		&=  \exp\left(-\left(1+\frac{A}{C\gamma}\frac{\ln(t)}{t}\right)\frac{\gamma}{\beta^*}B e^{-C\left(n-\frac{1}{C}\ln(t)\right)} - A\left(n-\frac{1}{C}\ln(t)\right)+\mathcal{O}(1)\right) \\ 
		&=  \exp\left(-\left(1+\frac{A}{C\gamma}\frac{\ln(t)}{t}\right)\frac{\gamma}{\beta^*}B e^{-C\left(n-n_0(t)\right)} - A\big(n-n_0(t)\big)+\mathcal{O}(1)\right)  
	\end{align*}
	with $n_0(t)=\frac{1}{C}\ln(t)$. Thus for $(n, t)\in \{(\tilde n, \tilde t)\,:\, \tilde n \geq \mu(\tilde t)n_0(\tilde t)\}$ we have
	\begin{equation*}
	    \dfrac{I_n(t)}{\mathcal{E}\big(n-n_0(t)\big)}=\exp\left(-\frac{A}{C\gamma}\frac{\ln(t)}{t}\frac{\gamma}{\beta^*}B e^{-C\left(n-n_0(t)\right)}+\mathcal{O}(1)\right) = e^{\mathcal{O}(1)}, 
	\end{equation*}
	since 
	\begin{equation*}
	    0\leq \dfrac{\ln(t)}{t}e^{-C(n-n_0(t))}\leq e^{\ln(\ln t) - \ln(t) -(\mu(t)-1)\ln(t)} = e^{(1-\mu_0)\ln(\ln t )}\xrightarrow[t\to+\infty]{}0.
	\end{equation*}
	On the other hand if $n\leq \mu(t)n_0(t)$ then 
	\begin{align*}
	    I_n(t)&=\exp\left[-\frac{\gamma}{\beta^*}Bt e^{-Cn}-\frac{A}{C\beta^*}\ln(t)e^{-Cn} - An+\frac{A}{C}\ln(t)+\mathcal{O}(1) \right] \\ 
	    &\leq \exp\left[-\left(\frac{\gamma}{\beta^*}Bt-\frac{A}{C}\right) e^{-Cn} - An + \frac{A}{C}\left(\ln(t)-e^{-n}\right)+\mathcal{O}(1)\right] \leq e^{-An+\mathcal{O}(1)}, 
	\end{align*}
	because $n\leq \mu(t)n_0(t)=\frac{\mu_0}{C}\ln(\ln t) $ with $\mu_0>C$ implies $e^{-n}\geq \ln(t)^{\frac{2}{C}}\geq \ln(t)$ for $t$ large, whence 
	\begin{equation*}
	    \displaystyle\sum_{n=0}^{\lfloor \mu(t)n_0(t)\rfloor}I_n(t) \xrightarrow[t\to+\infty]{}0
	\end{equation*}
	by Lebesgue's dominated convergence Theorem.
	This proves Claim \ref{claim:exponential-exponential}.
\end{proof}
}

\subsubsection{Monovalent example 3: the Gaussian-exponential case.}

{
\begin{assumption}\label{as:gaussian-exponential}
Assume that the initial data is Gaussian and the convergence of the fitness to its maximum is exponential:
\begin{equation}\label{case1}
	I_n^0=e^{-An^2}\text{ while }\beta^*-\beta_n=Be^{-Cn}\text{ for all $n\geq 0$},
\end{equation}
with $A>0$, $B>0$ and $C>0$ given constants.
\end{assumption}
}
Under this Assumption we can prove that the dynamics of the entire family of variants is actually guided by only three variants at a time at most (and most of the times by only one variant at a time) provided $C$ is sufficiently small. We let $W_0$ be the principal branch of the Lambert-$W$ function, that is to say the smooth real function satisfying $W_0(x)e^{W_0(x)}=x$ for all $x\geq e^{-1}$ and $\lim_{x\to+\infty}W_0(x)=+\infty$; see \textcite{Corless-etal-96} for details. We define:
\begin{equation}\label{eq:defY}
    Y(t) = \frac{1}{C}W_0\left(\dfrac{BC^2}{2A} t\overline{S}(t)\right).
\end{equation}
\begin{claim}\label{claim:gaussian-exponential}
    Let Assumption \ref{as:gaussian-exponential} hold. If $C$ is sufficiently small, then there exists $\delta>0$ small enough such that:
    \begin{equation*}
\sum_{n=0}^\infty I_n(t)\sim 
\begin{cases} 
I_{N(t)-1}(t)+I_{N(t)}(t)\text{ for $t\gg 1$ and $-1/2\leq R(t)\leq -1/2+\delta$},\\
I_{N(t)}\text{ for $t\gg 1$ and $-1/2+\delta \leq R(t)\leq1/2-\delta$},\\
I_{N(t)}(t)+I_{N(t)+1}(t)\text{ for $t\gg 1$ and $1/2-\delta<R(t)\leq 1/2$},
\end{cases}
\end{equation*}
    where $N(t)=\left\lfloor Y(t)+\frac{1}{2}\right\rfloor$, $ R(t) = Y(t)-N(t)$ and $Y(t)$ is defined in \eqref{eq:defY}.
\end{claim}
\begin{figure}[H]
    \centering
    \includegraphics{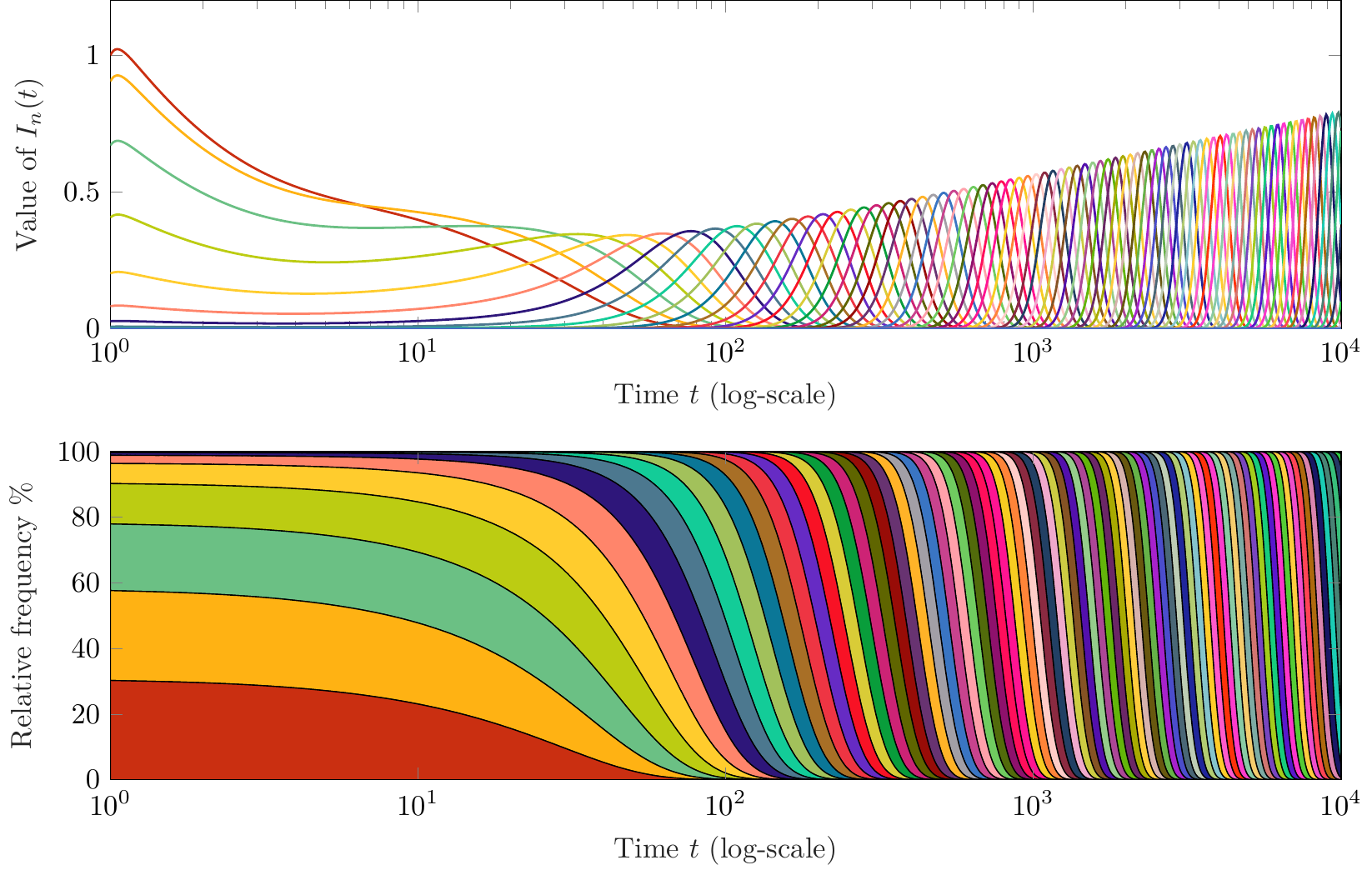}
	\caption{Plots of the solution $I_n(t)$ in the monovalent case 3. {Colors are chosen at random so that each variant has a unique color.} The bottom figure suggests that the diversity of variants diminishes in time since the steepness of the separation between variants is slowly increasing (and our theoretical predictions show that only one variant remains at carefully chosen times). \textbf{Top figure:} Value of $I_n(t)$. \textbf{Bottom figure:} Relative frequencies of the variants as a function of time.  \textbf{Parameters:} $\Lambda=10$, $\theta=10$, $S_0=1$, $\gamma={2}$, $\beta^*=1$, $A=0.1$, $B=1$, $C=0.03$. We used a total of $N=100$ variants for this simulation. For interpretation of the colors in the figure(s), the reader is referred to the online version of this article. }\label{Fig3}
\end{figure}
{
We define for $\xi>0$, $x(\xi)>0$ at the solution of the so-called {transcendental} Lambert 
equation
$$
xe^{{Cx}}=\frac{{B}C}{2A}\xi.
$$
{Remark that $x(\xi)$ can be expressed thanks to the principal branch of the Lambert-W function, $W_0(z)$. More precisely, $x(\xi) = \frac{1}{C}W_0\left(\frac{BC^2}{2A}\xi\right)$. }

{

We denote by $n(\xi)$ the integer which is the closest of $x(\xi)$, that is
\begin{equation}\label{def-n-r}
    x(\xi)=n(\xi)+r(\xi)\text{ with }n(\xi)=\left\lfloor x(\xi)+\frac{1}{2}\right\rfloor\in \mathbb N\text{ and }r(\xi)\in [-1/2,1/2).
\end{equation}

\begin{lemma}\label{lem:Gaussian-est}
The function $F$ defined in \eqref{eq:defF} satisfies the following asymptotic behavior
\begin{equation*}
    F(e^{-Cr(\xi)}\xi)\sim e^{-An(\xi)^2-B\xi e^{-Cx(\xi)}} {=e^{-An(\xi)^2-\frac{2A}{C}x(\xi)}}\text{ as }\xi\to\infty.
\end{equation*}
\end{lemma}

\begin{proof}
We define the function $G:\R^2\to\R$ by
\begin{equation}\label{funct-G}
G(x,\xi)=Ax^2+B\xi e^{-Cx},\;\forall (x,\xi)\in\R^2.
\end{equation}
Now we first claim that we have:
\begin{equation}\label{est1}
\sum_{n=n(\xi)+1}^\infty e^{-G(n,e^{-Cr(\xi)}\xi)}=o\left(e^{-G(n(\xi),e^{-Cr(\xi)}\xi)}\right)\text{ as }\xi\to \infty.
\end{equation}
To prove this property, observe that for $p\geq 1$ we have
\begin{equation*}
\begin{split}
G(n(\xi)+p,e^{-Cr(\xi)}\xi)-G(n(\xi),e^{-Cr(\xi)}\xi)&=Be^{-C(n(\xi)+p)}e^{-Cr(\xi)}\xi+A(n(\xi)+p)^2-Be^{-Cn(\xi)}e^{-Cr(\xi)}\xi-An(\xi)^2\\
&=B\left(e^{-Cp}-1\right)e^{-Cx(\xi)}\xi+A\left(2n(\xi)p+p^2\right)\\
&=B\left(e^{-Cp}-1\right)\frac{2A}{BC}x(\xi)+A\left(2n(\xi)p+p^2\right)\\
&=x(\xi)\left[-\left(1-e^{-Cp}\right)\frac{2A}{C}+2Ap\right]+A\left(p-r(\xi)\right)^2-Ar(\xi)^2\\
&\geq \frac{2A}{C}x(\xi)\left[-\left(1-e^{-Cp}\right)+Cp\right]+A\left(p-1/2\right)^2-A/4.
\end{split}
\end{equation*}
From the above estimate we obtain
\begin{equation}
\sum_{n=n(\xi)+1}^\infty e^{-G(n,e^{-Cr(\xi)}\xi)+G(n(\xi),e^{-Cr(\xi)}\xi)}\leq e^{-\frac{2A}{C}x(\xi)\left[C-\left(1-e^{-C}\right)\right]}\sum_{p=1}^\infty e^{-A\left(p-1/2\right)^2+A/4}\to 0\text{ as }\xi\to\infty,
\end{equation}
which proves \eqref{est1}.

Next we claim that we have
\begin{equation}\label{est2}
\sum_{n=0}^{n(\xi)-1} e^{-G(n,e^{-Cr(\xi)}\xi)}=o\left(e^{-G(n(\xi),e^{-Cr(\xi)}\xi)}\right)\text{ as }\xi\to \infty.
\end{equation}
To see this, note that for $1\leq p\leq n(\xi)$ we have
\begin{equation*}
\begin{split}
G(n(\xi)-p,e^{-Cr(\xi)}\xi)-G(n(\xi),e^{-Cr(\xi)}\xi)&=Be^{-C(n(\xi)-p)}e^{-Cr(\xi)}\xi+A(n(\xi)-p)^2-Be^{-Cn(\xi)}e^{-Cr(\xi)}\xi-An(\xi)^2\\
&=B(e^{Cp}-1)e^{-Cx(\xi)}\xi+A[(n(\xi)-p)^2-n(\xi)^2]\\
&=B(e^{Cp}-1)e^{-Cx(\xi)}\xi-A[2px(\xi)-p^2-2pr(\xi)]\\
&=\frac{2A}{C}x(\xi)[e^{Cp}-1-Cp]+A(p^2-2pr(\xi))\\
&\geq 2Ax(\xi)[e^{C}-1-C]+A(p^2-p).
\end{split}
\end{equation*}
Therefore we get
\begin{equation}
\sum_{n=0}^{n(\xi)-1} e^{-G(n,e^{-Cr(\xi)}\xi)+G(n(\xi),e^{-Cr(\xi)}\xi)}\leq e^{-\frac{2A}{C}x(\xi)\left[e^C-1-C\right]}\sum_{p=0}^{n(\xi)} e^{-A\left(p^2-p\right)}\to 0\text{ as }\xi\to\infty.
\end{equation}
Coupling \eqref{est1} and \eqref{est2} yields
    \begin{equation*}
F\left(e^{-Cr(\xi)}\xi\right)=e^{-G(n(\xi), e^{-Cr(\xi)}\xi)}(1+o(1))\text{ as }\xi\to\infty.
    \end{equation*}
    That completes the proof.
\end{proof}

When $C$ is sufficiently small we obtain a more refined estimate that is useful to understand the large time behavior of the sequence $\left(I_n(t)\right)_{n\geq 0}$.

\begin{lemma}\label{lem:Gaussian-C}
Fix $C>0$ small enough such that
\begin{equation}\label{cond-C}
e^{-C/2}(e^{2C}-1)-2C>0\text{ and }C-e^{C/2}\left(1-e^{-2C}\right)>0.
\end{equation}
Then the function $F$ defined in \eqref{eq:defF} satisfies
\begin{equation*}
    F(\xi)=\sum_{p=-1,0,1} e^{-A(n(\xi)+p)^2-B\xi e^{-C(n(\xi)+p)}}+o\left(e^{-An(\xi)^2-B\xi e^{-Cn(\xi)}}\right)\text{ as }\xi\to\infty.
\end{equation*}
\end{lemma}

\begin{proof}
Recall the definition of the function $G$ in \eqref{funct-G}, let us show that we have:
\begin{equation}\label{est1}
\sum_{p=n(\xi)+2}^\infty e^{-G(p,\xi)}=o\left(e^{-G(n(\xi),\xi)}\right)\text{ as }\xi\to \infty.
\end{equation}
To prove this property, observe that for $p\geq 2$ we have
\begin{equation*}
\begin{split}
G(n(\xi)+p,\xi)-G(n(\xi),\xi)&=Be^{-C(n(\xi)+p)}\xi+A(n(\xi)+p)^2-Be^{-Cn(\xi)}\xi-An(\xi)^2\\
&=B\left(e^{-Cp}-1\right)e^{-Cx(\xi)}e^{Cr(\xi)}\xi+A\left(2n(\xi)p+p^2\right)\\
&=e^{Cr(\xi)}B\left(e^{-Cp}-1\right)\frac{2A}{BC}x(\xi)+A\left(2n(\xi)p+p^2\right)\\
&=x(\xi)\left[-e^{Cr(\xi)}\left(1-e^{-Cp}\right)\frac{2A}{C}+2Ap\right]+A\left(p-r(\xi)\right)^2-Ar(\xi)^2\\
&\geq \frac{2A}{C}x(\xi)\left[-e^{C/2}\left(1-e^{-Cp}\right)+Cp\right]+A\left(p-1/2\right)^2-A/4
\end{split}
\end{equation*}
Now the map $p\mapsto -e^{C/2}\left(1-e^{-Cp}\right)+Cp$ is increasing for $p\geq 1$ so that
\begin{equation}
\sum_{n=n(\xi)+2}^\infty e^{-G(n,\xi)+G(n(\xi),\xi)}\leq e^{-\frac{2A}{C}x(\xi)\left[2C-e^{C/2}\left(1-e^{-2C}\right)\right]}\sum_{p=1}^\infty e^{-A\left(p-1/2\right)^2+A/4}.
\end{equation}
    Finally thanks to \eqref{cond-C}  we have
$$
2C-e^{C/2}\left(1-e^{-2C}\right)>0 \text{ for $C>0$ small enough} 
$$
so that 
\begin{equation}
\sum_{p=n(\xi)+2}^\infty e^{-G(p,\xi)+G(n(\xi),\xi)}\leq e^{-\frac{2A}{C}x(\xi)\left[2C-e^{C/2}\left(1-e^{-2C}\right)\right]}\sum_{p=1}^\infty e^{-A\left(p-1/2\right)^2+A/4}\to 0\text{ as }\xi\to\infty,
\end{equation}
that is
\begin{equation}\label{esti1-bis}
     \sum_{p=n(\xi)+2}^\infty e^{-G(p,\xi)}=o\left(e^{-G(n(\xi),\xi)}\right)\text{ as }\xi\to\infty.
\end{equation}

Next we claim that we have
\begin{equation}\label{est2-bis}
\sum_{p=0}^{n(\xi)-2} e^{-G(p,\xi)}=o\left(e^{-G(n(\xi),\xi)}\right)\text{ as }\xi\to \infty.
\end{equation}
To see this, note that for $2\leq p\leq n(\xi)$ we have
\begin{equation*}
\begin{split}
G(n(\xi)-p,\xi)-G(n(\xi),\xi)&=Be^{-C(n(\xi)-p)}\xi+A(n(\xi)-p)^2-Be^{-Cn(\xi)}\xi-An(\xi)^2\\
&=B(e^{Cp}-1)e^{-Cx(\xi)}e^{Cr(\xi)}\xi+A[(n(\xi)-p)^2-n(\xi)^2]\\
&=B(e^{Cp}-1)e^{-Cx(\xi)}e^{Cr(\xi)}\xi-A[2px(\xi)-p^2-2pr(\xi)]\\
&=\frac{2A}{C}x(\xi)[e^{Cr(\xi)}e^{Cp}-1-Cp]+A(p^2-2pr(\xi))\\
&=\frac{2A}{C}x(\xi)[e^{-C/2}e^{Cp}-1-Cp]+A(p^2-2pr(\xi))\\
&\geq 2Ax(\xi)[e^{-C/2}(e^{C}-1)-C]+A(p^2-p).
\end{split}
\end{equation*}
Now the map $p\mapsto e^{-C/2}e^{Cp}-1-Cp$ is increasing for $p\geq 1$ so that
\begin{equation*}
\begin{split}
G(n(\xi)-p,\xi)-G(n(\xi),\xi)&\geq 2Ax(\xi)[e^{-C/2}(e^{2C}-1)-2C]+A(p^2-p).
\end{split}
\end{equation*}
    Thanks to \eqref{cond-C} we have
$$
e^{-C/2}(e^{2C}-1)-2C>0,
$$
so that we obtain
\begin{equation*}
\sum_{p=0}^{n(\xi)-2} e^{-G(p,\xi)+G(n(\xi),\xi)}\to 0\text{ as }\xi\to\infty,
\end{equation*}
and the result follows.
\end{proof}

Now let us prove that when $C$ is small enough, when the time becomes large, at most three variants can survive at the same time.
Our precise lemma reads as follows.
\begin{corollary}
Fix $C$ satisfying \eqref{cond-C} and let Assumption \ref{case1} be satisfied. 
Define for $t\gg 1$, $N(t)$ and $R(t)$ the closest integer part and the fractional part of $Y(t)$ the solution of 
$$
Y(t) = \frac{1}{C}W_0\left(\frac{BC^2}{2A}t\overline S(t)\right).
$$
Then the following holds true:
\begin{equation*}
    \sum_{p=0}^{N(t)-2} I_p(t)+\sum_{p=N(t)+2}^\infty I_p(t)=o\left(I_{N(t)}(t)\right)\text{ as }t\to\infty.
\end{equation*}
\end{corollary}

Note that using the notations introduced above (see \eqref{def-n-r}) we have
$$
N(t)=n\left(t\overline S(t)\right)\text{ and }R(t)=r\left(t\overline S(t)\right).
$$
\begin{remark}
The above corollary means that, at least for $C$ small enough, at most three variants can simultaneously survive in the large time, the variants $N(t)-1$, $N(t)$ and $N(t)+1$.
\end{remark}

\begin{proof}
The proof is a direct consequence of on estimates \eqref{esti1-bis} and \eqref{est2-bis}.
Indeed we have for all $n\geq 0$ and $t>0$
$$
\frac{I_n(t)}{I_{N(t)}(t)}=e^{-A(n^2-N(t)^2)-B(e^{-Cn}-e^{-CN(t)}))t\overline{S}(t)}=e^{-G(n, t\overline{S}(t))+G(N(t),t\overline{S}(t))} ,
$$
and the result follows from these two estimates.
\end{proof}

\begin{proof}[Proof of Claim \ref{claim:gaussian-exponential}.]
Now let us show that most of the time, only one variants can survive.
To see this, using the same notations as above, note that we have
\begin{align*}
\frac{I_{N-1}(t)}{I_{N(t)+1}(t)}&=e^{-A((N(t)-1)^2-(N(t)+1)^2)-B(e^{-C(N(t)-1)}-e^{-C(N(t)+1)}))t\overline{S}(t)}\\
&= e^{-A(-4N(t)-Be^{-CN(t)-CR(t)}e^{CR(t)}(e^{C}-e^{-C}))t\overline{S}(t)}\\
&=e^{4AN(t)-2A/C Y(t)e^{CR(t)}(e^{C}-e^{-C}))}\\
&=e^{2AN(t)\left(2-1/C e^{CR(t)}(e^{C}-e^{-C}))\right)+\mathcal O(1)},
\end{align*}
while 
\begin{align*}
\frac{I_{N-1}(t)}{I_{N(t)}(t)}&=e^{-A((N(t)-1)^2-N(t)^2)-B(e^{-C(N(t)-1)}-e^{-CN(t)}))t\overline{S}(t)}\\
&=e^{2AN(t)\left(1-1/C e^{CR(t)}(e^{C}-1))\right)+\mathcal O(1)},
\end{align*}
and
\begin{align*}
\frac{I_{N+1}(t)}{I_{N(t)}(t)}&=e^{-A((N(t)+1)^2-N(t)^2)-B(e^{-C(N(t)+1)}-e^{-CN(t)}))t\overline{S}(t)}\\
&=e^{-2AN(t)\left(1+1/C e^{CR(t)}(e^{-C}-1))\right)+\mathcal O(1)}.
\end{align*}
As a consequence, setting
$$
X_1(C)=\frac{1}{C}\ln \frac{C}{\sinh(C)},\;X_2(C)=\frac{1}{C}\ln \frac{C}{e^C-1}\text{ and } X_3(C)=-\frac{1}{C} \ln \frac{1-e^{-C}}{C},
$$
for all $\delta>0$, we have
\begin{equation}\label{lim1}
\lim_{\substack{t\to\infty\\ X_1(C)+\delta\leq R(t)}} \frac{I_{N-1}(t)}{I_{N(t)+1}(t)}=0,
\end{equation}
\begin{equation}\label{lim2}
\lim_{\substack{t\to\infty\\ R(t)\leq X_1(C)-\delta}} \frac{I_{N+1}(t)}{I_{N(t)-1}(t)}=0,
\end{equation}
and
\begin{equation}\label{lim3}
\lim_{\substack{t\to\infty\\ X_2(C)+\delta \leq R(t)}} \frac{I_{N-1}(t)}{I_{N(t)}(t)}=0\text{ and }\lim_{\substack{t\to\infty\\ X_3(C)-\delta \geq R(t)}} \frac{I_{N+1}(t)}{I_{N(t)}(t)}=0.
\end{equation}
To understand the meaning of the above limits, observe that
$$
X_1(C)=-C/6+\mathcal O(C^2),\;\; X_2(C)=-\frac{1}{2}+\mathcal O(C)\text{ and }X_3(C)=\frac{1}{2}+\mathcal O(C)\text{ as }C\to 0.
$$
As a consequence, using \eqref{lim3}, if $C$ is sufficiently small then for some $\delta>0$ small enough we have 
\begin{equation*}
\begin{split}
&I_{N(t)-1}(t)=o\left(I_{N(t)}(t)\right)\text{ for $t\gg 1$ and }-\frac{1}{2}+\delta\leq R(t)\leq \frac{1}{2},\\
&I_{N(t)+1}(t)=o\left(I_{N(t)}(t)\right)\text{ for $t\gg 1$ and }-\frac{1}{2}<R(t)\leq \frac{1}{2}-\delta.
\end{split}
\end{equation*}
This means that most of the time (when $R(t)$ is close to $-1/2$, only the variant $N(t)$ and $N(t)-1$ can survive.
Using \eqref{lim2}, if $C$ is sufficiently small then then for some $\delta>0$ small enough we have
$$
I_{N(t)+1}(t)=o\left(I_{N(t)-1}(t)\right)\text{ for $t\gg 1$ and }R(t)\in \left[-\frac{1}{2}+\delta,-\delta\right].
$$
For $C>0$ is small enough, there exists $\delta>0$ small enough such that the following picture for the survival of the variants holds
\begin{equation*}
\sum_{n=0}^\infty I_n(t)\sim 
\begin{cases} 
I_{N(t)-1}(t)+I_{N(t)}(t)\text{ for $t\gg 1$ and $-1/2\leq R(t)\leq -1/2+\delta$},\\
I_{N(t)}\text{ for $t\gg 1$ and $-1/2+\delta \leq R(t)\leq1/2-\delta$},\\
I_{N(t)}(t)+I_{N(t)+1}(t)\text{ for $t\gg 1$ and $1/2-\delta<R(t)\leq 1/2$}.
\end{cases}
\end{equation*}
This proves Claim \ref{claim:gaussian-exponential}.
\end{proof}

}

\subsection{Replacement dynamics 2: bivalent $\gamma_n$}
\label{sec:bivalent-gamma}
{
We place ourselves in the case when the sequence $\gamma_n$ takes only two values, an more precisely we assume that there exist two constants $0<\gamma_1<\gamma_2$ such that $\gamma_{2n+1}\equiv \gamma_1>0$ and $\gamma_{2n}\equiv \gamma_2$ for all $n\in\mathbb{N}$.
Our goal is to give examples of possible behaviors when Assumption \ref{as:disintegration} is not satisfied; in particular, we will not assume that the total mass converges. 

Recalling that $\alpha^*:=\sup_{n\in\mathbb{N}}\frac{\beta_n}{\gamma_n}$, we can write the total mass as  
\begin{align}
	\nonumber \sum_{n=0}^{+\infty}I_n(t) &= \sum_{n=0}^{+\infty}I_n^0e^{\beta_n \overline{S}(s)-\gamma_n t} = \sum_{n=0}^{+\infty} I_n^0e^{\left(\frac{\beta_n}{\gamma_n}-\alpha^*\right)\gamma_n t\overline{S}(t) + \alpha^*(\gamma_n-\gamma^*)t\overline{S}(t)+\alpha^*\gamma^*t\overline{S}(t) - (\gamma_n-\gamma^*)t-\gamma^*t}\\
	\nonumber&=e^{\gamma^*t(\alpha^*\overline{S}(t)-1)}\sum_{n=0}^{+\infty} I_n^0 e^{\gamma_n\left(\frac{\beta_n}{\gamma_n}-\alpha^*\right)t\overline{S}(t)+(\gamma_n-\gamma^*)(\alpha^*\overline{S}(t)-1)t}\\
	\nonumber&=e^{\gamma^*t(\alpha^*\overline{S}(t)-1)}\sum_{n=0}^{+\infty} I_n^0 e^{-\gamma_n\left(\alpha^*-\frac{\beta_n}{\gamma_n}\right)t\overline{S}(t)-\alpha^*(\gamma^*-\gamma_n)\left(\overline{S}(t)-\frac{1}{\alpha^*}\right)t} \\ 
	&=e^{\gamma_1t(\alpha^*\overline{S}(t)-1)} \mathcal{F}\big(t, \overline{S}(t)\big), \label{eq:bivalent-somme}
\end{align} 
where 
\begin{equation}\label{eq:bivalentcalF}
    \mathcal{F}\big(t, \overline{S}(t)\big):=\sum_{n=0}^{+\infty} I_n^0 e^{-\gamma_n\left(\alpha^*-\frac{\beta_n}{\gamma_n}\right)t\overline{S}(t)-\alpha^*(\gamma^*-\gamma_n)\left(\overline{S}(t)-\frac{1}{\alpha^*}\right)t}. 
\end{equation}
Since $\gamma_{2n+1}=\gamma_1$ and $\gamma_{2n}=\gamma_2$, we set
\begin{equation*}
    \beta^1_{n}:= \beta_{2n+1}, \quad I^{0, 1}_n:= I^0_{2n+1}, \quad \beta^{2}_n:=\beta_{2n},\quad  \text{ and } I^{0,2}_{2n}:=I^0_{2n}, 
\end{equation*}
so we can rewrite \eqref{eq:bivalentcalF} as 
\begin{align}
    \mathcal{F}\big(t, \overline{S}(t)\big) & = \sum_{n=0}^{+\infty} I_n^{0,1} e^{-\gamma_1\left(\alpha^*-\frac{\beta^1_n}{\gamma_1}\right)t\overline{S}(t)} + \sum_{n=0}^{+\infty} I_n^{0,2} e^{-\gamma_2\left(\alpha^*-\frac{\beta^2_n}{\gamma_2}\right)t\overline{S}(t)-(\gamma_1-\gamma_2)\left(\alpha^*\overline{S}(t)-1\right)t}\nonumber\\
    &=F_1\big(t\overline{S}(t)\big)+ e^{-(\gamma_1-\gamma_2)(\alpha^*\overline{S}(t)-1)t}F_2\big(t\overline{S}(t)\big),\label{eq:equationcalF1}
\end{align}
where 
\begin{equation}\label{eq:exbivalentF1F2}
    F_1(\xi):= \sum_{n=0}^{+\infty} I_n^{0,1} e^{-\gamma_1\left(\alpha^*-\frac{\beta^1_n}{\gamma_1}\right)\xi} 
    \text{ and } F_2(\xi):=\sum_{n=0}^{+\infty} I_n^{0,2} e^{-\gamma_2\left(\alpha^*-\frac{\beta^2_n}{\gamma_2}\right)\xi}.
\end{equation}

To go a bit further, we notice that taking the logarithm of  \eqref{eq:bivalent-somme} leads to
\begin{equation*}
    \gamma_1\big(\alpha^*\overline{S}(t)-1\big)t = -\ln \mathcal{F}\big(t, \overline{S}(t)\big) + \ln\mathcal{I}(t), 
\end{equation*}
so \eqref{eq:equationcalF1} becomes
\begin{equation}\label{eq:bivalent-implicit-F}
    \mathcal{F}\big(t, \overline{S}(t)\big) = F_1\big(t\overline{S}(t)\big) + \left(\dfrac{\mathcal{F}(t, \overline{S}(t)\big)}{\mathcal{I}(t)}\right)^{\frac{\gamma_1-\gamma_2}{\gamma_1}} F_2\big(t\overline{S}(t)\big).
\end{equation}
\textbf{In the sequel we will assume $\gamma_1 = 2\gamma_2$.} In that case, \eqref{eq:bivalent-implicit-F} is a second-order polynomial equation in $\sqrt{\mathcal{F}\big(t, \overline{S}(t)\big)}$ which can be inverted to give the following expression of $\mathcal{F}$:
\begin{equation}\label{eq:bivalent-explicit-F}
	\mathcal{F}\big(t, \overline{S}(t)\big) =\frac{1}{4\mathcal{I}(t)} \left(F_2\big(t\overline{S}(t)\big) + \sqrt{F_2\big(t\overline{S}(t)\big)^2 + 4F_1\big(t\overline{S}(t)\big)\mathcal{I}(t)}\right)^2.
\end{equation}
}

\subsubsection{Bivalent example 1: unexpected selection}
{First let us explain the title of the subsection. Given Proposition \ref{prop:convergence-mass}, we can prove that the phenotype that is eventually selected by competition (after the selection for fitness) is the one that maximizes $\gamma_n$ (in our case, $\gamma_1$). In this section we will prove that, for some carefully chosen initial data, selection can make all pathogen expressing the $\gamma_1$ phenotype disappear.  This strikes us as an unexpected result. We work under the following assumption. 
{
\begin{assumption}\label{as:unexpected-selection}
    We assume that $\gamma_1 = 2\gamma_2$,  $\beta_n^1=\gamma_1\alpha^* - B_1e^{-C_1n}$, $I_n^{0,1}=e^{-A_1n^2}$, $\beta_n^2 = \gamma_2\alpha^* - B_2e^{-C_2n}$ and $I^{0,2}_n= e^{-A_2 n}$.
\end{assumption}
}
Our claim is as follows. 
\begin{claim}\label{claim:unexpected-selection}
    Under Assumption \ref{as:unexpected-selection}, we have 
    \begin{equation*}
	\lim_{t\to+\infty}\sum_{n=0}^{+\infty} I_n^1(t) = 0 \text{ and } \lim_{t\to+\infty} \sum_{n=0}^{+\infty}I_n^2(t)=\dfrac{\theta}{\gamma_2\alpha^*}\big(\mathcal{R}_0-1\big) >0, 
    \end{equation*}
    wherein $I_n^1(t) = I_{2n+1}(t) $ and $I^2_{n}(t) = I_{2n}(t)$.
\end{claim}
\begin{figure}[H]
    \centering
    \includegraphics{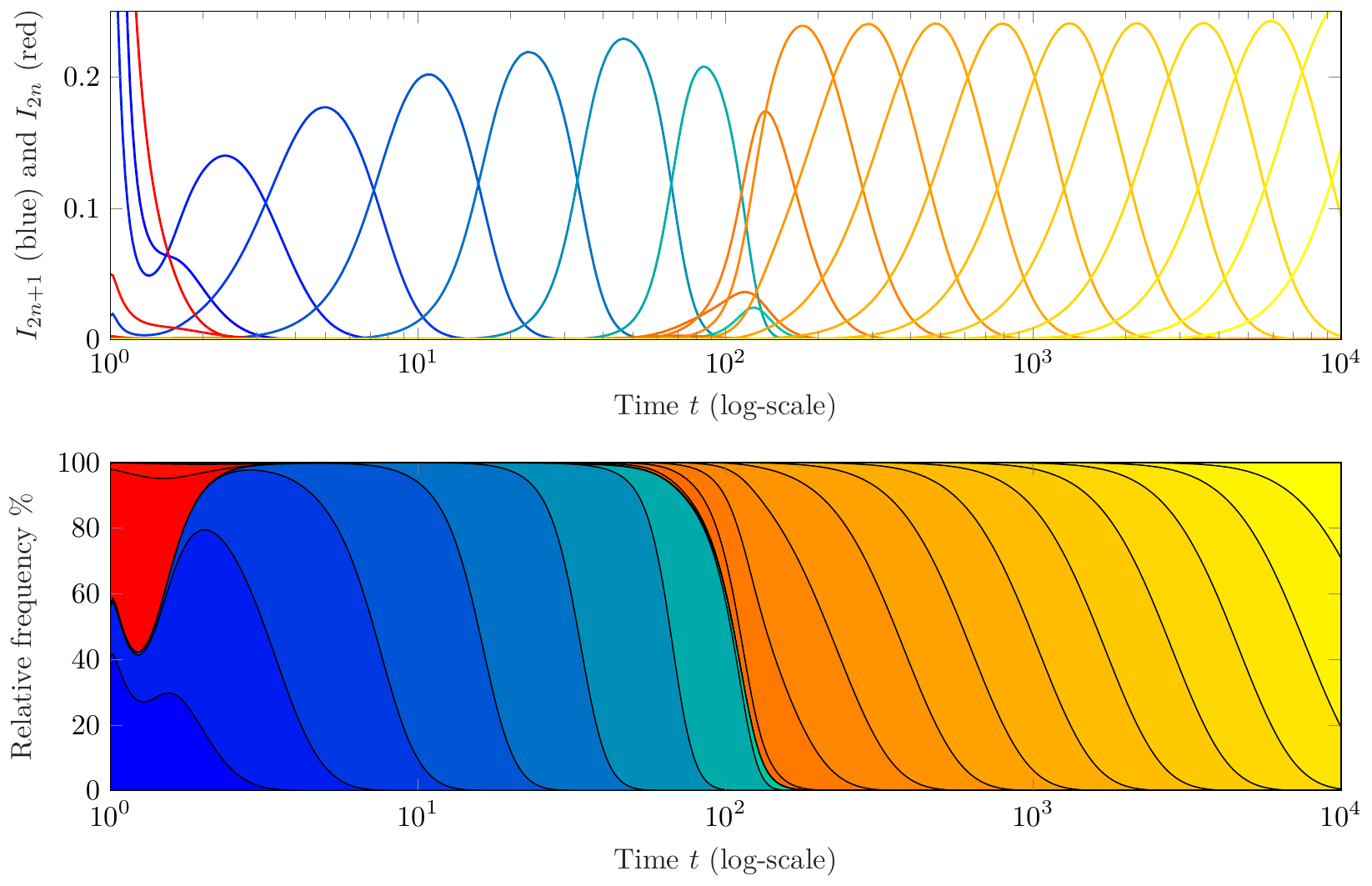}
	\caption{Plots of the solution $I_n(t)$ in the Bivalent case 1. The family $I_n(t)$ is divided in two groups:$I_{2n+1}(t)=I^1_n(t)$ associated with $\gamma_1$ which are plotted in a family of cold colors (blue to green), and $I_{2n}(t)=I^2_n(t)$ associated with $\gamma_2$ which are plotted in a family of warm colors (red to yellow). Both figures (top and bottom) suggest that the $\gamma_1$ family becomes prevalent at first, but the $\gamma_2$ family eventually wins the competition (after $t\approx 10^2$). \textbf{Top figure:} Value of the density of variants $I_n(t)$. \textbf{Bottom figure:} Relative frequencies of the variants as a function of time. \textbf{Parameters:} $\Lambda=10$, $\theta=10$, $S_0=1$, $\alpha^*=2$,  $\gamma_1=10$, $\gamma_2=10$, $A_1=1$, $A_2=3$, $B_1=B_2=10$, $C_1=C_2=\frac{1}{2}$. We used a total of $N_1=10$ variants for the family $1$ and $N_2=20$ variants for the family 2. For interpretation of the colors in the figure(s), the reader is referred to the online version of this article.}\label{Fig4}
\end{figure}
\begin{proof}[Proof of Claim \ref{claim:unexpected-selection}]
    We rewrite \eqref{eq:bivalent-explicit-F} as:
    \begin{equation*}
	    \mathcal{F}\big(t, \overline{S}(t)\big) = \frac{F_2\big(t\overline{S}(t)\big)^2}{4\mathcal{I}(t)}\left(1+\sqrt{1+4\mathcal{I}(t)\frac{F_1\big(t\overline{S}(t)\big)}{F_2\big(t\overline{S}(t)\big)^2}}\right)^2.
    \end{equation*}
    We know that $\mathcal{I}(t)\asymp 1$, that $F_1\big(t\overline{S}(t)\big)=\mathcal{O}\left(e^{-\frac{A_1}{2C_1}\ln(t\overline{S}(t))^2}\right)$ by Lemma \ref{lem:Gaussian-est}  and that $F_2\big(t\overline{S}(t)\big)\asymp \big(t\overline{S}(t)\big)^{-\frac{A_2}{C_2}}$ by Claim \ref{claim:asymp-F-example-1}. In particular, we obtain
    \begin{equation*}
	4\mathcal{I}(t)\frac{F_1\big(t\overline{S}(t)\big)}{F_2\big(t\overline{S}(t)\big)^2}=\mathcal{O}\left(e^{-\frac{A_1}{2C_1}\ln(t\overline{S}(t))^2 + \frac{A_2}{C_2}\ln(t\overline{S}(t))}\right) \xrightarrow[t\to+\infty]{} 0,
    \end{equation*}
    so that
    \begin{equation*}
	    \mathcal{F}\big(t, \overline{S}(t)\big) = \frac{F_2\big(t\overline{S}(t)\big)^2}{4\mathcal{I}(t)}\big(1+o(1)\big) 
    \end{equation*}
    and by \eqref{eq:bivalent-somme}
    \begin{equation*}
	\gamma_1(\alpha^*\overline{S}(t)-1)t = -2\ln\big(F_2(t\overline{S}(t))\big) + \mathcal{O}(1).
    \end{equation*}
    Finally we obtain
    \begin{equation*}
	\sum_{n=0}^{+\infty} I^1_n(t) = e^{\gamma_1(\alpha^*\overline{S}(t)-1)t}F_1\big(t\overline{S}(t)\big) = \mathcal{O}\left(\dfrac{F_1\big(t\overline{S}(t)\big)}{F_2\big(t\overline{S}(t)\big)^2}\right) \xrightarrow[t\to+\infty]{}0.
    \end{equation*}

    Next we show the convergence of $\sum I^2_n(t)$ to a positive constant. We know that $\sum I^1_n(t)\to 0$ and, by Theorem \ref{thm:discrete}, $S(t)\to \frac{1}{\alpha^*}$ and $S'(t)\to 0$. Using the first equation in \eqref{eq:main-a} we have:
    \begin{align*}
	-\dfrac{S'(t)-\Lambda + \theta S(t)}{S(t)} &=\sum_{n=0}^{+\infty}\beta_n I_n(t) = \sum_{n=0}^{+\infty}\beta^1_n I^1_n(t)+\sum_{n=0}^{+\infty} \beta^2_nI^2_n(t) = \sum_{n=0}^{+\infty} \beta^2_n I^2_n(t)+\mathcal{O}\left(\sum_{n=0}^{+\infty}I^1_n(t)\right) \\ 
	&=\sum_{n=0}^{+\infty} \gamma_2\left(\dfrac{\beta^2_n}{\gamma_2}-\alpha^*\right)I^2_n(t) + \sum_{n=0}^{+\infty} \gamma_2\alpha^* I^2_n(t) + o(1) =  \gamma_2\alpha^*\sum_{n=0}^{+\infty}I^2_n(t) + o(1), 
    \end{align*}
    where we used the fact that $\sum \left(\frac{\beta_n}{\gamma_2}-\alpha^*\right)I^2_n(t)\to 0$, which will justify below. Admitting this fact temporarily, let us finish the argument. We have now
    \begin{equation*}
	\sum_{n=0}^{+\infty} I^2_n(t)  = -\dfrac{1}{\gamma_2\alpha^*}\left(\dfrac{S'(t)-\Lambda+\theta S(t)}{S(t)}\right) + o(1) = \dfrac{\theta}{\gamma_2\alpha^*}\left(\dfrac{\Lambda}{\theta}\alpha^*-1\right) + o(1) = \dfrac{\theta}{\gamma_2\alpha^*}\left(\mathcal{R}_0-1\right)+o(1), 
    \end{equation*}
    which is exactly the second part of Claim \ref{claim:unexpected-selection}.

    There remains to show that $\sum \left(\frac{\beta_n}{\gamma_2}-\alpha^*\right)I^2_n(t)\to 0$ as $t\to+\infty$. Fix $\varepsilon>0$ arbitrarily. Let $I^\infty:=\frac{2\Lambda}{\min(\theta, \gamma_0)}$ so that by Lemma \ref{lem:bounds-nomut} we have $\sum I^2_n(t)\leq I^\infty$ for $t$ sufficiently large. Let $\mathcal{N}_\varepsilon$ be the set of indices defined by 
    \begin{equation*}
	\mathcal{N}_\varepsilon:=\left\{n\in\mathbb{N}\,:\,\left|\frac{\beta^2_n}{\gamma_2}-\alpha^*\right|\leq \frac{\varepsilon}{2I^\infty}\right\}. 
    \end{equation*}
    Then
    \begin{equation*}
	\sum_{n\in\mathcal{N}_\varepsilon}\left(\frac{\beta_n}{\gamma_2}-\alpha^*\right) I^2_n(t) \leq \frac{\varepsilon}{2I^\infty} \sum_{n=0}^{+\infty}I^2_n(t) \leq \dfrac{\varepsilon}{2}, 
    \end{equation*} 
    for $t$ sufficiently large. On the other hand, for $n\in\mathcal{N}_\varepsilon^c:=\mathbb{N}\backslash\mathcal{N}_\varepsilon$, we have $\alpha^*-\frac{\beta_n}{\gamma_2}\geq \frac{\varepsilon}{2I^\infty}$ so that
    \begin{equation*}
	I^2_n(t) = I^{0,2}_n e^{\gamma_2\left(\frac{\beta_n}{\gamma_2}-\alpha^*\right)\overline{S}(t)t + \gamma_2\left(\alpha^*\overline{S}(t)-1\right)t} \leq I_n^{0,2}e^{-\frac{\gamma_2\varepsilon}{2I^\infty}\overline{S}(t)t + \gamma_2(\alpha^*\overline{S}(t)-1)t} \leq I^{0,2}_n \exp\left(-\dfrac{\gamma_2\varepsilon}{8I^\infty\alpha^*}t\right), 
    \end{equation*}
    whenever $\overline{S}(t)\geq \frac{1}{2\alpha^*}$ and $|\alpha^*\overline{S}(t)-1| \leq \dfrac{\varepsilon}{8I^\infty\alpha^*}$, which is true for $t$ sufficiently large. Thus 
    \begin{equation*}
	\sum_{n\in\mathcal{N}^c_\varepsilon}I^2_n(t) \leq e^{-\frac{\gamma_2\varepsilon}{8I^\infty\alpha^*}t} \sum_{n\in\mathcal{N}^c_\varepsilon} I^{0,2}_n \xrightarrow[t\to+\infty]{}0, 
    \end{equation*}
    and finally 
    \begin{equation*}
	\sum_{n=0}^{+\infty}I^2_n(t) = \sum_{n\in\mathcal{N}_\varepsilon}I^2_n(t)+\sum_{n\in\mathcal{N}^c_\varepsilon} I^2_n(t) \leq \varepsilon, 
    \end{equation*}
    for $t$ sufficiently large. This finishes the proof of Claim \ref{claim:unexpected-selection}.
\end{proof}
}

\subsubsection{Bivalent example 2: alternating persistence}
\label{sec:alternating-persistence}

{In this subsection we provide an example which shows a very particular asymptotic behavior: the types $\gamma_1$ and $\gamma_2$ are \textit{both} asymptotically persistent as $t\to+\infty$, and become alternatively prevalent in the population. This causes the total mass of infected $\sum I_n(t)$ to fluctuate between two distinct values. This shows, in particular, that an additional assumption (like, for instance, Assumption \ref{as:disintegration}) is really necessary to obtain the asymptotic behavior of the mass and that the conclusions of Theorem \ref{thm:discrete} are, in some sense, sharp.

We work under the following assumption. 
\begin{assumption}\label{as:alternating-persistence}
	We assume that $\gamma_1 = 2\gamma_2$,  $\beta_n^1=\gamma_1\alpha^* - Be^{-Cn}$, $I_n^{0,1}=I_0^1e^{-An^2}$, $\beta_n^2 = \gamma_2\alpha^* - \frac{B}{2}e^{-\frac{C}{3}n}$ and $I^{0,2}_n=I_0^2e^{-\frac{A}{18}n^2}$ for some positive constants $A>0$. $B>0$ and $C>0$.
\end{assumption}
Our claim is as follows. 
\begin{claim}\label{claim:alternating-persistence}
    Under Assumption \ref{as:alternating-persistence}, there exist  two sequences $t^1_k\to +\infty$ and $t^2_k\to+\infty$ as $k\to+\infty$, such that 
    \begin{equation*}
	\lim_{k\to+\infty}\sum_{n=0}^{+\infty} I_n^1(t^1_k) = \dfrac{\theta}{\gamma_1\alpha^*}\big(\mathcal{R}_0-1\big) >0\text{ and } \lim_{t\to+\infty} \sum_{n=0}^{+\infty}I_n^2(t^1_k)=0, 
    \end{equation*}
    and 
    \begin{equation*}
	\lim_{k\to+\infty}\sum_{n=0}^{+\infty} I_n^1(t^2_k) =0\text{ and } \lim_{k\to+\infty} \sum_{n=0}^{+\infty}I_n^2(t^2_k)=\dfrac{\theta}{\gamma_2\alpha^*}\big(\mathcal{R}_0-1\big) >0, 
    \end{equation*}
    wherein $I_n^1(t) = I_{2n+1}(t) $ and $I^2_{n}(t) = I_{2n}(t)$.
\end{claim}
\begin{figure}[H]
    \centering
    \includegraphics{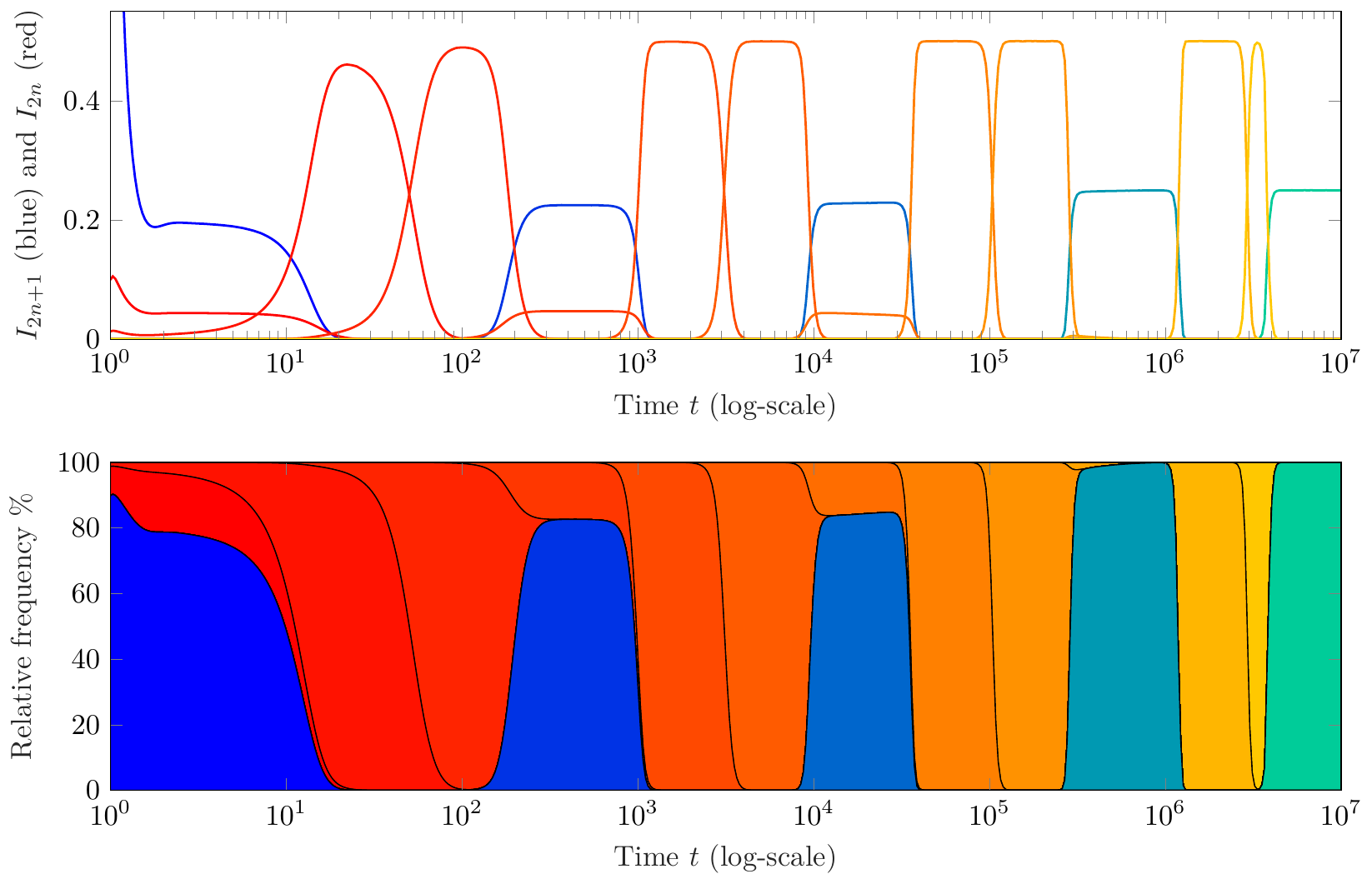}
	\caption{Plots of the solution $I_n(t)$ under Assumption \ref{as:alternating-persistence}. The family $I_n(t)$ is divided in two groups:$I_{2n+1}(t)=I^1_n(t)$ associated with $\gamma_1$ which are plotted in a family of cold colors (blue to green), and $I_{2n}(t)=I^2_n(t)$ associated with $\gamma_2$ which are plotted in a family of warm colors (red to yellow). The top figure suggests that the total  pathogen population fluctuates in time.  The bottom figure suggests that the $\gamma_1$ and $\gamma_2$ families become alternatively prevalent. \textbf{Top figure:} Value of the density of variants $I_n(t)$. \textbf{Bottom figure:} Relative frequencies of the variants as a function of time. The family associated with $\gamma_1$ (blue to green) and the family associated with $\gamma_2$ (red to yellow) become alternatively dominant. \textbf{Parameters:} $\Lambda=5$, $\theta=5$, $S_0=1$, $\alpha^*=2$,  $\gamma_1=10$, $\gamma_2=5$, $I_0^1=1$, $I_0^2=0.1$, $A=36$, $B=2$, $C=3$. We used a total of $N_1=6$ variants for the family $1$ and $N_2=15$ variants for the family 2. For interpretation of the colors in the figure(s), the reader is referred to the online version of this article.} \label{Fig5}
\end{figure}
\begin{proof}[Proof of Claim \ref{claim:alternating-persistence}]
	Recalling \eqref{eq:exbivalentF1F2} and Lemma \ref{lem:Gaussian-est}, we have in this case: 
	\begin{align*}
		F_1(e^{-C r_1(\xi)}\xi)  &\sim  e^{-An_1(\xi)^2 -\frac{2A}{C}x_1(\xi)}, \\
		F_2(e^{-\frac{C}{3} r_2(\xi)}\xi)  &\sim  e^{-\frac{A}{18}n_2(\xi)^2 -\frac{A}{3C}x_2(\xi)}, 
	\end{align*}
	where
	\begin{equation*}
		x_1(\xi):= n_1(\xi)+r_1(\xi) = \frac{1}{C}W_0\left(\dfrac{BC^2}{2A}\xi\right)\text{ and }
		x_2(\xi):= n_2(\xi)+r_2(\xi) = \frac{3}{C}W_0\left(\dfrac{BC^2}{2A}\xi\right),
	\end{equation*}
	$n_1(\xi), n_2(\xi)\in\mathbb{N}$ and $r_1(\xi), r_2(\xi)\in\left[-\frac{1}{2}, \frac{1}{2}\right)$.
	In other words, with $x(\xi):= \frac{1}{C}W_0\left(\dfrac{BC^2}{2A}\xi\right)$, we have
	\begin{align*}
		n_1(\xi) &= \left\lfloor x(\xi)+\frac{1}{2} \right\rfloor,  & n_2(\xi)&=\left\lfloor 3x(\xi)+\frac{1}{2}\right\rfloor, \\ 
		r_1(\xi)&= x(\xi)-n_1(\xi)\in\left(-\frac{1}{2}, \frac{1}{2}\right),  &
		r_2(\xi)&= 3x(\xi)-n_2(\xi)\in\left(-\frac{1}{2}, \frac{1}{2}\right).
	\end{align*}
	Let $\xi^1_n$ and $\xi^2_n$ be  defined by the relation 
	\begin{equation*}
		\xi^1_n:= \dfrac{2A}{BC}\left(n-\frac{1}{3}\right) e^{C\left(n-\frac{1}{3}\right)} \text{ and }
		\xi^2_n:= \dfrac{2A}{BC}\left(n+\frac{1}{3}\right) e^{C\left(n+\frac{1}{3}\right)}, 
	\end{equation*}
	so that we have, by definition,
	\begin{align*}
		x(\xi^1_n) &= n-\frac{1}{3}, & n_1(\xi^1_n) &=n, & r_1(\xi^1_n) &=-\frac{1}{3}, & n_2(\xi^1_n) &= 3n-1, & r_2(\xi^1_n) = 0, \\ 
		x(\xi^2_n) &= n+\frac{1}{3}, & n_1(\xi^2_n) &=n, & r_1(\xi^2_n) &=+\frac{1}{3}, & n_2(\xi^2_n) &= 3n+1, & r_2(\xi^2_n) = 0.
	\end{align*}
	Since $W_0(X)= \ln(X)-\ln\ln X+o(1)$ (see \cite[(4.19) p.349]{Corless-etal-96}), we have for $n$ sufficiently large:
	\begin{align*}
		x\left( e^{C r_1(\xi^1_n)}\xi^1_n\right) &=  \frac{1}{C}\ln\left(\frac{BC^2}{2A}e^{C r_1(\xi^1_n)}\xi^1_n\right)-\frac{1}{C}\ln \ln \left(\frac{BC^2}{2A}e^{Cr_1(\xi^1_n)}\xi^1_n\right) + o(1)\\
		&= r_1(\xi^1_n)+\frac{1}{C}\ln\left(\frac{BC^2}{2A}\xi\right) - \frac{1}{C}\ln\ln\left(\frac{BC^2}{2A} \xi \right)- \frac{1}{C}\ln\left(1+\frac{C r_1(\xi^1_n)}{\ln\left(\frac{BC^2}{2A}\xi\right)}\right)+o(1) \\ 
		&=x(\xi^1_n)+r_1(\xi^1_n)+o(1) = n-\frac{1}{3}-\frac{1}{3} +o(1) = n-\frac{2}{3}+o(1), 
	\end{align*}
	thus for $n$ sufficiently large, 
	\begin{equation*}
		n_1\left(e^{Cr_1(\xi^1_n)}\xi^1_n\right) = \left\lfloor x\left(\xi e^{Cr_1(\xi^1_n)}\right)+\frac{1}{2}\right\rfloor = n-1.
	\end{equation*}
	Similarly, for $n$ sufficiently large, we have
	\begin{equation*}
		x\left(\xi^2_ne^{Cr_1(\xi^2_n)}\right)  = n+\frac{2}{3}+o(1) \text{ and }n_1\left(e^{Cr_1(\xi^2_n)}\xi^2_n\right) = \left\lfloor x\left(\xi e^{Cr_1(\xi^2_n)}\right)+\frac{1}{2}\right\rfloor = n+1.
	\end{equation*}
	Thus, for $n$ sufficiently large, we have
	\begin{align*}
		\frac{F_1(\xi^1_n)}{F_2(\xi^1_n)^2} &\sim \exp\left[-A n_1\left(e^{Cr_1(\xi^1_n)}\xi^1_n\right)^2 + A\left(\frac{n_2(\xi^1_n)}{3}\right)^2 - \frac{2A}{C}x_1\left(\xi^1_n e^{Cr_1(\xi^1_n)}\right) + \frac{2A}{3C}x_2(\xi^1_n) \right] \\
		&= \exp\left[A\left(\frac{n_2(\xi^1_n)}{3}-n_1\left(e^{Cr_1(\xi^1_n)}\xi^1_n\right)\right)\left(\frac{n_2(\xi^1_n)}{3}+n_1\left(e^{Cr_1(\xi^1_n)}\xi^1_n\right)\right) + \frac{2A}{C}\left(\frac{x_2(\xi^1_n)}{3} - x_1\left(e^{Cr_1(\xi^1_n)}\xi^1_n\right)\right)  \right] \\
		&=\exp\left[ A\left(n-\frac{1}{3}-(n-1)\right)\left(n-\frac{1}{3}+n-1\right) + \frac{2A}{C} \left(n-\frac{1}{3}+\frac{r_2(\xi^1_n)}{3} - \left(n-1+r_1\left(\xi^1_n e^{Cr_1(\xi)}\right)\right)\right)\right] \\ 
		& = \exp\left[ \frac{4A}{3}\left(n-\frac{2}{3}\right) +\mathcal{O}(1)\right]\xrightarrow[n\to+\infty]{}+\infty,
	\end{align*}
	and similarly, 
	\begin{align*}
		\frac{F_1(\xi^2_n)}{F_2(\xi^2_n)^2} &\sim \exp\left[-A n_1\left(e^{Cr_1(\xi^2_n)}\xi^2_n\right)^2 + A\left(\frac{n_2(\xi^2_n)}{3}\right)^2 - \frac{2A}{C}x_1\left(\xi^2_n e^{Cr_1(\xi^2_n)}\right) + \frac{2A}{3C}x_2(\xi^2_n) \right] \\
		&= \exp\left[A\left(\frac{n_2(\xi^2_n)}{3}-n_1\left(e^{Cr_1(\xi^2_n)}\xi^2_n\right)\right)\left(\frac{n_2(\xi^2_n)}{3}+n_1\left(e^{Cr_1(\xi^2_n)}\xi^2_n\right)\right) + \frac{2A}{C}\left(\frac{x_2(\xi^2_n)}{3} - x_1\left(e^{Cr_1(\xi^2_n)}\xi^2_n\right)\right)  \right] \\
		&=\exp\left[ A\left(n+\frac{1}{3}-(n+1)\right)\left(n+\frac{1}{3}+n+1\right) + \frac{2A}{C} \left(n+\frac{1}{3}+\frac{r_2(\xi^2_n)}{3} - \left(n+1+r_1\left(\xi^2_n e^{Cr_1(\xi)}\right)\right)\right)\right] \\ 
		& = \exp\left[ -\frac{4A}{3}\left(n+\frac{2}{3}\right) +\mathcal{O}(1)\right]\xrightarrow[n\to+\infty]{}0.
	\end{align*}

	Now we conclude the proof. Let $t^1_n$ be a sequence of times such that $t^1_n\overline{S}(t^1_n) = \xi^1_n$. Clearly $t^1_n\to+\infty$ as $n\to+\infty$. Then, recalling \eqref{eq:bivalent-explicit-F}, we have:
	\begin{align*}
		\mathcal{F}(t^1_n, \overline{S}(t^1_n) )&=\frac{F_1(t^1_n\overline{S}(t^1_n))}{\mathcal{I}(t^1_n)} \left(\dfrac{F_2(t^1_n\overline{S}(t^1_n))}{\sqrt{F_1(t^1_n\overline{S}(t^1_n))}}+\sqrt{\dfrac{\mathcal{I}(t^1_n)F_2(t^1_n\overline{S}(t^1_n))^2}{4F_1(t^1_n\overline{S}(t^1_n))}+1}\right)^2 = \frac{F_1(t^1_n\overline{S}(t^1_n))}{\mathcal{I}(t^1_n)}\big(1+o(1)\big), 
	\end{align*}
	and by \eqref{eq:bivalent-somme}, 
	\begin{equation*}
		\gamma_1t^1_n\left(\alpha^*\overline{S}(t^1_n)-1\right) = -\ln F_1\big((t^1_n\overline{S}(t^1_n)\big)+\mathcal{O}(1). 
	\end{equation*}
	Finally, 
	\begin{align*}
		\sum_{n=0}^{+\infty} I_n^2(t^1_n) &= I_0^2e^{\frac{\gamma_2}{\gamma_1}\gamma_1(\alpha^*\overline{S}(t^1_n)-1)} F_2\big(t^1_n\overline{S}(t^1_n)\big) =  I_0^2e^{\frac{1}{2}\left(-\ln F_1(t^1_n\overline{S}(t^1_n))+\mathcal{O}(1)\right)} F_2\big(t^1_n\overline{S}(t^1_n)\big) \\ 
		&=\mathcal{O}\left(\dfrac{F_2(\xi^1_n)}{\sqrt{F_1(\xi^1_n)}}\right) \xrightarrow[n\to+\infty]{} 0, 
	\end{align*}
	and by using \eqref{eq:main-a} together with Theorem \ref{thm:discrete} we obtain
	\begin{equation*}
		\lim_{n\to+\infty} \sum_{n=0}^{+\infty} I^1_n(t^1_n) = \dfrac{\theta}{\gamma_1\alpha^*}\left(\mathcal{R}_0-1\right),
	\end{equation*}
	as we did in the proof of Claim \ref{claim:unexpected-selection}. The first part of the Claim \ref{claim:alternating-persistence} is proved. Proceeding similarly, let $t^2_n$ be a sequence of times such that $t^2_n\overline{S}(t^2_n) = \xi^1_n$. Clearly $t^2_n\to+\infty$ as $n\to+\infty$. We have 
	\begin{align*}
		\mathcal{F}(t^2_n, \overline{S}(t^2_n) )&=\frac{F_2(t^2_n\overline{S}(t^2_n))}{4} \left(1+ \sqrt{1+\dfrac{4F_1(t^2_n\overline{S}(t^2_n))}{\mathcal{I}(t^2_n)F_2(t^2_n\overline{S}(t^2_n))^2}}\right)^2 = \frac{F_2(t^2_n\overline{S}(t^2_n))^2}{4}\big(1+o(1)\big), 
	\end{align*}
	so
	\begin{equation*}
		\gamma_1t^2_n\left(\alpha^*\overline{S}(t^2_n)-1\right) = -2\ln F_2\big(t^2_n\overline{S}(t^2_n)\big)+\mathcal{O}(1), 
	\end{equation*}
	and
	\begin{align*}
		\sum_{n=0}^{+\infty} I_n^1(t^2_n) &=I_0^1 e^{\gamma_1(\alpha^*\overline{S}(t^2_n)-1)} F_1\big(t^2_n\overline{S}(t^2_n)\big) = I_0^1 e^{-2\ln F_1(t^2_n\overline{S}(t^2_n))+\mathcal{O}(1)} F_1\big(t^2_n\overline{S}(t^2_n)\big) 
		=\mathcal{O}\left(\dfrac{{F_1(\xi^2_n)}}{F_2(\xi^2_n)^2}\right) \xrightarrow[n\to+\infty]{} 0, 
	\end{align*}
	so finally 
	\begin{equation*}
		\lim_{n\to+\infty} \sum_{n=0}^{+\infty} I^2_n(t^2_n) = \dfrac{\theta}{\gamma_2\alpha^*}\left(\mathcal{R}_0-1\right).
	\end{equation*}
	This is the second part of  Claim \ref{claim:alternating-persistence}, and the proof is finished.
\end{proof}

}
\section{Discussion}
\label{sec:discussion}

It is a classical result in evolutionary epidemiology \parencite{Ewald-1983, Ali-Hur-Mid-Van-09} known as the `trade-off hypothesis' that evolution favors variants that maximize the basic reproductive number $\mathcal{R}_0$; here we show once again the robustness of this prediction by considering an infinite number of variants competing for the hosts.  Yet, we also go beyond the standard prediction and show that a complexity persists in the asymptotic behavior of the epidemic even in our {simplistic} model, with many possible {outcomes} ranging from the simple convergence to a global equilibrium (the case of finite system as in \textcite{Hsu-78}, or case \ref{item:persistence-convergence} of Theorem \ref{thm:discrete}) to an eternal transient state (the example given in Section \ref{sec:alternating-persistence}).

{Our examples in section \ref{sec:examples} shed light on the variety of dynamics that can be observed for the diversity of pathogen variants in the host population. It is likely to be dependent not only on the distribution of the parameters $\gamma_n$ and $\alpha_n=\frac{\beta_n}{\gamma_n}$, but also on the initial number of infected corresponding to these parameters; depending on the choices we make, we can observe an enhancement, preservation or erosion of the diversity of variants in the population. Example 1 shows the case of an algebraically converging fitness function with an algebraically decreasing initial data. Figure \ref{Fig1} (bottom) shows the numerical computation of the relative frequencies of the variants; the number of variants with non-negligible proportion seems to be increasing with time, hence in this case the diversity seems to be increasing with time (enhancement of diversity). In Figure \ref{Fig2} (bottom) we show the numerical computations for the case of a fitness function converging exponentially fast with an exponentially decreasing initial data; here the diversity seems to be approximately constant in time, with a constant number of variants that dominate the others (preservation of diversity). Figure \ref{Fig3} (bottom) shows a similar fitness function but with Gaussian initial data; here the transition between variants becomes steeper with time, and our analysis in Claim \ref{claim:gaussian-exponential} suggests that only one variant dominates all the other for increasingly large periods asymptotically (although this may happen in a larger time frame than the one presented in Figure \ref{Fig3}). Thus in this case we observe an erosion of the diversity of variants. Those three different types behaviors can be proved analytically, see Claim \ref{claim:algebraic-algebraic}, Claim \ref{claim:exponential-exponential} and Claim \ref{claim:gaussian-exponential}.}    } 

{While Figure \ref{Fig1}, Figure \ref{Fig2} and Figure \ref{Fig3} focus on the case when $\gamma_n$ is a constant, in Figure \ref{Fig4} and Figure \ref{Fig5} we investigate the case when $\gamma_n$ oscillates between two values and we show that surprising behaviors may occur then. Figure \ref{Fig4} illustrates that, when Assumption \ref{as:disintegration} does not hold, we cannot hope to generalize the results of Proposition \ref{prop:convergence-mass}: indeed the family of variants with the highest value of $\gamma_n$ (in blue to green) gets extinct and the family of variants with the lowest value of $\gamma_n$ (in red to yellow) dominates in this case. This is proved in Claim \ref{claim:unexpected-selection}. Figure \ref{Fig5} illustrates that in some particular cases, neither family uniformly dominates the population asymptotically, but they both dominate the population alternatively. Claim \ref{claim:alternating-persistence} shows moreover that, in this situation, the total population of infected does not converge to a limit but oscillates between two distinct values. }

The assumption that an infinite number of variants exist at the same time corresponds to a creative simplification of reality that allows us to describe possible behaviors occurring in nature. The classical theory for finite systems gives us access to a single type of dynamics: convergence to a unique equilibrium. Yet practical observations in real-life epidemics such as the COVID-19 epidemic suggest that it is not the case \parencite{Brussow-2022}. In the end, a description of the succession of infinitely many variants might give a more realistic description of the observed phenomenon than a finite model.

We predict that, when there is equality between different types, the variants that are associated with a fast dynamics are favored in the long run. This may seem counterintuitive since the parameter that is maximized, $\gamma$, incidentally corresponds to the inverse of the infection period; however, we conjecture that the behavior that it selected is the one associated with a fast transmission rather than low infection period. This could be checked by splitting the different coefficients associated with recovery, host mortality, and transmission in \eqref{eq:main}. We leave such a refined description for future work.

Our work could be extended in several directions. An important addition would be to account for the influence of the age of infection in our model, as in the original article of \textcite{Ker-McK-1927} (see also \textcite{Demongeot-Griette-Magal-Maday-2023}). We consider it an exciting motivation for future works.

\section{Proof of the mathematical results}
\label{sec:proof}

\subsection{Proof of Proposition \ref{prop:Cauchy}}

We prove that problem \eqref{eq:main} is well posed. Thanks to Assumption~\ref{as:params} we observe that the map $F:\mathbb{R}\times\ell^1 \rightarrow \mathbb{R}\times\ell^1 $ defined by the right hand side of the system of equations \eqref{eq:main-a} is locally Lipschitz on the Banach space $\mathbb{R}\times\ell^1$. Therefore problem \eqref{eq:main} admits a unique maximal solution $\big(S(t), (I_n(t))_{n\in\mathbb{N}}\big)\in C^1\big([0,T_{\rm max}), \mathbb{R}\big)\times C^1\big([0,T_{\rm max}), \ell^1\big)$ for some $T_{\rm max}>0$ possibly infinite. Moreover, formula \eqref{eq:I_n} holds, i.e.  
$$
I_n(t) = e^{t\left(\beta_n \overline{S}(t) - \gamma_n\right)}I^0_n,
$$
in particular $I_n$ is a non-negative function for all $n \in \mathbb{N}$ and
$\big(I_n^0\big)_{n\in\mathbb{N}}\in \ell^1_+ \longmapsto \big(I_n(t)\big)_{n\in\mathbb{N}} \in \ell^1_+$ is continuous for all $t>0$.

Next, the following a priori estimate proves that  the solution $\big(S(t), (I_n(t))_{n\in\mathbb{N}}\big)$ is uniformly bounded in $\mathbb{R}\times \ell_+^1$ so that $T_{\rm max}=\infty$:
\begin{lemma}\label{lem:apriori-bounds}  
	Suppose that Assumption \ref{as:params} holds true.
	For all $t \in [0,T_{\rm max})$,
	\begin{gather}\label{ineq:apriori-bounds}
		0 \le S(t) \le S_0e^{-\theta t}+ \frac{\Lambda}{\theta}\left(1 - e^{-\theta t}\right)
		\le \max\left(S_0,\frac{\Lambda}{\theta} \right)			
		,\\
		0 \le S(t) +\sum _{n \in \mathbb N}I_n(t)\leq \frac{\Lambda}{\min(\theta, \gamma_0)} + \left(S_0+\sum _{n \in \mathbb N}I_n^0 -\frac{\Lambda}{\min(\theta, \gamma_0)}\right)e^{-\min(\theta, \gamma_0)t}  
		. 
	\end{gather} 
\end{lemma}
\begin{proof}
	Using the first equation of \eqref{eq:main-a}, we have $S(t) \ge 0$ and due to the positivity of $I_n(t)$ we  readily prove the Lemma inequalities for $S(t)$.

	Next, adding  equations of \eqref{eq:main-a}  we remark that 
	\begin{equation*}
		\frac{\dd}{\dd t}\left(S(t)+\sum _{n \in \mathbb N}I_n(t) \right)\leq \Lambda - \theta S(t) - \gamma_0\sum _{n \in \mathbb N}I_n(t),
	\end{equation*}
	which yields the inequality for $S(t) +\sum _{n \in \mathbb N}I_n(t)$.
	The Lemma is proved.
\end{proof} %

\subsection{Proof of Proposition \ref{prop:Extinction}}

We first assume that $\mathcal{R}_0 <1$. Let $\varepsilon>0$ such that
$$
m=1 -  \mathcal{R}_0 \left(1+\varepsilon\frac{ \theta}{\Lambda}\right)  >0,
$$
there exists $t_0(\varepsilon)>0$ such that
$$
S(t) \le \frac{\Lambda}{\theta}+\varepsilon.
$$
Recall that the equations of \eqref{eq:main-a} for $I_n$ write for all $n \in \mathbb{N}$
$$
\frac{\dd}{\dd t} I_n (t)=\beta_n I_n S(t)-\gamma_nI_n(t),
$$
or, by definition of $\mathcal{R}_0$, as 
$$
\frac{\dd}{\dd t} I_n (t) \le \gamma_n\left( \mathcal{R}_0\frac{\theta}{\Lambda}S(t)-1 \right)I_n(t).
$$
For all $n \in \mathbb{N}$ and $t\ge t_0(\varepsilon)$
$$
\frac{\dd}{\dd t} I_n (t) \le -m \gamma_n I_n(t),
$$
therefore
$$
0 \le \sum_{n \in \mathbb{N}}I_n(t) \le e^{-m \gamma_0 (t-t_0(\varepsilon))}I_n(t_0(\varepsilon)),
$$
which proves that $\lim_{t\to+\infty} \sum_{n=0}^{+\infty} I_n(t) = 0$.

Next, as $S(t)$ is bounded, we consider $\underline{S} := \liminf_{t\to+\infty} S(t)$.
Let  $(t_k)_{n\geq 0}$ be a sequence that tends to $\infty$ as $k\to\infty$ and such that $\lim_{k \to \infty} S'(t_k)=0$ and $\lim_{ k \to +\infty} S(t_k) = \underline{S}$. Since $(\beta_n)$ is bounded, from the following equality for all $k$
$$
\frac{\dd}{\dd t}S(t_k) 	= \Lambda -\theta S(t_k) - \sum_{n=0}^{+\infty}\beta_n S(t_k)I_n(t_k)
$$		
passing to the limit, we obtain $\liminf_{t\to\infty} S(t)=\frac{\Lambda}{\theta}$. Similarly, $\limsup_{t\to\infty} S(t)=\frac{\Lambda}{\theta}$, therefore
\begin{equation*}
	\lim_{t\to+\infty} S(t) = \frac{\Lambda}{\theta},
\end{equation*}
which concludes the first part of the proof of the proposition.	The second part of Proposition \ref{prop:Extinction} will be proved in section \ref{sec:R0=1}, because we need more tools to prove it.

\subsection{Proof of Theorem  \ref{thm:discrete}  }
\label{sec:discrete}

{
Before we start the proof of Theorem \ref{thm:discrete}, we introduce a few notions that will be useful along this Section. Indeed, in order to get a compactness of the orbits, we need to include our dynamical system in a larger space. We define the distance on $\mathbb{N}$:
\begin{equation*}
	d(n, m):= \left|\frac{1}{1+n}-\frac{1}{1+m}\right|+|\alpha_n-\alpha_m|+|\gamma_n-\gamma_m|, \text{ for all } n,m\in\mathbb{N}.
\end{equation*}
We let $\overline{\mathbb{N}}$ be the topological completion of $\mathbb{N}$ for the distance $d$. It is essentially the smallest closed set for the distance $d$ containing $\mathbb{N}$. Because of Assumption \ref{as:omega} the sets $\omega(\alpha)$ and $\omega(\gamma)$ are finite and we have, up to a topological isomorphism which we will omit in the rest of the proof,
\begin{equation*}
	\big(\overline{\mathbb{N}}, d\big)=\big([-K, +\infty)\cap \mathbb{Z}, d\big), 
\end{equation*}
wherein we have set $K:=\#\omega(\alpha)\times \#\omega(\gamma)$, 
\begin{equation*}
	\omega(\alpha)\times \omega(\gamma)=:\{(\alpha_{-i}, \gamma_{-i}), i=1, \ldots, K \}, 
\end{equation*}
and
\begin{equation*}
	d(n, m):= \left|\dfrac{\mathbbm{1}_{n\geq 0}}{1+n}-\dfrac{\mathbbm{1}_{m\geq 0}}{1+m}\right|+|\alpha_n-\alpha_m|+|\gamma_n-\gamma_m|, \text{ for all } n,m\in\mathbb{Z}\cap [-K, +\infty).
\end{equation*}
In particular $\overline{\mathbb{N}}$  is  Hausdorff and countable, which implies that the Borel $\sigma$-algebra is the set of all parts of $\overline{\mathbb{N}}$, and therefore any Borel measure $\mu\in\mathcal{M}\big(\overline{\mathbb{N}}\big)$ can be represented by a summable sequence:
\begin{equation*}
	\mu = \sum_{n=-K}^{+\infty} \mu_n \delta_n, 
\end{equation*}
where $\delta_n$ is the Dirac mass concentrated on $n\in\overline{\mathbb{N}}$ and $(\mu_n)_{n\geq -K}$ is a summable sequence of real numbers. Finally, $\overline{\mathbb{N}}$ is compact for the topology generated by $d$.

In what follows we will obtain the compactness of the orbit by using the weak-$\ast$ topology on the space of measures $\mathcal{M}\big(\overline{\mathbb{N}}\big)$. To mark the difference, we will write $\mathcal{M}^\ast(\overline{\mathbb{N}})$, $\mathcal{M}^\ast_+(\overline{\mathbb{N}})$ instead of $\mathcal{M}\big(\overline{\mathbb{N}}\big)$, $\mathcal{M}_+(\overline{\mathbb{N}})$ when the space is equipped with the weak-$\ast$ topology. Recall the topology on $\mathcal{M}(\overline{\mathbb{N}}) $ is generated by the norm
\begin{equation*}
	\Vert (\mu)_{n\in\overline{\mathbb{N}}}\Vert_{\mathcal{M}(\overline{\mathbb{N}})}:=\sum_{n\in\overline{\mathbb{N}}} |\mu_n|, 
\end{equation*}
and that the topology on $\mathcal{M}^\ast(\overline{\mathbb{N}})$ is that of the weak-$\ast$ convergence: convergence of a sequence $\mu^n\rightharpoonup \mu $ for this topology holds if, and only if, 
\begin{equation*}
	\sum_{k=-K}^{+\infty} \varphi_k\mu^n_k\xrightarrow[n\to+\infty]{} \sum_{k=-K}^{+\infty} \varphi_k \mu_k, 
\end{equation*}
for all $(\varphi_n)\in C\big(\overline{\mathbb{N}}\big)$ the space of continuous sequences over $\overline{\mathbb{N}}$, which is characterized by 
\begin{equation*}
	\left[(\varphi_k)\in C\big(\overline{\mathbb{N}}\big)\right] \Longleftrightarrow \left[ \lim_{j\to+\infty} \varphi_{k_j} = \varphi_{-i} \text{ whenever } \alpha_{k_j}\to \alpha_{-i} \text{ and } \gamma_{k_j}\to \gamma_{-i} \text{ with } k_j\to+\infty \text{ and }i>0\right].
\end{equation*}
	 
In what follows we will consider equation \eqref{eq:main} with an initial data $(I_n)_{n\in\overline{\mathbb{N}}}\in\mathcal{M}_+\big(\overline{\mathbb{N}}\big)$. Because of our construction of $\overline{\mathbb{N}}$, Assumptions \ref{as:params} and \ref{as:omega} need not be adapted to the new framework. Assumption \ref{as:init}, however, does. Let us the replace Assumption \ref{as:init} with the following: 

\begin{assumption}\label{as:init-bis}  
	We let $S_0>0$, $(I_n^0)_{n\in\overline{\mathbb{N}}}\in \mathcal{M}_+(\overline{\mathbb{N}})$ be given and assume that there exists a sequence of indices $n_k\in\overline{\mathbb{N}}$ with 
	\begin{equation*}
		I_{n_k}^0>0 \text{ and } \lim_{k\to+\infty} \frac{\beta_{n_k}}{\gamma_{n_k}}=\alpha^*.
	\end{equation*}
\end{assumption}
\noindent As before the sequence of indices $(n_k)$ need not be strictly monotone and can be eventually stationary. 
}

The following lemma holds true.
\begin{lemma}\label{lem:bounds-nomut}  
	Suppose that Assumption \ref{as:params} holds true. Then we have
	\begin{gather*}
		0<\frac{\min(\theta, \gamma_0)}{\theta\Lambda\min(\theta, \gamma_0)+\beta^\infty}\leq \liminf_{t\to+\infty}S(t)\leq \limsup_{t\to+\infty}S(t)\leq \frac{\Lambda}{\theta}<+\infty,\\
		\limsup_{t\to+\infty}\sum _{n \in \overline{\mathbb{N}}}I_n(t) \leq \frac{\Lambda}{\min(\theta, \gamma_0)}<+\infty.
	\end{gather*} 
\end{lemma}
\begin{proof} 
	The upper bounds are proved as in  Lemma \ref{lem:apriori-bounds}.
	Next we return to the $S$-component of equation \eqref{eq:main} and let $\varepsilon>0$ be given. We have, for $t_0$ sufficiently large and $t\geq t_0$,  
	\begin{equation*} 
		S_t = \Lambda -\left( \theta + \sum _{n \in \overline{\mathbb{N}}} \beta_n I_n(t) \right) S(t)\geq \Lambda - \left(\theta+\beta^\infty\frac{\Lambda}{\min(\theta, \gamma_0)}+\varepsilon\right)S(t),
	\end{equation*}
	therefore 
	\begin{equation*}
		S(t)\geq e^{-\left(\theta+\frac{\Lambda\beta^\infty}{\min(\theta, \gamma_0)}+\varepsilon\right)(t-t_0)}S(t_0)+\frac{\Lambda\min(\theta, \gamma_0)}{(\theta+\varepsilon)\min(\theta, \gamma_0)+\Lambda\beta^\infty}\left(1-e^{-\left(\theta+\frac{\Lambda\beta^\infty}{\min(\theta, \gamma_0)}+\varepsilon\right)(t-t_0)}\right),
	\end{equation*}
	so that finally by letting $t\to+\infty$ we get
	\begin{equation*}
		\liminf_{t\to+\infty} S(t) \geq \frac{\min(\theta, \gamma_0)\Lambda}{(\theta+\varepsilon)\min(\theta, \gamma_0)+\Lambda\beta^\infty}.
	\end{equation*}
	Since $\varepsilon>0$ is arbitrary we have shown
	\begin{equation*}
		\liminf_{t\to+\infty} S(t) \geq \frac{\min(\theta, \gamma_0)\Lambda}{\theta\min(\theta, \gamma_0)+\Lambda\beta^\infty}.
	\end{equation*}
	The Lemma is proved.
\end{proof} %

\begin{lemma}\label{lem:limsup-no-mut} 
	Suppose that the Assumptions \ref{as:params} and \ref{as:init-bis} hold true.
	Let $\big(S(t), I_i(t)\big)$ be the corresponding solution of \eqref{eq:main}.
	Then
	\begin{equation*}
		\limsup_{T\to+\infty}\frac{1}{T}\int_0^T S(t)\dd t \leq \dfrac{1}{\alpha^*}.
	\end{equation*}
\end{lemma}
\begin{proof}
	Let us remark that the second component of \eqref{eq:main} for any $n \in \mathbb N$ can be written as
	\begin{align}
		I_n(t) &= I_0^n e^{\beta_n\int_0^tS(s)\dd s-\gamma_nt},\notag \\
		&=I_0^n  \exp\left(\gamma_n\int_0^t S(s)\dd s\left[\frac{\beta_n}{\gamma_n}-\frac{t}{\int_0^tS(s)\dd s}\right]\right)\label{eq:I-nomut}.
	\end{align}
	Assume by contradiction that the conclusion of the Lemma does not hold, {\it i.e.} there exists $\varepsilon>0$ and a sequence $T_n\to +\infty$ such that 
	\begin{equation*}
	    \frac{1}{T_n}\int_0^{T_n}S(t)\dd t\geq  \dfrac{1}{\alpha^*}+\varepsilon \text{ for all }n\in\mathbb{N}.
	\end{equation*}
	Then
	\begin{equation*}
		\frac{T_n}{\int_0^{T_n}S(t)\dd t}\leq  \frac{1}{\dfrac{1}{\alpha^*}+\varepsilon} \leq \alpha^*-\varepsilon',
	\end{equation*}
	for some $\varepsilon'>0$. Since  $\alpha^*=\sup_n \frac{\beta_n}{\gamma_n}$ there exists some $k \in \mathbb N$ with $I^0_k>0$ such that $\frac{\beta_k}{\gamma_k}-\alpha^*+\varepsilon'>0$, and 
	\begin{align*}
		\sum_{ i \in \mathbb{N}} I_i(T_n) & \ge I_0^k  \exp\left(\gamma_k\int_0^{T_n} S(s)\dd s\left[\frac{\beta_k}{\gamma_k}-\frac{T_n}{\int_0^{T_n} S(s)\dd s}\right]\right) \\
		& \ge I_0^k  \exp\left(\gamma_k \left( \alpha_k - \alpha^* + \varepsilon' \right)\int_0^{T_n} S(s)\dd s \right).
	\end{align*} 
	Since   $\int_0^{T_n}S(t)\dd t \to +\infty$ when $n\to+\infty$, we have therefore 
	\begin{equation*}
		\limsup_{t\to+\infty} \sum _{i \in \overline{\mathbb{N}}}I_i(t) \geq \limsup_{n\to+\infty}\sum _{i \in \overline{\mathbb{N}}}I_i(T_n) = +\infty,
	\end{equation*}
	which is a contradiction since $I_n(t) $ is bounded in $\mathcal{M}_+(\overline{\mathbb{N}})$ by Lemma \ref{lem:bounds-nomut}.
	This completes the proof of the Lemma.
\end{proof}
{Remark that, in the Lemma above, Assumption \ref{as:init-bis} is essential. Indeed, were this assumption not true, we could not guarantee that the index $k$ defined in the proof corresponds to a strictly positive $I_k^0$, hence the contradiction would not be guaranteed either. }

The following weak persistence property holds.
\begin{lemma}\label{lem:weak-persistent-nomut} 
	Suppose that Assumptions \ref{as:params} and \ref{as:init-bis} hold true.
	Let $\big(S(t), I_n(t)\big)$ be the corresponding solution of \eqref{eq:main}.
	Then 
	\begin{equation}\label{eq:weak-persistence}
		\limsup_{t\to+\infty} \sum _{n \in \overline{\mathbb{N}}}I_n(t)  \ge \frac{\theta}{\beta^\infty} \big( \mathcal{R}_0-1 \big) >0. 
	\end{equation}
\end{lemma}
\begin{proof} 
	Let us recall that we made the hypothesis that
	\begin{equation*}
		{\mathcal R}_0  = \frac{\Lambda}{\theta} \alpha^*   >1.
	\end{equation*}
	Assume by contradiction that  for  $t_0$ sufficiently large we have
	\begin{equation*}
		\sum_{i \in \overline{\mathbb{N}}} I_i(t)   \le  \eta < \frac{\theta}{\beta^\infty} \big({\mathcal R}_0 -1 \big) \text{ for all }t\geq t_0,
	\end{equation*}
	with $\eta>0$.

	As a consequence of Lemma \ref{lem:limsup-no-mut} we have 
	\begin{equation}\label{eq:infS1}
		\liminf_{t\to+\infty} S(t)\leq \limsup_{T\to+\infty}\frac{1}{T}\int_0^T S(t)\dd t \leq \frac{1}{\alpha^*}.
	\end{equation}

	Let $\underline{S} := \liminf_{t\to+\infty} S(t)$.
	Let  $(t_n)_{n\geq 0}$ be a sequence that tends to $\infty$ as $n\to\infty$ and such that $\lim_{n \to +\infty} S'(t_n)=0$ and $\lim_{ n \to +\infty} S(t_n) = \underline{S}$. As $\sum _{i \in \overline{\mathbb{N}}}I_i(t_n)  \leq \eta$ for $n$ large enough we deduce from the equality
	\begin{equation*}
		S'(t_n)=\Lambda - \theta S(t_n) - S(t_n) \sum_{i \in \overline{\mathbb{N}}}  \beta_i I_i(t_n),
	\end{equation*}
	that
	\begin{equation*} 
		0 \ge \Lambda - \theta \underline{S} - \underline{S} \beta^\infty \eta
	\end{equation*}
	so that
	\begin{equation*}
		\underline{S} \ge  \frac{\Lambda}{\theta+ \beta^\infty\eta} > \frac{\Lambda}{\theta+ \beta^\infty\eta} =\frac{\Lambda}{\theta {\mathcal R}_0} = \frac{1}{\alpha^*},
	\end{equation*}
	which contradicts \eqref{eq:infS1}.
\end{proof}

{
\begin{proposition}[Compactness of the orbit and concentration]\label{prop:compactness}
	Suppose that Assumptions \ref{as:params}, \ref{as:init-bis} and \ref{as:omega} hold true. 
	Then, the map $t\mapsto \big(S(t), (I_n(t))_{n\in\overline{\mathbb{N}}}\big)\in\mathbb{R}\times \mathcal{M}^\ast_+\big(\overline{\mathbb{N}}\big) $ is continuous and the corresponding  orbit, 
	\begin{equation*}
		O\left(S_0, (I^0_n)_{n\in\overline{\mathbb{N}}}\right):=\left\{\left(S(t), \big(I_n(t)\big)_{n\in\overline{\mathbb{N}}}\right)\,:\,t\geq 0\right\}\subset\mathbb{R}\times\mathcal{M}^\ast_+\big(\overline{\mathbb{N}}\big),
	\end{equation*}
	is precompact for the weak-$\ast$ topology. Moreover  if  $\mathcal{R}_0 >1$ and if $t_k\to+\infty$ is an arbitrary sequence along which
	\begin{equation*}
		\liminf_{k\to+\infty} \sum_{n=-K}^{+\infty} I_n(t_k)>0, 
	\end{equation*}
	then one can extract from $(t_k)$ a subsequence $(t_{k_i})$ such that the shifted orbits
	\begin{equation*}
		t\mapsto \left( S(t+t_{k_i}), (I_n(t+t_{k_i}))_{n\in\overline{\mathbb{N}}}\right)
	\end{equation*}
	converge pointwise in $\mathbb{R}\times \mathcal{M}^\ast_+\big(\overline{\mathbb{N}}\big) $ to a complete orbit $\left(S^\infty(t), \big(I_n^\infty(t)\big)_{n\in\overline{\mathbb{N}}}\right)$ that satisfies the following properties:
	\begin{equation}\label{eq:compacness-limit-non-trivial}
		\sum_{n=-K}^{+\infty} I^\infty_n(t) >0 \text{ and }S^\infty(t)>0 \text{  for all }t\in\mathbb{R}, 
	\end{equation}
	and 
	\begin{equation}\label{eq:concentration}
		\sum_{\{n\,:\,\alpha_n<\alpha^*\}} I_n^\infty(t) = 0\text{  for all  }t\in\mathbb{R}.
	\end{equation}
	Finally, the convergence $S(t_{k_i}+t)\to S^\infty(t)$ is locally uniform in $C^1(\mathbb{R})$.
\end{proposition}
\begin{proof}
	First of all let us remark that 
	\begin{equation*}
		I_n(t) = e^{\beta_n\int_0^t S(s)\dd s-\gamma_nt}I^0_n= e^{\gamma_n\left(\alpha_n\int_0^t S(s)\dd s-t\right)}I^0_n.
	\end{equation*} 
	Since $\alpha_n$ and $\gamma_n$ are in $C\big(\overline{\mathbb{N}}\big)$, the map $t\mapsto e^{\gamma_n\left(\alpha_n\int_0^t S(s)\dd s-t\right)} $ is continuous from $\mathbb{R}^+ $ to $C\big(\overline{\mathbb{N}}\big)$, and therefore $t\mapsto (I_n(t))_{n\in\overline{\mathbb{N}}}$ is continuous from $\mathbb{R}$ to $\mathcal{M}_+^\ast\big(\overline{\mathbb{N}}\big)$. Thus $t\to\big(S(t),(I_n(t))_{n\in\overline{\mathbb{N}}}\big)$ is indeed continuous from $\mathbb{R}^+$ to $\mathbb{R}\times\mathcal{M}^\ast_+(\overline{\mathbb{N}})$. Since $S(t)$ is uniformly bounded and $(I_n(t))_{n\in\overline{\mathbb{N}}}$ is uniformly bounded in total variation norm by Lemma \ref{lem:bounds-nomut}, it follows from the Prohorov Theorem \parencite[Vol. II Theorem 8.6.2 p. 202]{Bog-07} that the orbit is precompact in $\mathbb{R}\times\mathcal{M}^\ast_+\big(\overline{\mathbb{N}}\big)$.

	Let $t_k$ be a sequence such that $t_k\to+\infty$. 	Because $\sum_{n\in\overline{\mathbb{N}}} I_n(t_k)$ is bounded, we can extract from $(t_k)$ a subsequence, still denotes $t_k$, along which 
	\begin{equation*}
		S(t_k)\to S^\infty_0 \text{ and } (I_n(t_k))_{n\in\overline{\mathbb{N}}}\rightharpoonup (I_n^{\infty, 0})_{n\in\overline{\mathbb{N}}}.
	\end{equation*}
	We remark that $S'(t+t_k)$ is bounded thanks to Lemma \ref{lem:bounds-nomut} and
	\begin{align*} 
		S''(t+t_k) &= -\theta S'(t+t_k) - S'(t+t_k) \sum_{n\in\overline{\mathbb{N}}} \beta_n I_n(t+t_k) - S(t+t_k) \sum_{n\in\overline{\mathbb{N}}} \beta_n\left(\beta_nS(t+t_k)-\gamma_n\right) I_n(t+t_k) 
	\end{align*}
	is also bounded, locally uniformly in $t$. Thus up to a further extraction and diagonal extraction process, the shifted orbit $S(t+t_{k})$ converges locally uniformly in $C^1(\mathbb{R})$.
	Because
	\begin{equation*}
		I_n(t_k+t) = e^{\gamma_n\left(\alpha_n\int_{t_k}^{t}S(t_k+\sigma)\dd\sigma - (t-t_k)\right)} I_n(t_k), \qquad n\in\overline{\mathbb{N}}, 
	\end{equation*}
	we can pass to the weak-$\ast$ limit in the above formula and we get
	\begin{equation}\label{eq:prop-compactness-formula-In}
		I_n^\infty(t):= e^{\gamma_n\left(\alpha_n\int_{0}^{t}S^\infty(\sigma)\dd\sigma - t\right)} I^{\infty, 0}_n, \qquad n\in\overline{\mathbb{N}}, t\in\mathbb{R},
	\end{equation}
	so that $(S^\infty(t), (I_n^\infty(t))_{n\in\overline{\mathbb{N}}})$ is a solution of \eqref{eq:main} with starting from the initial data $(S^\infty_0, (I^{\infty, 0}_n))$ and with $t\in\mathbb{R}$: a complete orbit. That $S^\infty(t)>0$ is a consequence of Lemma \ref{lem:bounds-nomut}. Since {the constant function} $n\mapsto 1\in C\big(\overline{\mathbb{N}}\big)$ we have
	\begin{equation*}
		\sum_{n\in\overline{\mathbb{N}}} I_n^{\infty, 0} = \sum_{n\in\overline{\mathbb{N}}}1\cdot I_n^{\infty, 0} =  \lim_{k\to+\infty} \sum_{n\in\overline{\mathbb{N}}} 1\cdot I_n(t_k) 
	\end{equation*}
	so if 
	\begin{equation*}
		\liminf_{k\to+\infty} \sum_{n\in\overline{\mathbb{N}}}I_n(t_k)>0,
	\end{equation*}
	then we have
	\begin{equation*}
		\sum_{n\in\overline{\mathbb{N}}}I^{\infty, 0}_n>0,
	\end{equation*}
	and thanks to \eqref{eq:prop-compactness-formula-In}, 
	\eqref{eq:compacness-limit-non-trivial} is proved. 

	Next we show the concentration property \eqref{eq:concentration}. Recalling that $\liminf_{t\to+\infty}\frac{1}{t}\int_0^tS(s)\dd s \geq \underline{S}$ for a positive constant $\underline{S}>0$  and $\limsup_{t\to+\infty} \frac{1}{t}\int_0^t S(s)\dd s \leq \frac{1}{\alpha^*}$, we have for $t_k$ sufficiently large $\frac{1}{t_k}\int_0^{t_k}S(s)\dd s\leq \frac{1}{\alpha^*}+\frac{\varepsilon\underline{S}}{2\alpha^*}$ and thus
	\begin{align*}
		\sum_{\{n\,:\,\alpha_n<\alpha^*-\varepsilon\}}I_n(t_k) &= \sum_{\{n\,:\,\alpha_n<\alpha^*-\varepsilon\}} I_n^0 \exp\left[\gamma_nt\left(\alpha^*\frac{1}{t_k}\int_0^{t_k}S(\sigma)\dd \sigma - 1-\big(\alpha^*-\alpha_n\big)\frac{1}{t_k}\int_0^{t_k}S(\sigma)\dd \sigma\right)\right]\\ 
		&\leq\sum_{\{n\,:\,\alpha_n<\alpha^*-\varepsilon\}} I_n^0 \exp\left[\gamma_nt\left(\alpha^*\left(\frac{1}{\alpha^*}+\frac{\varepsilon\underline{S}}{2\alpha^*}\right)- 1-\varepsilon\underline{S}\right)\right] \\ 
		&=e^{-\gamma_0\frac{\varepsilon}{2}\underline{S}t}\sum_{n\in\overline{\mathbb{N}}} I_n^0 \xrightarrow[k\to+\infty]{}0,
	\end{align*}
	Since $\varepsilon>0$ is arbitrary, \eqref{eq:concentration} is proved.

	 This finishes the proof of Proposition \ref{prop:compactness}.
\end{proof}
 
We are now ready to state our uniform persistence result. 
\begin{proposition}[Uniform persistence]\label{prop:unifpers-nomut}
	Suppose that Assumptions \ref{as:params}, \ref{as:init-bis} and \ref{as:omega} hold true, and that $\mathcal{R}_0 >1$. Then, 
	\begin{equation}\label{eq:strong-persistence}
		\liminf_{t\to+\infty} \sum_{n \in \overline{\mathbb N}} I_n(t) >0.
	\end{equation}
\end{proposition}
\begin{proof}
	Let us show the uniform persistence property. We adapt the argument of \textcite[Proposition 3.2]{Mag-Zha-05} in our non-metric context. 
	Assume by contradiction  that \eqref{eq:strong-persistence} does not hold. Then, there exists a sequence $t_k\to+\infty$ such that 
	\begin{equation}\label{eq:uniform-persistance-aim}
		\lim_{k\to+\infty}\sum_{n=-K}^{+\infty} I_n(t_k)=0.
	\end{equation} 
	By Lemma \ref{lem:bounds-nomut} we have $\inf_kS(t_k)>0$. Because of \eqref{eq:weak-persistence}, for eack $k$ sufficiently large and up to replacing $t_k$ by a subsequence, there exists $s_k$ with $t_{k-1}<s_k<t_k$ such that 
	\begin{equation}\label{eq:norm-persistence}
		\sum_{n=-K}^{+\infty} I_n(s_k) = \frac{1}{2}\cdot \frac{\theta}{\beta^\infty} \big( \mathcal{R}_0-1 \big) \text{ and } \sup_{\sigma\in [s_k, t_k]}\sum_{n=-K}^{+\infty} I_n(\sigma) \leq \frac{1}{2}\cdot \frac{\theta}{\beta^\infty} \big( \mathcal{R}_0-1 \big) .
	\end{equation}
	By Proposition \ref{prop:compactness}, there exists a subsequence of $s_k$, still denoted $s_k$, such that $S(s_k+t) \to S_0^\infty(t) $ and $(I_n(s_k+t))\rightharpoonup \big(I_n^{\infty}(s_k+t)\big)$, and moreover \eqref{eq:compacness-limit-non-trivial} and \eqref{eq:concentration} hold for the limit orbit. 

	Next we show that $(S^\infty(t), (I_n^\infty(t)))$ satisfy the assumptions \ref{as:params}, \ref{as:init} and \ref{as:omega}. The assumptions \ref{as:params} and \ref{as:omega} are readily checked since the values of the coefficients $\alpha_n$ and $\gamma_n$ have not changed (neither have $\beta_n:=\gamma_n\alpha_n$). We deduce from \eqref{eq:concentration} that 
	\begin{align*}
		\sum_{\{n\,:\,\alpha_n=\alpha^*\}} I^{\infty, 0}_n &=\sum_{n=-K}^{+\infty} I^{\infty, 0}_n - \sum_{\{n\,:\,\alpha_n<\alpha^*\}}I_n(s_k) =\sum_{n=-K}^{+\infty} I^{\infty, 0}_n =\frac{1}{2}\cdot \frac{\theta}{\beta^\infty} \big( \mathcal{R}_0-1 \big)>0,  
	\end{align*}
	thus in particular  the set $\{n\in\overline{\mathbb{N}}\,:\, I_n^{\infty, 0}>0\text{ and } \alpha_n=\alpha^*\}$ is nonempty. Hence, up to removing the terms for which $I_n^{\infty, 0}=0$, we have
	\begin{equation*}
		\sup_{\{n\,:\, I_n^{\infty, 0}>0\}} \frac{\beta_n}{\gamma_n}=\sup_{\{n\,:\, I_n^{\infty, 0}>0\}} \alpha_n=\alpha^*, 
	\end{equation*}
	Assumptions \ref{as:init}, \ref{as:params} and \ref{as:omega} are still satisfied along the subsequence $I^{\infty, 0}_n>0$.
	In particular  we can apply Lemma \ref{lem:weak-persistent-nomut} and thus
	\begin{equation} \label{eq:230320-1}
		\limsup_{t\to+\infty} \sum_{n=-K}^{+\infty} I_n^\infty(t) \geq \frac{\theta}{\beta^\infty} \big( \mathcal{R}_0-1 \big). 
	\end{equation}
	Now we conclude the argument. Let $T_k:=t_k-s_k$. There are two possibilities. 

	$\bullet$ $T_k$ is bounded. In that case, we further extract a subsequence so that $T_k\to T$.  Then we have
	\begin{equation*}
		\sum_{n=-K}^{+\infty} I_n^\infty(T) = \lim_{k\to+\infty} \sum_{n=-K}^{+\infty} I_n(t_k) = 0, 
	\end{equation*}
	and by the uniqueness of the solution to \eqref{eq:main}, we have $I^\infty_n(t) \equiv 0$ for all $t\geq T$. This contradicts \eqref{eq:230320-1}.

	$\bullet$ $T_k$ is unbounded. In that case, we further extract a subsequence so that $T_k\to +\infty$.  But since
	\begin{equation*}
		\sup_{\sigma\in [s_k, t_k]}\sum_{n=-K}^{+\infty} I_n(\sigma) \leq \frac{1}{2}\cdot \frac{\theta}{\beta^\infty} \big( \mathcal{R}_0-1 \big) , 
	\end{equation*}
	we have 
	\begin{equation*}
		\sup_{\sigma\in [0, +\infty)}\sum_{n=-K}^{+\infty} I^\infty_n(\sigma) \leq \frac{1}{2}\cdot \frac{\theta}{\beta^\infty}\big( \mathcal{R}_0-1 \big)
	\end{equation*}
	and this, again, contradicts \eqref{eq:230320-1}.

	These contradictions prove that \eqref{eq:uniform-persistance-aim} cannot hold, hence \eqref{eq:strong-persistence} holds. The proof of Proposition \ref{prop:unifpers-nomut} is finished.
\end{proof}
\begin{lemma}\label{lem:liminf-no-mut}
	Suppose that Assumptions \ref{as:params},  \ref{as:init-bis} and \ref{as:omega} hold true and
	let $\big(S(t), (I_n(t))\big)$ be the corresponding solution of \eqref{eq:main}.
	Assume that {$\mathcal R_0>1 $}. 
	Then
	\begin{equation*}
		\liminf_{T\to+\infty}\frac{1}{T}\int_0^T S(t)\dd t \geq \dfrac{1}{\alpha^*}, 
	\end{equation*}
	with
	$ \alpha^*$ given in \eqref{def-alpha*}.
\end{lemma}
\begin{proof}
	Assume by contradiction that the conclusion of the Lemma does not hold, {\it i.e.} there exist $\varepsilon>0$ and a sequence $t_k\to +\infty$ such that 
	\begin{equation*}
		\frac{1}{t_k}\int_0^{t_k}S(t)\dd t\leq  \dfrac{1}{\alpha^*}-\varepsilon.
	\end{equation*}
	Then
	\begin{align*}
		\sum_{n\in\overline{\mathbb{N}}}I_n(t_k)&= \sum_{n\in\overline{\mathbb{N}}} I_n^0 e^{\beta_n\int_0^{t_k} S(s)\dd s-\gamma_n t} = \sum_{n\in\overline{\mathbb{N}}} I_n^0 e^{\gamma_nt_k\left(\frac{\beta_n}{\gamma_n}\frac{1}{t_k}\int_0^{t_k} S(s)\dd s-1\right)} \leq \sum_{n\in\overline{\mathbb{N}}} I_n^0 e^{\gamma_nt_k\left(\frac{\beta_n}{\gamma_n}\left(\frac{1}{\alpha^*}-\varepsilon\right)-1\right)} \\ 
		&\leq \sum_{n\in\overline{\mathbb{N}}} I_n^0 e^{-\varepsilon\beta_nt_k} \leq e^{-\varepsilon\beta_0t_k}\sum_{n\in\overline{\mathbb{N}}} I_n^0 \xrightarrow[k\to+\infty]{}0.
	\end{align*} 
	Therefore
	\begin{equation*}
		\liminf_{t\to+\infty}\sum_{n\in\overline{\mathbb{N}}}I_n(t)\leq \lim_{k\to+\infty}\sum_{n\in\overline{\mathbb{N}}}I_n(t_k)=0,
	\end{equation*}
	which is in contradiction with Lemma \ref{prop:unifpers-nomut}.
	This proves the Lemma.
\end{proof}

Next we give a Lyapunov functional that works for a special case in our model. It is close to the Lyapunov functional given for the Lotka-Volterra case of \cite{Hsu-78}, although in our case it is not possible to factorize $S(t)$ in the first equation of \eqref{eq:main-a}, which makes the computations intractable when $\frac{\beta_n}{\gamma_n}$ is not a constant.
\begin{proposition}[Lyapunov functional]\label{prop:Lyapunov}
    Let Assumption \ref{as:params} hold true and assume that $\mathcal{R}_0>1$. Let $(S^*, I^*_n)_{n\in\overline{\mathbb{N}}}$ be a stationary solution of \eqref{eq:main-a}, i.e. $S^*:=\frac{1}{\alpha^*}$ and $(I^*_n)_{n\in\overline{\mathbb{N}}}$ satisfies 
		 \begin{equation*}
			\sum_{n\in\overline{\mathbb{N}}} \gamma_n I_n^* = \frac{\theta}{\alpha^*}\left(\mathcal R_0-1\right).
		 \end{equation*} 
		 Assume furthermore that $\frac{\beta_n}{\gamma_n}\equiv \alpha^*$ is constant  whenever $I^*_n>0$ and let $\mathcal{N}_{+}$ be the set of indices for which $I^n_0>0$,. 
		 Define $g(x)=x-\ln(x)$, and let 
		\begin{equation*}
			D(V):=\left\{(S, (I_n)_n)\,:\, S>0, I_n=0\text{ whenever } I_n^*=0, \text{ and } \inf_{n\in\mathcal{N}_+}\,\dfrac{I_n}{I_n^*}>0\right\}\subset \mathbb{R}\times \ell^\infty\big((I_n^*)^{-1}\big), 
		\end{equation*}
		where $\ell^\infty\big((I_n^*)^{-1}\big)$ is equipped with the norm $\Vert (\varphi_n)_n\Vert_{\ell^\infty\big((I_n^*)^{-1}\big)}:=\sup_{n\in\overline{\mathbb{N}}}\left|\frac{\varphi_n}{I_n^*}\right|$.
		Then the functional
		\begin{equation}\label{eq:Liapunov}
			V(S, (I_n)):= S^* g\left(\frac{S}{S^*}\right)-S^* + \sum_{n\in\mathcal{N}_{+}} \left[I^*_ng\left(\frac{I_n}{I^*_n}\right) -I^*_n\right].
		\end{equation}
		is well-defined and continuous on the open set $D(V)$ for the topology induced by $\mathbb{R}\times \ell^\infty\big((I_n^*)^{-1}\big)$. 

		Moreover, if $\big(S(t), (I_n(t))_{n\in\mathbb{N}}\big)$ is a solution of \eqref{eq:main} 
		such that  $\big(S(0), (I_n(0))\big)\in D(V)$, then $\big(S(t), (I_n(t))_n\big)\in D(V)$ for all $t\geq 0$, $t\mapsto V(S(t), (I_n(t))_n)\in C^1$ and we have
		\begin{equation}\label{eq:Lyapunov-decreasing}
			\frac{\dd}{\dd t} V\big(S(t), (I_n(t))_{n\in\overline{\mathbb{N}}}\big) = -\dfrac{(S(t)-S^*)^2}{S(t)}\left(-\theta-\sum_{n\in\mathcal{N}_+} \alpha\gamma_n I^*_n\right).
		\end{equation}
\end{proposition}
\begin{proof}
	First we check that $(S^*, (I^*_n)_{n\in\mathbb{N}})$ is a stationary solution of \eqref{eq:main}. Indeed, 
	\begin{equation*}
		\frac{\dd}{\dd t}(I^*_n) = 0 = \gamma_n(\alpha^* S^*-1)I_n^* , 
	\end{equation*}
	and 
	\begin{equation*}
		\frac{\dd}{\dd t}(S^*) = 0 = \Lambda - \theta S^* - S^*\sum_{n\in\overline{\mathbb{N}}} \alpha^*\gamma_ne^{\tau\gamma_n}I^0_n .
	\end{equation*}

	Let $\big(S(t), (I_n(t))_{n\in\mathbb{N}}\big)$ be such that $\big(S(0), (I_n(0))\big)\in D(V)$, we check that $\big(S(t), (I_n(t))\big)\in D(V)$ for all $t\geq 0$. That $S(t)>0$ is a consequence of Lemma \ref{lem:bounds-nomut}. Then, we remark that 
	\begin{equation*}
		I_n(t) = I_n(0) e^{\beta_n \int_0^t S(s)\dd s-\gamma_n t}\geq I_n(0)e^{\beta_0\int_0^t S(s)\dd s - \gamma_\infty t}, 
	\end{equation*}
	therefore 
	\begin{equation*}
		\inf_{n\in\mathcal{N}_+} \dfrac{I_n(t)}{I_n^*}\geq e^{\beta_0\int_0^t S(s)\dd s - \gamma_\infty t}\inf_{n\in\mathcal{N}_+} \dfrac{I_n(0)}{I_n^*} >0, 
	\end{equation*}
	and we have proved that $\big(S(t), (I_n(t))\big)\in D(V)$ for all $t\geq 0$. The continuity and continuous differentiability of $t\mapsto V\big(S(t), (I_n(t))\big)$ follow from classical arguments.

	Next, writing $V_1(t) = S^* g\left(\frac{S(t)}{S^*}\right) $ and $V_2(t) = \sum_{n\in\overline{\mathbb{N}}} I^*_ng\left(\frac{I_n(t)}{I^*_n}\right)$, we have
	\begin{align*}
		V_1'(t) &= S^*\frac{S'(t)}{S^*}g'\left(\frac{S(t)}{S^*}\right) = \left(\Lambda - \theta S(t) - S(t)\sum_{n\in\overline{\mathbb{N}}}\beta_nI_n(t)\right)\left(1-\frac{S^*}{S(t)}\right)  \\  
		&=\left(\Lambda - \theta S(t) - S(t)\sum_{n\in\overline{\mathbb{N}}}\beta_nI_n(t) -\Lambda + \theta S^* + S^*\sum_{n\in\overline{\mathbb{N}}}\beta_nI^*_n\right)\left(1-\frac{S^*}{S(t)}\right) \\ 
		& = -\theta\frac{(S(t)-S^*)^2}{S(t)} + \left(S^*\sum_{n\in\overline{\mathbb{N}}}\beta_nI^*_n - S(t)\sum_{n\in\overline{\mathbb{N}}}\beta_nI_n(t)\right) \left(1-\frac{S^*}{S(t)}\right), \\
		& = -\theta\frac{(S(t)-S^*)^2}{S(t)} + \dfrac{S(t)-S^*}{S(t)}\sum_{n\in\overline{\mathbb{N}}}\alpha^*\gamma_n\big(I^*_nS^* - I_n(t)S(t)\big) , 
	\end{align*}
	and 
	\begin{align*}
		V_2'(t) &= \sum_{n\in\overline{\mathbb{N}}} I^*_n \frac{ I_n'(t)}{I^*_n} g'\left(\frac{I_n(t)}{I^*_n}\right) = \sum_{n\in\overline{\mathbb{N}}} \gamma_n\left(\alpha^* S(t) - 1\right)I_n(t)\left(1-\frac{I^*_n}{I_n(t)}\right)  \\ 
		& = \sum_{n\in\overline{\mathbb{N}}} \gamma_n\left(\alpha^* S(t) - 1\right)\left(I_n(t)-I^*_n\right)  \\ 
		& = \sum_{n\in\overline{\mathbb{N}}} \gamma_n\alpha^*\left( S(t) - S^*\right)\left(I_n(t)-I^*_n\right).
	\end{align*}
	Recalling $S^*=\frac{1}{\alpha^*}$,  we have therefore
	\begin{align*} 
		\frac{\dd}{\dd t} V(S(t), (I_n(t))_n) &= \frac{\dd }{\dd t} V_1(t) + \frac{\dd}{\dd t} V_2(t) \\
		& = -\theta\frac{(S(t)-S^*)^2}{S(t)}  + \dfrac{S(t)-S^*}{S(t)}\sum_{n\in\overline{\mathbb{N}}}\alpha^*\gamma_n\left(I_n^*S^*-I_n(t)S(t)+I_n(t)S(t)-I_n^*S(t)\right) \\
		& = -\theta\frac{(S(t)-S^*)^2}{S(t)}  - \dfrac{\big(S(t)-S^*\big)^2}{S(t)}\sum_{n\in\overline{\mathbb{N}}}\alpha^*\gamma_nI^*_n.
	\end{align*}
	Proposition \ref{prop:compactness} is proved.
\end{proof}
\begin{lemma}\label{lem:strong-compactness}
	Suppose that Assumptions \ref{as:params}, \ref{as:init-bis} and \ref{as:omega} holds true and assume that there exists an index $i\in\overline{\mathbb{N}}$ such that $I^0_i>0$ and
	\begin{equation*}
		\alpha^*=\sup_{n\in\mathbb{N}}\frac{\beta_n}{\gamma_n} = \frac{\beta_i}{\gamma_i}.
	\end{equation*}
	For $s<t$, let
	\begin{equation}\label{eq:defeta}
		\eta(t;s):= (t-s)\big(\alpha^*\overline{S}(t; s)-1\big), \text{ where }\overline{S}(t;s):=\frac{1}{t-s}\int_s^tS(\sigma)\dd \sigma.
	\end{equation}
	Then there exists a constant $\overline{\eta}>0$ such that for any $0\leq s<t$ one has 
	\begin{equation}\label{eq:strongcompactness+}
		-\overline{\eta}\leq \eta(t;s) \leq \overline{\eta}<+\infty.
	\end{equation}
	If moreover $\big(S(t), (I_n(t))_n\big)\in\mathbb{R}\times\mathcal{M}_+(\overline{\mathbb{N}})$ is a uniformly bounded complete orbit such that $\displaystyle\liminf_{t\to-\infty}\sum_{\{k\,:\,\alpha_k=\alpha^*\}} I_n(t)>0$, then for any $s<t$ one has
	\begin{equation}\label{eq:strongcompactness-}
		-\overline{\eta}\leq \eta(t;s)\leq \overline{\eta}<+\infty.
	\end{equation}
\end{lemma}
\begin{proof}
	Let us write $I_n(t)$ as
	\begin{equation*}
		I_n(t) = I_n(s)\exp\left[\gamma_n (t-s) \left(\big(\alpha^*\overline{S}(t;s)-1\big) - \left(\alpha^*-\frac{\beta_n}{\gamma_n}\right) \overline{S}(t;s)\right)\right] = I_n(s) e^{\gamma_n\eta(t;s)-\left(\alpha^*-\alpha_n\right)\overline{S}(t;s)}. 
	\end{equation*}
	We claim that $\eta(t;s)$ is uniformly bounded in $t, s$. Indeed, recalling $\alpha_n=\frac{\beta_n}{\gamma_n}$, we have by Jensen's inequality 
	\begin{equation*}
		\exp\left[\sum_{\{n\,:\,\alpha_n=\alpha^*\}} \gamma_n\eta(t;s)\dfrac{I_n(s)}{\displaystyle\sum_{\{k\,:\,\alpha_k=\alpha^*\}}I_k(s)}\right]\leq \sum_{\{n\,:\,\alpha_n=\alpha^*\}} e^{\gamma_n\eta(t;s)} \dfrac{I_n(s)}{\displaystyle\sum_{\{k\,:\,\alpha_k=\alpha^*\}}I_k(s)}, 
	\end{equation*}
	hence
	\begin{align*}
		\eta(t;s) &\leq \dfrac{\displaystyle\sum_{\{k\,:\,\alpha_k=\alpha^*\}}I_k(s)}{\displaystyle\sum_{\{k\,:\alpha_k=\alpha^*\}} \gamma_k I_k(s)} \ln\left[\sum_{\{n\,:\,\alpha_n=\alpha^*\}} e^{\gamma_n\eta(t;s)} \dfrac{I_n(s)}{\displaystyle\sum_{\{k\,:\,\alpha_k=\alpha^*\}}I_k(s)}\right] \\
		&=\dfrac{\displaystyle\sum_{\{k\,:\,\alpha_k=\alpha^*\}}I_k(s)}{\displaystyle\sum_{\{k\,:\alpha_k=\alpha^*\}} \gamma_k I_k(s)} \ln\left[ \dfrac{\displaystyle\sum_{\{n\,:\,\alpha_n=\alpha^*\}}I_n(t)}{\displaystyle\sum_{\{k\,:\,\alpha_k=\alpha^*\}}I_k(s)}\right].
	\end{align*}
	Applying Lemma \ref{lem:bounds-nomut}, the total mass $\sum_{n\in\overline{\mathbb{N}}}I_n(t)$ is bounded above, and by Proposition \ref{prop:unifpers-nomut} the sum $\displaystyle\sum_{\{k\,:\,\alpha_k=\alpha^*\}}I_k(s) $ is bounded below when $s\to+\infty$; therefore there exists $\overline{\eta}<+\infty$, independent of $s$, such that 
	\begin{equation*}
		\eta(t;s) \leq \overline{\eta}.
	\end{equation*}
	If $\big(S(t), (I_n(t))_n\big)\in\mathbb{R}\times\mathcal{M}_+(\overline{\mathbb{N}})$ is a complete orbit and $\displaystyle\liminf_{t\to-\infty}\sum_{\{k\,:\,\alpha_k=\alpha^*\}} I_n(t)>0$, then there exists  an upper bound valid for all $t, s\in\mathbb{R}$.

	On the other hand, we claim that $\displaystyle\liminf_{t\to+\infty} \eta(t;s)\geq \underline{\eta}$ for a constant $\underline{\eta}$. Indeed, assume by contradiction that there exists a sequence $t_k\to+\infty$ and $s_k\geq 0$  such that $\eta(t_k;s_k)\to-\infty$.  Then
	\begin{equation*}
		\sum_{n\in\overline{\mathbb{N}}}I_n(t) =  \sum_{n\in\overline{\mathbb{N}}}I_n(s_k) e^{\gamma_n \eta(t_k;s_k)- \left(\alpha^*-\alpha_n\right) \overline{S}(t_k;s_k)} \leq e^{\gamma_0 \eta(t_k; s_k)}\sum_{n\in\overline{\mathbb{N}}}I_n(s_k)  \xrightarrow[t\to+\infty]{}0, 
	\end{equation*}
	which contradicts Proposition \ref{prop:unifpers-nomut}. We have proved \eqref{eq:strongcompactness+}. \eqref{eq:strongcompactness-} is proved by identical arguments.  
\end{proof}

Next we derive a kind of  LaSalle principle \parencite{LaSalle-1960} that shows that complete orbits concentrated on the set $\frac{\beta_n}{\gamma_n}=\alpha^*$ are ``almost'' stationary.  
\begin{lemma}\label{lem:convergence-S}
	Let Assumption \ref{as:params} and \ref{as:omega} hold true and assume that $(I^0_n)_{n\in\overline{\mathbb{N}}}$ is non-trivial.   Assume furthermore that $\frac{\beta_n}{\gamma_n}\equiv \alpha^*$ is constant  whenever $I^0_n>0$ and that $\mathcal{R}_0>1$. Let $\big(S(t), (I_n(t))_n\big)$ be a complete orbit of \eqref{eq:main} passing through $(S_0, I^0_n)$ at $t=0$ and suppose that 
	\begin{equation}\label{eq:LaSalle-boundbelow}
		\liminf_{t\to-\infty}\sum_{n\in\overline{\mathbb{N}}}I_n(t)>0.
	\end{equation}
	Then we have
	\begin{equation*}
		S(t) \equiv \frac{1}{\alpha^*} \text{ for all } t\in\mathbb{R}. 
	\end{equation*}
\end{lemma}
\begin{proof}
	Because of our assumption we have 
	\begin{equation}\label{eq:LaSalle-boundbelow}
		\liminf_{t\to-\infty}\sum_{n\in\overline{\mathbb{N}}}I_n(t)=\liminf_{t\to-\infty}\sum_{\{n\,:\,\alpha_n=\alpha^*\}}I_n(t)>0.
	\end{equation}
	therefore Lemma \ref{lem:strong-compactness} implies that $\eta(t;s)$, as defined by \eqref{eq:defeta}, is uniformly bounded.  Let $\eta(t):=\eta(t;0)$. We define the distribution:
	\begin{equation*}
		I^*_n:=\begin{cases} 
			0 & \text{ if } \frac{\beta_n}{\gamma_n}<\alpha^*, \\
			e^{\tau \gamma_n} I^0_n, & \text{ if }\frac{\beta_n}{\gamma_n}=\alpha^*,  
		\end{cases}
	\end{equation*}
	where $\tau\geq 0$ is the unique solution of the equation 
	\begin{equation*}
		\sum_{n\in\overline{\mathbb{N}}} \gamma_n e^{\tau \gamma_n}\mathbbm{1}_{\alpha_n=\alpha^*}I_n^0 = \frac{\theta}{\alpha^*}\left(\mathcal R_0-1\right).
	\end{equation*}
	Then $(I^*_n)$ is a stationary distribution with 
	$\displaystyle	\inf_{n\in\overline{\mathbb{N}}} \frac{I^0_n}{I^*_n}\geq e^{-|\tau|\gamma_0} $,
	therefore Proposition \ref{prop:Lyapunov} implies that $V\big(S(t), (I_n(t))_n\big)$  defined by \eqref{eq:Liapunov} is well-defined along the orbit and decreasing.

	We claim that $V\big(S(t), (I_n(t))_n\big)$ is constant along the orbit. Indeed, 
	let $t_k\to-\infty$ be an aribtrary sequence; since $S(t)$ and $\eta(t)$ are uniformly bounded, we extract from $t_k$ a subsequence, still denoted $t_k$, such that $S(t_k)\to S^{-\infty}$ and $\eta(t_k)\to \eta^{-\infty}$. Then, we have
	\begin{equation*}
		I_n(t_k) \xrightarrow[k\to+\infty]{} I^{-\infty}_n:= I_n^0 e^{\gamma_n \eta^{-\infty}}, \text{ in } \mathcal{M}_+(\overline{\mathbb{N}}) \text{ and } \ell^\infty((I^*_n)^{-1}).
	\end{equation*}
	Moreover $t\mapsto V(S(t), (I_n(t)))$ is decreasing and bounded so there exists $V^\infty$ such that 
	\begin{equation*}
		\lim_{t\to-\infty}V(S(t), (I_n(t))) = V^\infty.
	\end{equation*}
	The shifted orbits $(S(t+t_k), (I_n(t+t_k))_n)$ converge, up to a further extraction, to a complete orbit $(S^{-\infty}(t), (I^{-\infty}_n(t))_n)$; we have $ V(S^{-\infty}(t), (I^{-\infty}_n(t))_n)\equiv V^{-\infty}$ for all $t\in\mathbb{R}$, therefore $V'(S^{-\infty}(t), (I^{-\infty}_n(t))_n)=0$ and, by \eqref{eq:Lyapunov-decreasing}, 
	\begin{equation*}
		S^{-\infty}(t)\equiv \frac{1}{\alpha^*}, \qquad (S^{-\infty})'(t)\equiv 0, \qquad \text{ for all } t\in\mathbb{R},
	\end{equation*}
	and by the second line in \eqref{eq:main-a}, 
	\begin{equation}\label{eq:lemconvergence-uniquenesstau}
		\sum_{n\in\overline{\mathbb{N}}} I^{-\infty}(t) = \Lambda - \frac{\theta}{\alpha^*} \qquad \Longleftrightarrow \qquad \sum_{n\in\overline{\mathbb{N}}} \gamma_n e^{\gamma_n\eta^{-\infty}}I^0_n =\frac{\theta}{\alpha^*}(\mathcal{R}_0-1).
	\end{equation}
	Since $\eta\mapsto \sum_{n\in\overline{\mathbb{N}}} \gamma_n e^{\gamma_n\eta}I^0_n$ is strictly increasing, 
	the equation \eqref{eq:lemconvergence-uniquenesstau} has a unique solution which is $\eta=\tau$; therefore $\eta^{-\infty}=\tau$, $I^{-\infty}_n\equiv I^*_n$ and 
	\begin{equation*}
		V(S^{-\infty}, (I^{-\infty}_n)_{n\in\overline{\mathbb{N}}}) = \min_{(S, (I_n))\in D(V)} V(S, (I_n)) =0. 
	\end{equation*}
	Back to the original complete orbit, by the continuity of $V$ in $\mathbb{R}\times \ell^\infty((I^*_n)^{-1})$ we have that 
	\begin{equation*}
		\lim_{k\to+\infty} V(S(t_k), (I_n(t_k))_n) = V(S^{-\infty}, (I^{-\infty}_n)_{n\in\overline{\mathbb{N}}}) = \min_{(S, (I_n))\in D(V)} V(S, (I_n)), 
	\end{equation*}
	and since $V(S(t), (I_n(t))_n)$ is decreasing this means that 
	\begin{equation*}
		S(t)\equiv \frac{1}{\alpha^*}.
	\end{equation*}
	   Lemma \ref{lem:convergence-S} is proved.
\end{proof}
We are now in the position to prove Theorem \ref{thm:discrete}.
\begin{proof}[Proof of Theorem \ref{thm:discrete}]
	Recall that $(I_n)_{n\in\mathbb{N}}\in \ell^1_+$ is given as a sequence over $\mathbb{N}$; without loss of generality, we set 
	\begin{equation*}
		\omega(\alpha)\times\omega(\gamma)=:\{(\alpha_{-i}, \gamma_{-i}) \,:\, i=1, \ldots, K\},
	\end{equation*}
	as described as the beginning of Section \ref{sec:discrete}, and $I^0_{-i}=0$, so that $(I_n)_{n\in\overline{\mathbb{N}}}$ is well-defined as a member of $\mathcal{M}_+(\overline{\mathbb{N}})$ and the new system is strictly equivalent to the original system for all $t\geq 0$. \medskip

	We start by dealing with case \ref{item:persistence-convergence} and assume that there is $n_0\in\mathbb{N}$ such that $\frac{\beta_{n_0}}{\gamma_{n_0}}=\alpha^*$.
	Thanks to Lemma \ref{lem:strong-compactness} we know that 
	\begin{equation*}
		I_n(t) = I^0_n e^{\gamma_n\eta(t; 0)-(\alpha_n-\alpha^*) \overline{S}(t; 0)} 
	\end{equation*}
	and $\eta(t; 0)$ is uniformly bounded for $t\in [0, +\infty)$.
	Let $t_k$ be a sequence such that $t_k\to+\infty$ and $\eta(t_k; 0)\to \eta^{\infty}\in [0, \overline{\eta}]$. Then
	\begin{equation*}
		I_n(t_k) = I_n^0e^{\gamma_n t_k(\alpha^*\overline{S}(t_k)-1)-\gamma_n(\alpha^*-\alpha_n)t_k\overline{S}(t_k)} \Longrightarrow 
		I^\infty_n(0) = \begin{cases}
			0 & \text{ if } \alpha_n<\alpha^*, \\ 
			I_n^0e^{\gamma_n\eta^{\infty}} & \text{ if } \alpha_n=\alpha^*. 
		\end{cases}
	\end{equation*}
	By Proposition \ref{prop:compactness}, we can extract from $(t_k)$ a subsequence, still denoted $(t_k)$, such that $(S(t+t_k), I(t+t_k))$ converges in $\mathbb{R}\times \mathcal{M}_+^\ast(\overline{\mathbb{N}})$ to a complete orbit $(S^\infty(t), (I_n^\infty(t)))$ with $S^\infty>0$, and by Proposition \ref{prop:unifpers-nomut} the limit $I_n^\infty$ is non-trivial. By Proposition \ref{prop:compactness}, we have that $I^\infty_n(t)\equiv 0$ whenever $\alpha_n<\alpha^*$; thus $\beta_n=\alpha^*\gamma_n$ whenever $I^\infty_n>0$. Hence we can apply Lemma  \ref{lem:convergence-S} to show that 
	\begin{equation*}
		S^\infty(t) \equiv \frac{1}{\alpha^*} \text{ and  } (S^\infty)'(t)\equiv 0.
	\end{equation*}
	Thus 
	\begin{equation*}
		0=\Lambda-\theta S^\infty - S^\infty\sum_{n\in\overline{\mathbb{N}}}\beta_nI^\infty(0) \Longrightarrow \sum_{n\in\overline{\mathbb{N}}}\beta_n I^\infty(0) =\theta\left(\frac{\Lambda}{\theta}\alpha^*-1\right)
	\end{equation*}	
	and finally
	\begin{equation*}
		\sum_{n\in\overline{\mathbb{N}}}\gamma_n \mathbbm{1}_{\alpha_n=\alpha^*} e^{\gamma_n\eta^{\infty}}I_n^0 = \frac{\theta}{\alpha^*}\big(\mathcal{R}_0-1).
	\end{equation*}
	This equation has a unique solution $\tau\geq 0$, since the left-hand side is a strictly increasing function of $\eta^{\infty}$. Thus we have proved that $\eta(t; 0)$ converges to this value $\tau\geq 0$ and, finally, 
	\begin{equation*}
		\lim_{t\to+\infty} I_n(t) = \lim_{t\to+\infty} I_n^0e^{\gamma_n t(\alpha^*\overline{S}(t)-1)-\gamma_n(\alpha^*-\alpha_n)t\overline{S}(t)}  
		= \begin{cases}
			0 & \text{ if } \alpha_n<\alpha^*, \\ 
			I_n^0e^{\gamma_n\tau} & \text{ if } \alpha_n=\alpha^*. 
		\end{cases}
	\end{equation*}
	This finishes the proof of Theorem \ref{thm:discrete} case \ref{item:persistence-convergence}.\medskip


	Next we deal with case \ref{item:persistence-divergence} and assume that for all $n\in\mathbb{N}$ we have $\frac{\beta_n}{\gamma_n}<\alpha^*$. Let $t_k\to+\infty$ be arbitrary. By Proposition \ref{prop:compactness} we can extract from $t_k$ a subsequence, still denoted $t_k$, such that the shifted orbits $\big(S(t+t_k), (I_n(t+t_k))_{n\in\overline{\mathbb{N}}}\big)$ converge pointwise in $\mathbb{R}\times\mathcal{M}^\ast_+(\overline{\mathbb{N}})$  to a complete orbit $\big(S^\infty(t), (I^\infty_n(t))_{n\in\overline{\mathbb{N}}}\big)$ and the limit $(I^\infty_n)_{n\in\overline{\mathbb{N}}}$ is non-trivial by Proposition \ref{prop:unifpers-nomut}. Because of our assumption that $\frac{\beta_n}{\gamma_n}<\alpha^*$ we have $I^\infty_n=0$ for all $n\geq 0$ therefore $\frac{\beta_n}{\gamma_n}=\alpha^*$ whenever $I^\infty_n(0)>0$.
	Hence we can apply Lemma \ref{lem:convergence-S} and get 
	\begin{equation*}
		S^\infty(t) \equiv \frac{1}{\alpha^*}, \qquad (S^\infty)'(t)\equiv 0, \qquad \text{ for all } t\in\mathbb{R}.
	\end{equation*}
	Because the original sequence is arbitrary, we have proved that 
	\begin{equation*}
		S(t) \xrightarrow[t\to+\infty]{} \frac{1}{\alpha^*}, \qquad S'(t)\xrightarrow[t\to+\infty]{} 0, \qquad I_n(t) \xrightarrow[t\to+\infty]{} 0 \text{ for all }n\in\mathbb{N}, 
	\end{equation*}
	the fact that the mass does not vanish is a consequence of Proposition \ref{prop:unifpers-nomut} and the concentration property is a consequence of \eqref{eq:concentration} in  Proposition \ref{prop:compactness}. This finishes the proof of item \ref{item:persistence-divergence} and ends the proof of Theorem \ref{thm:discrete}.
\end{proof}

}

\subsection{Proof of Proposition \ref{prop:convergence-mass}}
{
Now under the additional assumption \ref{as:disintegration}, we prove Proposition \ref{prop:convergence-mass}.
\begin{proof}[Proof of Proposition \ref{prop:convergence-mass}]
	We decompose the proof in several steps. We let $\overline{S}(t):=\frac{1}{t}\int_0^t S(s)\dd s$. Our method is the following: we fix $\varepsilon>0$ and show that the set of indices $\mathcal{C}_\varepsilon:=\{n\,:\, \alpha_n\geq \alpha^*-\varepsilon\text{ and } \gamma_n\geq \gamma^*-\varepsilon\}$ concentrates asymptotically all the mass. \medskip

	\noindent\textbf{Step 1:}  We show that 
	\begin{equation}\label{eq:disintegration-vanish1}
		\sum_{n=0}^{+\infty} \mathbbm{1}_{\alpha_nS(t)< 1}I_n(t)\xrightarrow[t\to+\infty]{}0.
	\end{equation}

	Indeed since $\overline{S}(t)\to\frac{1}{\alpha^*}$,  the function $ \mathbbm{1}_{\alpha_nS(t)< 1}I_n(t) $ converges pointwise to 0 as $t\to+\infty$ because of our assumption that $\frac{\beta_n}{\gamma_n}=\alpha_n<\alpha^*$ for all $n$ (in fact, this function is asymptotically stationary equal to 0 for all fixed $n\in\mathbb{N}$). Moreover  we have $ \mathbbm{1}_{\alpha_nS(t)< 1}I_n(t)=\mathbbm{1}_{\alpha_nS(t)< 1} e^{t\gamma_n(\alpha^*\overline{S}(t)-1)}I^0_n\leq I^0_n$, so the sequence is uniformly dominated by $\big(I^0_n\big)_{n\in\mathbb{N}}\in\ell^1_+$. By the Lebesgue dominated convergence Theorem, we have therefore
	\begin{equation*}
		\lim_{n\to+\infty}\sum_{n=0}^{+\infty} \mathbbm{1}_{\alpha_nS(t)< 1}I_n(t)=0, 
	\end{equation*}
	which proves \eqref{eq:disintegration-vanish1}.
	\medskip

	\noindent\textbf{Step 2:} We show that  for any $\varepsilon>0$,
	\begin{equation}\label{eq:disintegration-bounded-measure}
		\limsup_{t\to +\infty}\sum_{n=0}^{+\infty}\mathbbm{1}_{\alpha_n\overline{S}(t)\geq 1} e^{t(\gamma^*-\varepsilon)\left(\alpha_n\overline{S}(t)-1\right)}I^0_n <+\infty.
	\end{equation}

	Indeed let $(A_n)_{n\in\mathbb{N}}$ be a strictly increasing  enumeration of the set $\{\alpha_n\,:\, n\in\mathbb{N}\}$. We remark that, by changing the order of summation, we have
	\begin{align*}
		\sum_{n=0}^{+\infty}\mathbbm{1}_{\alpha_n\overline{S}(t)\geq 1}I_n(t) &= \sum_{n=0}^{+\infty} \sum_{\{k\,:\, \alpha_k=A_n\}} e^{\gamma_n t \left(A_n\overline{S}(t)-1\right)} I^0_k, 
	\end{align*}
	so that, according to Step 1, we have for $t$ sufficiently large
	\begin{align*}
		\sum_{\{n\,:\, \gamma_n\geq \gamma^*-\varepsilon\}} I_n(t)&= \sum_{\{n\,:\, \gamma_n\geq \gamma^*-\varepsilon\}} \left(\mathbbm{1}_{\alpha_n\overline{S}(t)<1}+\mathbbm{1}_{\alpha_n\overline{S}(t)\geq 1}\right)I_n(t)  \\ 
		&= o(1)+ \sum_{\{n\,:\, A_n\geq \frac{1}{\overline{S}(t)}\}} \sum_{\{k\,:\, \alpha_k=A_n\text{ and } \gamma_k\geq\gamma^*-\varepsilon\}} e^{\gamma_k t \left(A_n\overline{S}(t)-1\right)} I^0_k  \\ 
		&\geq o(1)+\sum_{\{n\,:\, A_n\geq \frac{1}{\overline{S}(t)}\}} \sum_{\{k\,:\, \alpha_k=A_n\text{ and } \gamma_k\geq\gamma^*-\varepsilon\}} e^{(\gamma^*-\varepsilon) t \left(A_n\overline{S}(t)-1\right)} I^0_k\\
		&= o(1)+\sum_{\{n\,:\, A_n\geq \frac{1}{\overline{S}(t)}\}} e^{(\gamma^*-\varepsilon) t \left(A_n\overline{S}(t)-1\right)} \sum_{\{k\,:\, \alpha_k=A_n\text{ and } \gamma_k\geq\gamma^*-\varepsilon\}} I^0_k \\
		&= o(1)+\sum_{\{n\,:\, A_n\geq \frac{1}{\overline{S}(t)}\}} e^{(\gamma^*-\varepsilon) t \left(A_n\overline{S}(t)-1\right)}\dfrac{\displaystyle\sum_{\{k\,:\, \alpha_k=A_n\text{ and } \gamma_k\geq\gamma^*-\varepsilon\}} I^0_k}{\displaystyle\sum_{\{k\,:\, \alpha_k=A_n\}} I^0_k} \sum_{\{k\,:\, \alpha_k=A_n\} }I^0_k  \\ 
		&\geq o(1)+m\sum_{\{n\,:\, A_n\geq \frac{1}{\overline{S}(t)}\}} \sum_{\{\alpha_k=A_n\}}e^{(\gamma^*-\varepsilon) t \left(A_n\overline{S}(t)-1\right)}I^0_k,
	\end{align*}
	wherein $m>0$ is the constant provided by Assumption \ref{as:disintegration} and the error term $o(1)$ collects terms going to 0 as $t\to+\infty$. Thus 
	\begin{align*}
		\sum_{\{n\,:\, \alpha_n\geq \alpha^*-\varepsilon\}}\mathbbm{1}_{\alpha_n\overline{S}(t)\geq 1} e^{t(\gamma^*-\varepsilon)\left(\alpha_n\overline{S}(t)-1\right)} I^0_n & = \sum_{\{n\,:\, A_n\geq \frac{1}{\overline{S}(t)}\}} \sum_{\{\alpha_k=A_n\}}e^{(\gamma^*-\varepsilon) t \left(A_n\overline{S}(t)-1\right)}I^0_k \\
		&\leq o(1)+\frac{1}{m}\sum_{\{n\,:\, \gamma_n\geq\gamma^*-\varepsilon\}}I_n(t) = \mathcal{O}(1).
	\end{align*}
	This proves \eqref{eq:disintegration-bounded-measure}.
	\medskip

	\noindent\textbf{Step 3:} We show that, for all $\varepsilon>0$, 
	\begin{equation}\label{eq:disintegration-complement-vanishes}
		\sum_{\{n\,:\, \alpha_n\geq \alpha^*-\varepsilon \text{ and } \gamma_n\leq\gamma^*-\varepsilon\}}I_n(t) \xrightarrow[t\to+\infty]{}0. 
	\end{equation}
	
	Indeed, we have for $t$ sufficiently large
	\begin{align*}
		\smashoperator[r]{\sum_{\{n\,:\, \alpha_n\geq \alpha^*-\varepsilon \text{ and } \gamma_n\leq\gamma^*-\varepsilon\}}}\mathbbm{1}_{\alpha_n\overline{S}(t)\geq 1}I_n(t) 
		&=\sum_{n=0}^{+\infty} \mathbbm{1}_{A_n\overline{S}(t)\geq 1}\sum_{\{\alpha_k=A_n\}}e^{\gamma_n t \left(A_n\overline{S}(t)-1\right)}\mathbbm{1}_{\gamma_k\leq\gamma^*-\varepsilon}I^0_k \\ 
		&\leq \sum_{n=0}^{+\infty} \mathbbm{1}_{A_n\overline{S}(t)\geq 1}\sum_{\{\alpha_k=A_n\}}e^{ t (\gamma^*-\varepsilon)\left(A_n\overline{S}(t)-1\right)}\mathbbm{1}_{\gamma_k\leq\gamma^*-\varepsilon}I^0_k\\ 
		&\leq\sum_{n=0}^{+\infty} \mathbbm{1}_{A_n\overline{S}(t)\geq 1} e^{-t\frac{\varepsilon}{2}(A_n\overline{S}(t)-1)}\left(\sum_{\{\alpha_k=A_n\}}I^0_k\right)e^{t \left(\gamma^*-\frac{\varepsilon}{2}\right)\left(A_n\overline{S}(t)-1\right)}  \\
		&=\sum_{n=0}^{+\infty}  e^{-t\frac{\varepsilon}{2}(\alpha_n\overline{S}(t)-1)}\mathbbm{1}_{\alpha_n\overline{S}(t)\geq 1}e^{t \left(\gamma^*-\frac{\varepsilon}{2}\right)\left(\alpha_n\overline{S}(t)-1\right)} \mathbbm{1}_{\gamma_n\leq\gamma^*-\varepsilon}I^0_n. 
	\end{align*}
	Reducing $\varepsilon$ if necessary we may assume that $\frac{2\gamma^*-\varepsilon}{\varepsilon}>1$.
	Hence by Hölder's inequality we have:
	\begin{align*}
		\smashoperator[r]{\sum_{\{n\,:\, \alpha_n\geq \alpha^*-\varepsilon \text{ and } \gamma_n\leq\gamma^*-\varepsilon\}}}\mathbbm{1}_{\alpha_n\overline{S}(t)\geq 1}I_n(t) &\leq \left(\sum_{n=0}^{+\infty} \left(e^{-t\frac{\varepsilon}{2}(\alpha_n\overline{S}(t)-1)}\right)^{\frac{\gamma^*-\frac{\varepsilon}{2}}{\frac{\varepsilon}{2}}}\mathbbm{1}_{\alpha_n\overline{S}(t)\geq 1}e^{ t \left(\gamma^*-\frac{\varepsilon}{2}\right)\left(\alpha_n\overline{S}(t)-1\right)}I^0_n\right)^{\frac{\frac{\varepsilon}{2}}{\gamma^*-\frac{\varepsilon}{2}}}\\ 
		&\quad \times \left(\sum_{n=0}^{+\infty}  \mathbbm{1}_{\alpha_n\overline{S}(t)\geq 1}e^{t \left(\gamma^*-\frac{\varepsilon}{2}\right)\left(\alpha_n\overline{S}(t)-1\right)} I^0_n \right)^{\frac{\gamma^*-\varepsilon}{\gamma^*-\frac{\varepsilon}{2}}} \\ 
		&\leq \left(\sum_{n=0}^{+\infty} \mathbbm{1}_{\alpha_n\overline{S}(t)\geq 1}I^0_n\right)^{\frac{\varepsilon}{2\gamma^*-\varepsilon}}\left(\sum_{n=0}^{+\infty}  \mathbbm{1}_{\alpha_n\overline{S}(t)\geq 1}e^{t \left(\gamma^*-\frac{\varepsilon}{2}\right)\left(\alpha_n\overline{S}(t)-1\right)} I^0_n \right)^{\frac{2(\gamma^*-\varepsilon)}{2\gamma^*-\varepsilon}}. 
	\end{align*}
	Since $  \mathbbm{1}_{\alpha_n\overline{S}(t)\geq 1}I^0_n $ is uniformly bounded by a summable sequence and converges to 0 for each $n\in\mathbb{N}$, we have by Lebesgue's dominated convergence Theorem:
	\begin{equation*}
		\sum_{n=0}^{+\infty} \mathbbm{1}_{\alpha_n\overline{S}(t)\geq 1}I^0_n \xrightarrow[t\to+\infty]{}0, 
	\end{equation*}
	and in Step 2 (with $\varepsilon$ replaced by $\frac{\varepsilon}{2}$) we proved that 
	\begin{equation*}
		\limsup_{t\to+\infty}\sum_{n=0}^{+\infty}  \mathbbm{1}_{\alpha_n\overline{S}(t)\geq 1}e^{t \left(\gamma^*-\frac{\varepsilon}{2}\right)\left(\alpha_n\overline{S}(t)-1\right)} I^0_n <+\infty, 
	\end{equation*}
	hence 
	\begin{equation*}
		\smashoperator[r]{\sum_{\{n\,:\, \alpha_n\geq \alpha^*-\varepsilon \text{ and } \gamma_n\leq\gamma^*-\varepsilon\}}}\mathbbm{1}_{\alpha_n\overline{S}(t)\geq 1}I_n(t) \xrightarrow[t\to+\infty]{}0. 
	\end{equation*}
	We have shown \eqref{eq:disintegration-complement-vanishes}.
	\medskip

	\noindent\textbf{Step 4:} We show  the convergence of the mass. To do so, we consider the equivalent system set on $\overline{\mathbb{N}}$ by setting $I^0_{-i}=0$ for $i\in\{-K, \ldots, -1\}$, as in the beginning of Section \ref{sec:discrete}. Let $(t_k)$ be any sequence such that $t_k\to+\infty$.  Thanks to Proposition \ref{prop:compactness} we can extract a subsequence, still denoted $t_k$, such that $(S(t_k+t), (I_n(t_k+t))_{n\in\overline{\mathbb{N}}})$ converges to a complete orbit $(S^\infty(t), (I^\infty_n(t))_{n\in\overline{\mathbb{N}}})$. Thanks to Theorem \ref{thm:discrete} we know that $I^\infty_n(t)\equiv 0$ whenever $\alpha_n<\alpha^*$ and thanks to Step 3 we know that $I^\infty_n(t)\equiv 0$ whenever $\gamma_n<\gamma^*$. By Lemma \ref{lem:convergence-S} we have then $S^\infty(t)\equiv \frac{1}{\alpha^*}$ hence 
	\begin{equation*}
		\frac{\dd}{\dd t} S^\infty(t) = 0 = \Lambda - \theta S^\infty(t)  - S^\infty(t) \sum_{n=0}^{+\infty} \gamma_n\alpha_n I^\infty_n(t) \Longleftrightarrow \frac{1}{\alpha^*}\sum_{n=-K}^{+\infty} \gamma^*\alpha^*I^\infty_n(t) = \Lambda - \frac{\theta}{\alpha^*}, 
	\end{equation*}
	from which we deduce
	\begin{equation*}
		\sum_{n=-K}^{+\infty}I^\infty_n(t) = \frac{\theta}{\alpha^*\gamma^*}\left(\mathcal{R}_0 - 1\right).
	\end{equation*}
	Moreover for sequence of indices $(n_j)$ satisfying \eqref{eq:propconvergence-suboptimal-trait}, the omega-limit set of the sequence $\big((\alpha_{n_j}, \gamma_{n_j})\big)_{j\in\mathbb{N}}$ is the set $\{(\alpha_{-i_1}, \gamma_{-i_1}), \ldots, (\alpha_{-i_N}, \gamma_{-i_N})\}$ for some $N<K$ and with $\gamma_{i_j}<\gamma^*$, $i=1, \ldots, N$. Thanks to Step 3, we have then
	\begin{equation*}
	    \sum_{j=0}^{+\infty} I_{n_j}(t_k+t)\xrightarrow[k\to+\infty]{} \sum_{j=1}^{N}I^\infty_{-i_j}(t)+ \sum_{j=0}^{+\infty} I^\infty_{n_j}(t)\equiv 0 \text{ for all }t\in\mathbb{R}. 
	\end{equation*}
	Since the sequence $(t_k)$ is arbitrary, we have indeed proved that $\sum_{n=0}^{+\infty} I_n(t)$ converges and that the limit is given by \eqref{eq:limit-mass-discrete}, and that for any sequence of indices $(n_k)$ satisfying \eqref{eq:propconvergence-suboptimal-trait}, \eqref{eq:propconvergence-suboptimal-vanishes} holds. Proposition \ref{prop:convergence-mass} is proved. 
\end{proof}
}

\subsection{Extinction in the case $\mathcal R_0=1$}
\label{sec:R0=1}

We continue the proof of Proposition \ref{prop:Extinction} in the limit case $\mathcal R_0=1$ and we shall show that the infection dies out in this situation.
\begin{proposition}
Let Assumption \ref{as:params} and \ref{as:omega} be satisfied.
Then for all initial data $(S_0, (I_n^0)_{n\in\mathbb N})$, the solution $(S(t), (I_n(t))_{n\in\mathbb N})$ satisfies
$$
\lim_{t\to\infty} S(t)=\frac{\Lambda}{\theta}\text{ and }\lim_{t\to\infty} \sum_{n=0}^\infty I_n(t)=0.
$$
\end{proposition}
\begin{proof}
To prove this proposition, as in Subsection \ref{sec:discrete}, we extend the system with $n\in\overline{\mathbb N}=\{-K,-K+1,-K+2, \cdots \}$.
Fix an arbitrary sequence $(t_k)$ with $t_k\to\infty$ as $k\to\infty$.
Then, up to extraction, the sequence of functions $(S(t+t_k), (I_n(t+t_k))_{n\in \overline{\mathbb N}})$ converges to a complete orbit $(S^\infty(t), (I_n^\infty(t))_{n\in \overline{\mathbb N}})$ locally uniformly in $t$ with values in $\R\times\mathcal M(\overline{\mathbb N})$.
Now recall that due to Lemma \ref{lem:bounds-nomut} one has
$$
S^\infty(t)\leq \frac{\Lambda}{\theta},\;\forall t\in\R,
$$
and $t\mapsto \left(S^\infty(t), (I_n^\infty(t))_{n\in \overline{\mathbb N}}\right)$ is solution of
\begin{equation*}
		\left\{\begin{aligned}\relax
			&\frac{\dd \phantom{t}}{\dd t}S^\infty(t) = \Lambda -\theta S^\infty(t) - S^\infty(t)\sum_{i=-K}^\infty \beta_i I_i^\infty(t)\\ 
			&\frac{\dd \phantom{t}}{\dd t}I_n^\infty(t)=\gamma_n\left(\alpha_n S^\infty(t)- 1\right)I_n^\infty(t), \;\forall n\geq -K \\ 
		\end{aligned}\right.
	\end{equation*}
Now to complete the proof of the proposition, let us show that
$$
\sum_{n\geq -K} I_n^\infty(t)=0,\;\forall t\in\R.
$$
To that aim consider the quantity
$$
\ell =\sup_{t\in\R} \sum_{n\geq -K} I_n^\infty(t)
$$
and assume by contradiction that $\ell>0$. Next consider a sequence $(s_p)\subset \R$ such that 
 $$
 \ell =\lim_{p\to\infty}\sum_{n\geq-K} I_n^\infty(s_p).
 $$
As above consider the sequence of functions $(S^\infty(t+s_p), (I_n^\infty(t+s_p))_{n\in \overline{\mathbb N}})$ and assume that, possibly along a sub-sequence, it converges to a complete orbit $(\hat S(t), (\hat I_n(t))_{n\in \overline{\mathbb N}})$ locally uniformly in $t\in \R$ with values in $\R\times\mathcal M(\overline{\mathbb N})$. Hence it becomes a solution for $t\in\R$ of
\begin{equation*}
		\left\{\begin{aligned}\relax
			&\frac{\dd \phantom{t}}{\dd t}\hat S(t) = \Lambda -\theta \hat S(t) - \hat S(t)\sum_{i=-K}^\infty \beta_i \hat I_i(t)\\ 
			&\frac{\dd \phantom{t}}{\dd t}\hat I_n(t)=\gamma_n\left(\alpha_n \hat S(t)- 1\right)\hat I_n(t), \;\forall n\geq -K \\ 
		\end{aligned}\right.
	\end{equation*}
	together with
$$
\ell =\sum_{n\geq -K} \hat I_n(0)=\sup_{t\in \R}\sum_{n\geq -K} \hat I_n(t).
$$	
	
Now observe that for all $n\geq -K$ one has
$$
\frac{\dd \phantom{t}}{\dd t}\hat I_n(t)\leq \gamma_n\left(\alpha_n \frac{\Lambda}{\theta}- 1\right)\hat I_n(t),\;\forall t\in\R.
$$
Hence, since $\mathcal R_0=1$ we have $\frac{\Lambda}{\theta}=\frac{1}{\alpha^*}$ and
$\hat I_n(t)\equiv 0$ for all $n\geq-K$ such that $\alpha_n<\alpha^*$. 
Now set $J=\{n\geq -K:\;\alpha_n=\alpha^*\}$ and the above system of equations reduces to
\begin{equation*}
		\left\{\begin{aligned}\relax
			&\frac{\dd \phantom{t}}{\dd t}\hat S(t) = \Lambda -\theta \hat S(t) - \hat S(t)\sum_{j\in J}\gamma_j \alpha^* \hat I_j(t)\\ 
			&\frac{\dd \phantom{t}}{\dd t}\hat I_n(t)=\gamma_n\left(\alpha^* \hat S(t)- 1\right)\hat I_n(t), \;\forall n\in J.\\ 
		\end{aligned}\right.
	\end{equation*}
Now summing-up the $\hat I_n-$components, we get
$$
\sum_{n\in J} \frac{\dd \phantom{t}}{\dd t}\hat I_n(0)=0=(\alpha^*\hat S(0)-1)\sum_{n\in J}\gamma_n \hat I_n(0),
$$
so that, since $\hat S(t)\leq \frac{\Lambda}{\theta}$ for all $t$, we obtain:
$$
\hat S(0)= \frac{1}{\alpha^*}=\frac{\Lambda}{\theta}\text{ and }\frac{\dd \phantom{t}}{\dd t}\hat S(0)=0.
$$
Substituting $t=0$ into the $\hat S-$equation yields 
$$
0=\dfrac{\dd \phantom{t}}{\dd t}\hat S(0) = \Lambda -\theta \frac{\Lambda}{\theta} - \frac{\Lambda}{\theta}\sum_{n\in J}\gamma_n \alpha^* I_n^\infty(0).
$$
This ensures that
$$
\alpha^*\left(\inf_{n\geq 0} \gamma_n\right)\ell\leq \sum_{n\in J}\gamma_n \alpha^* I_n^\infty(0)=0,
$$
a contradiction that proves that $\ell=0$ and completes the proof of the proposition.
\end{proof}
\printbibliography
\end{document}